\documentclass{fundam}

\usepackage{latexsym,amsmath}
\usepackage{refcount}
\usepackage{url}
\usepackage{graphicx}
\usepackage{graphicx,xcolor}

  \usepackage{hyperref}

\begin{document}

\setcounter{page}{307}
\publyear{2021}
\papernumber{2101}
\volume{184}
\issue{4}

   \finalVersionForARXIV

\title{A Polynomial-Time Construction of a Hitting Set\\
for Read-Once Branching Programs of Width 3}

\author{Ji\v{r}\'{\i} \v{S}\'{\i}ma\thanks{Address for correspondence:
           Institute of Computer Science of  the Czech Academy of Sciences, P.O.~Box~5,
           182~07~Prague~8, Czech Republic.\newline \newline
          \vspace*{-6mm}{\scriptsize{Received January 2021; \ accepted January 2022.}}}, Stanislav \v{Z}\'{a}k
\\
Institute of Computer Science of the Czech Academy of Sciences\\
Prague, Czech Republic\\
sima@cs.cas.cz,  stan@cs.cas.cz
}

\maketitle

\runninghead{J. \v{S}\'{\i}ma and  S. \v{Z}\'{a}k}{A Hitting Set for Width-3 1-Branching Programs}

\begin{abstract}
Recently, an interest in constructing pseudorandom or hitting set
generators for restricted branching programs has increased, which is
motivated by the fundamental issue of derandomizing space-bounded
computations. Such constructions have been known only in the case of
width~2 and in very restricted cases of bounded width. In this paper,
we characterize the hitting sets for read-once branching programs of
width~3 by a so-called richness condition. Namely, we
show that such sets hit the class of read-once conjunctions of DNF
and CNF (i.e.\ the weak richness). Moreover, we prove that
any rich set extended with all strings within Hamming distance of~3
is a hitting set for read-once branching programs of width~3.
Then, we show that any almost $O(\log n)$-wise independent set
satisfies the richness condition. By using such a set due to Alon
et al. (1992) our result provides an explicit polynomial-time
construction of a hitting set for read-once branching programs of
width~3 with acceptance probability $\varepsilon>5/6$. We announced this result at conferences more than ten years ago, including only proof sketches, which motivated a number of subsequent results on pseudorandom generators for restricted read-once branching programs. This paper contains our original detailed proof that has not been published yet.
\end{abstract}


\section{Introduction}

An \emph{$\varepsilon$-hitting set} for a class of Boolean functions of $n$
variables is a set $H\subseteq\{0,1\}^n$ such that for every function
$f$ in the class, the following is satisfied: If a random input is
accepted by $f$ with probability at least $\varepsilon$, then there is
also an input in $H$ that is accepted by $f$. An efficiently
constructible sequence of hitting sets for increasing $n$ is a
straightforward generalization of the \emph{hitting set generator}
introduced in \cite{Goldreich99}, which is a weaker (one-sided error)
version of a pseudorandom generator~\cite{NisanW94}. Recall that an
\emph{$\varepsilon$-pseudorandom generator} for a class of Boolean
functions of $n$ variables is a function
\linebreak
$\mathbf{g}:\{0,1\}^s\longrightarrow\{0,1\}^n$ which stretches a short
uniformly random \emph{seed} of \emph{length} $s$ bits into $n$ bits
($s\ll n$) that cannot be distinguished from uniform ones. In particular,
for every function $f$ in the class, condition
\mbox{$\left|Pr_{\mathbf{x}\sim U_n}\left[f(\mathbf{x})=1\right]-
Pr_{\mathbf{y}\sim U_s}\left[f(\mathbf{g}(\mathbf{y}))=1\right]\right|
\leq\varepsilon$} holds where $\mathbf{x}\sim U_n$ means that
$\mathbf{x}$ is uniformly distributed in $\{0,1\}^n$.

For the class of Boolean functions of polynomial complexity in any
reasonable model, it is easy to prove the existence of
$\varepsilon$-hitting set of polynomial size, if $\varepsilon>1/n^c$
for a constant $c$ where $n$ is the number of variables. The proof is
nonconstructive, since it uses a counting argument.
An important problem in complexity theory is to find polynomial-time
constructible hitting sets for functions of polynomial complexity in
different standard models like circuits, formulas, branching programs etc.
Such constructions would have consequences for the relationship between
deterministic and probabilistic computations in the respective models.

Looking for polynomial-time constructions of hitting sets for
unrestricted models belongs to the hardest problems in computer
science. Hence, restricted models are investigated. We consider \emph{read-once branching (1-branch\-ing) programs}, which is a restricted model of space-bounded computations~\cite{Wegener00}. Recall a branching program (see Section~\ref{1bp3} for precise definitions) is used to compute a Boolean function which is represented as a directed acyclic multi-graph with a root (a~source). This graph consists of inner nodes labeled with input variables and terminal nodes (sinks) labeled with Boolean output values 0 or 1. Each inner node has out-degree 2, while its two outgoing edges are labeled with 0 and 1, respectively. The computational path starts at the source, always follows the edge outgoing from the inner node whose label agrees with an assignment of the value to the input variable associated with this node, and terminates in a sink providing the output. A 1-branching program queries every input variable at most once along each computational path.
In a leveled branching program, the edges connect only nodes in the consecutive levels where the level of a node is defined as its distance from the source. Then the width of such a program is the maximum number of nodes on any of its levels.

For read-once branching programs of polynomial size, pseudorandom generators with seed length $O(\log^2 n)$ have been known for a long time through the result of Ni\-san~\cite{Nisan92}. Note that an explicit pseudorandom generator for this model which is computable in logarithmic space and has seed length $O(\log n)$ would suffice to derandomize the complexity class BPL (Bounded-error Probabilistic Log\-a\-rith\-mic-space). Recently, considerable
attention has been paid to improving the seed length to $O(\log n)$ in
the constant-width case, which is a fundamental problem with many
applications in circuit lower bounds and derandomization~\cite{MekaZ10,Vadhan12}. The problem has been resolved for width~2 but the known techniques provably fail for
width~3~\cite{BogdanovDVY09,BrodyV10,De11,FeffermanSUV,MekaZ10,Vadhan12},
which applies even to hitting set generators~\cite{BrodyV10}.

\eject

In the case of width 3, we do not know of any significant improvement
over Nisan's result except for some recent progress in the severely
restricted case of so-called regular oblivious read-once branching
programs. Recall that an \emph{oblivious} branching program queries
the input variables in a fixed order, which represents a provably
weaker computational model~\cite{BeameM10}. For constant-width
\emph{regular} oblivious 1-branching programs which have the in-degree
of all nodes equal to 2 (or 0), three independent constructions of
$\varepsilon$-pseudorandom generators with seed length
$O(\log n(\log\log n+\log(1/\varepsilon)))$ were
achieved~\cite{BravermanRRY10,BrodyV10,De11}. This seed length has later
been improved to $O(\log n\log(1/\varepsilon))$ for constant-width
\emph{permutation} oblivious 1-branching programs~\cite{Pudlak10,De11}
which are regular programs with the two edges incoming to any node
labeled 0 and 1, i.e.\ edges labeled with 0 respectively
1 create a permutation for each level-to-level transition~\cite{MekaZ10}.

In the constant-width regular 1-branching programs the fraction of
inputs that are queried at any node is always lower-bounded by a positive
constant. This excludes the fundamental capability of general
(non-regular) branching programs to recognize the inputs that
contain a given substring on a \mbox{non-constant} number of selected
positions. In our approach, we manage the analysis also for this
essential case. In particular, we identify two types of convergence of
the number of inputs along a computational path towards zero which
implement read-once DNFs and CNFs, respectively. Thus, we achieve the
construction of a hitting set generator for general width-3 1-branching
programs which need not be regular nor oblivious. In our previous
work~\cite{Sima07}, we constructed a hitting set for so-called
\emph{simple} width-3 1-branching programs which exclude one specific
pattern of level-to-level transition in their normalized form and cover
the width-3 regular case.

In the present paper, we provide a polynomial-time construction of
a hitting set for read-once branching programs of width~3 with acceptance
probability $\varepsilon>5/6$, which need not be oblivious. This represents
an important step in the effort of constructing hitting set generators for
the model of read-once branching programs of bounded width. For this purpose,
we formulate a so-called \emph{richness} condition which is independent of
a rather technical definition of branching programs. In fact, the (full)
richness condition implies its weaker version which is equivalent
to the definition of hitting sets for read-once conjunctions of DNF and CNF.
Thus, a related line of study concerns pseudorandom generators for read-once
formulas, such as read-once DNFs~\cite{DeETT10}.

We show that the richness property characterizes in a certain sense the
hitting sets for width-3 1-branching programs. In particular, its weaker
version proves to be necessary for such hitting sets, while the sufficiency
of richness represents the main result of this paper. More precisely, we
show that any rich set extended with all strings within Hamming distance
of~3 is a hitting set for 1-branching programs of width~3 with the acceptance probability greater than $5/6$. The same result with a weakly rich set holds for the oblivious width-3 1-branching programs~\cite{Sima10}. The proof is based on a detailed analysis of structural properties of the width-3 1-branching programs that reject all the inputs from the candidate hitting set. Then, we prove that for a suitable constant $C$, any almost $(C\log n)$-wise independent set which can be constructed in polynomial time by the result due to Alon et al.~\cite{Alon92} satisfies the richness condition, which implies our result. In addition, it follows from the latter result that almost $O(\log n)$-wise independent sets are weakly rich and hence, they hit the class of read-once conjunctions of DNF and CNF which is
a generalization of the earlier result from~\cite{DeETT10}.

A preliminary version of this article appeared as extended
abstracts~\cite{Sima11a,Sima11b} including only proof sketches, where our result was formulated for acceptance probability $\varepsilon>11/12$. Since then a number of results on pseudorandom generators for restricted 1-branching programs~\cite{GopalanMRTV12,Steinke12,ForbesS13,ReingoldSV13,Steinberger13,Watson13,ForbesSS14,BazziN17,MurtaghRSV17,ServedioT17,SteinkeVW17,Ahmadinejad19,DoronHH19,MekaRT19,ServedioT19,BravermanCG20,ChengH20,HozaZ20} have been achieved which were motivated and/or follow our study referring to our result; see e.g.\ the paper~\cite{HozaZ20} for a current survey of the newest achievements along this direction. This paper contains our original complete proof that has not been published yet.

The paper is organized as follows. After a brief review of basic
definitions regarding branching programs in Section~\ref{1bp3}
(see \cite{Wegener00} for more information), the weak richness condition
is formulated and proved to be necessary in Section~\ref{necc}.
The richness condition and its sufficiency is presented in
Section~\ref{suffc} including the intuition behind the proof. The
subsequent four Sections~\ref{dfpart}--\ref{indend} are devoted to
the technical proof of this proposition. Furthermore,
our theorem that any almost $O(\log n)$-wise independent set is
rich is presented in Section~\ref{kwISrich} where also the main steps
of the technical proof occupying the subsequent four
Sections~\ref{pclmod}--\ref{Taylorth} are introduced. Finally, our
result is summarized in Section~\ref{concl}.

\section{Normalized width-$w$ 1-branching programs}
\label{1bp3}

A \emph{branching program} $P$ on the set of input Boolean variables
$X_n=\{x_1,\ldots,$ $x_n\}$ is a directed acyclic multi-graph $G=(V,E)$
that has one \emph{source} $s\in V$ of zero in-degree and, except for
\emph{sinks} of zero out-degree, all the \emph{inner} (non-sink) nodes
have out-degree~2. In addition, the inner nodes get labels from $X_n$
and the sinks get labels from $\{0,1\}$. For each inner node, one of
the outgoing edges gets the label~0 and the other one gets the
label~1. The branching program $P$ computes Boolean function
$P:\{0,1\}^n\longrightarrow\{0,1\}$ as follows. The computational
path of $P$ for an input $\mathbf{a}=(a_1,\ldots,a_n)\in \{0,1\}^n$
starts at source $s$. At any inner node labeled by $x_i\in X_n$, input
variable $x_i$ is tested and this path continues with the outgoing
edge labeled by $a_i$ to the next node, which is repeated until the
path reaches the sink whose label gives the output value $P(\mathbf{a})$.
Denote by \mbox{$P^{-1}(a)=\{\mathbf{a}\in\{0,1\}^n\,|\,P(\mathbf{a})=a\}$}
 the set of inputs for which $P$ outputs $a\in\{0,1\}$. For inputs of arbitrary
lengths, infinite families $\{P_n\}$ of branching programs, each $P_n$
for one input length $n\geq 1$, are used.

A branching program $P$ is called \emph{read-once} (or shortly
\emph{1-branching} program) if every input variable from $X_n$
is queried at most once along each computational path. Here we consider
\emph{leveled} branching programs in which each node belongs to a
level, and edges lead from level $k\geq 0$ only to the next
level $k+1$. We assume that the source of $P$ creates level 0,
whereas the last level is composed of all sinks. The number of levels
decreased by 1 equals the \emph{depth} of $P$ which is the length of
its longest path, and the maximum number of nodes on one level is
called the \emph{width} of $P$. In addition, $P$ is called
\emph{oblivious} if all nodes at each level are labeled with
the same variable.

For a 1-branching program $P$ of width $w$ define a $w\times w$
\emph{transition matrix} $T_k$ on level $k\geq 1$ such that
$t_{ij}^{(k)}\in\{0,\frac{1}{2},1\}$ is the half of the number of
edges leading from node $v_j^{(k-1)}$ ($1\leq j\leq w$) on level $k-1$
of $P$ to node $v_i^{(k)}$ ($1\leq i\leq w$) on level $k$. For
example, $t_{ij}^{(k)}=1$ implies there is a \emph{double edge} from
$v_j^{(k-1)}$ to $v_i^{(k)}$. Clearly, $\sum_{i=1}^w t_{ij}^{(k)}=1$
since this sum equals the half of the out-degree of inner node
$v_j^{(k-1)}$, and $2\cdot\sum_{j=1}^w t_{ij}^{(k)}$ is the in-degree
of node $v_i^{(k)}$. Denote by a column vector
$\mathbf{p}^{(k)}=(p_1^{(k)},\ldots,p_w^{(k)})^{\sf T}$
the \emph{distribution}
of inputs among $w$ nodes on level $k$ of $P$, that is, $p_i^{(k)}$
is the probability that a random input is tested at node $v_i^{(k)}$,
which equals the ratio of the number of inputs from
$M(v_i^{(k)})\subseteq\{0,1\}^n$ that are tested at $v_i^{(k)}$ to all
$2^n$ possible inputs. It follows $\bigcup_{i=1}^w M(v_i^{(k)})=\{0,1\}^n$
and $\sum_{i=1}^w p_i^{(k)}=1$ for every level $k\geq 0$. Given the
distribution $\mathbf{p}^{(k-1)}$ on level $k-1$, the distribution on the
subsequent level $k$ can be computed using the transition matrix $T_k$ as
\begin{equation}
\label{pkAkpk1}
\mathbf{p}^{(k)}=T_k\cdot\mathbf{p}^{(k-1)}\,.
\end{equation}
It is because the ratio of inputs coming to node $v_i^{(k)}$ from
previous-level nodes equals
$p_i^{(k)}=$
\linebreak
$\sum_{j=1}^w t_{ij}^{(k)}p_j^{(k-1)}$ since each of the two
edges outgoing from node $v_j^{(k-1)}$ distributes exactly the half of
the inputs tested at $v_j^{(k-1)}$.

\medskip
We say that a 1-branching program $P$ of width $w$ is \emph{normalized} if
$P$ has the minimum depth among the programs computing the same function
(e.g.\ $P$ does not contain the identity transition $T_k$) and $P$ satisfies
\begin{equation}
\label{normdis}
1>p_1^{(k)}\geq p_2^{(k)}\geq\cdots\geq p_w^{(k)}>0
\end{equation}
for every $k\geq \log w$ (hereafter, $\log$ denotes the binary logarithm).
Obviously, condition (\ref{normdis}) can always be met by possible splitting
(if $p_w^{(k)}=0$) and permuting the nodes at each level of~$P$:
\begin{lemma}[\cite{Sima07}]
\label{lemnorm}
Any width-$w$ 1-branching program can be normalized.
\end{lemma}
In the sequel, we confine ourselves to the 1-branching programs of width $w=3$.
Any such normalized program $P$ satisfies $p_1^{(k)}+p_2^{(k)}+p_3^{(k)}=1$ and
$1>p_1^{(k)}\geq p_2^{(k)}\geq p_3^{(k)}>0$, which implies
\begin{equation}
\label{phi1}
p_1^{(k)}>\frac{1}{3}\,,\qquad
p_2^{(k)}<\frac{1}{2}\,,\qquad
p_3^{(k)}<\frac{1}{3}
\end{equation}
for every level $2\leq k\leq d$ where $d\leq n$ is the depth of $P$.
Note that the strict inequalities for $p_1^{(k)}$ and $p_3^{(k)}$ in
(\ref{phi1}) hold since $p_i^{(k)}\not=\frac{1}{3}$ according to
(\ref{pkAkpk1}) and $t_{ij}^{(k)}\in\{0,\frac{1}{2},1\}$.

\section{The weak richness condition is necessary}
\label{necc}

Let ${\cal P}$ be a class of branching programs and $\varepsilon>0$ be
a real constant. A set of input strings $H\subseteq\{0,1\}^*$ is called an
\emph{$\varepsilon$-hitting set} for class ${\cal P}$ if for sufficiently
large $n$, for every branching program $P\in{\cal P}$ with $n$ input
variables
\begin{equation}
\label{hitting}
\frac{\left|P^{-1}(1)\right|}{2^n}\geq\varepsilon
\quad\mbox{implies}\quad
(\exists\,\mathbf{a}\in H\cap\{0,1\}^n)\,P(\mathbf{a})=1\,.
\end{equation}
Furthermore, we say that a set $A\subseteq\{0,1\}^*$ is
\emph{weakly $\varepsilon$-rich} if for sufficiently large $n$, for any index
set $I\subseteq\{1,\ldots,n\}$, and for any partition
$\{Q_1,\ldots,Q_{q},R_{1},\ldots,R_{r}\}$ of $I$ where $q\geq 0$ and
$r\geq 0$, and for any $\mathbf{c}\in\{0,1\}^n$ the following implication
holds: If
\begin{equation}
\label{wacond}
\left(1-\prod_{j=1}^q\left(1-\frac{1}{2^{|Q_j|}}\right)\right)\times
\prod_{j=1}^r\left(1-\frac{1}{2^{|R_j|}}\right)\geq\varepsilon\,,
\end{equation}
then there exists $\mathbf{a}\in A\cap\{0,1\}^n$ such that
\begin{eqnarray}
\label{wcond1}
&&(\exists\,j\in\{1,\ldots,q\})\,(\forall\, i\in Q_j)\, a_i=c_i\\
\label{wcond2}
&\mbox{and}&
(\forall\,j\in\{1,\ldots,r\})\,(\exists\, i\in R_j)\,
a_i\not=c_i\,.\qquad
\end{eqnarray}
Particularly for $q=0$ inequality (\ref{wacond}) reads
\begin{equation}
\label{acondR}
\prod_{j=1}^r\left(1-\frac{1}{2^{|R_j|}}\right)\geq\varepsilon
\end{equation}
and conjunction
(\ref{wcond1}) and (\ref{wcond2}) reduces to the second conjunct
(\ref{wcond2}), while for $r=0$ inequality (\ref{wacond}) reads
\begin{equation}
\label{acondQ}
1-\prod_{j=1}^q\left(1-\frac{1}{2^{|Q_j|}}\right)\geq\varepsilon
\end{equation}
and conjunction
(\ref{wcond1}) and (\ref{wcond2}) reduces to the first conjunct
(\ref{wcond1}).

\medskip
Note that the product on the left-hand side of inequality (\ref{wacond})
expresses the probability that a random string $\mathbf{a}\in\{0,1\}^n$
(not necessarily in $A$) satisfies the conjunction~(\ref{wcond1}) and
(\ref{wcond2}). Moreover, this formula can be interpreted as a read-once
conjunction of DNF and CNF (each variable occurs at most once)
\begin{equation}
\label{DNF&CNF}
\bigvee_{j=1}^q\,\bigwedge_{i\in Q_j}\ell(x_i)\,\wedge\,
\bigwedge_{j=1}^r\,\bigvee_{i\in R_j}\lnot \ell(x_i)\,,
\quad\mbox{where}\quad
\ell(x_i)=\left\{
\begin{array}{ll}
x_i &\mbox{ for }c_i=1\\
\lnot x_i &\mbox{ for }c_i=0
\end{array}
\right.
\end{equation}
which accepts a random input with probability at least $\varepsilon$ according to~(\ref{wacond}). Hence, the weak richness condition is, in fact, equivalent to the definition of a hitting set for read-once conjunctions of DNF and CNF. The following proposition observes that the weak richness condition is necessary for any set to be a hitting set for width-3 1-branching programs. It is based on the clear facts that the 1-branching programs of width~3 can implement any read-once conjunction of DNF and CNF, and any hitting set for a class of functions hits any of its subclass. Nevertheless, we provide a detailed proof for a reader to get used to the introduced definitions and notations.
\begin{proposition}
Every $\varepsilon$-hitting set for the class of read-once
branching programs of width 3 is weakly $\varepsilon$-rich.
\end{proposition}

\begin{proof}
We proceed by transposition. Assume a set $H\subseteq\{0,1\}^*$ is not
weakly $\varepsilon$-rich which means that
for infinitely many $n$ there is an index set $I\subseteq\{1,\ldots,n\}$,
a partition $\{Q_1,\ldots,Q_{q},$ $R_{1},\ldots,R_{r}\}$ of $I$ satisfying
(\ref{wacond}), and a string $\mathbf{c}\in\{0,1\}^n$ such that every
$\mathbf{a}\in H\cap\{0,1\}^n$ meets
\begin{eqnarray}
\label{ncond1}
&&(\forall\,j\in\{1,\ldots,q\})\,(\exists\, i\in Q_j)\, a_i\not=c_i\\
\label{ncond2}
&\mbox{or}&
(\exists\,j\in\{1,\ldots,r\})\,(\,\forall i\in R_j)\,
a_i=c_i\,.
\end{eqnarray}
We will use this partition and $\mathbf{c}$ for constructing a (non-normalized
oblivious) width-3 1-branching program $P$ such that
\begin{equation}
\label{nhitt}
\frac{\left|P^{-1}(1)\right|}{2^n}\geq\varepsilon
\quad\mbox{and}\quad
(\forall\,\mathbf{a}\in H\cap\{0,1\}^n)\,P(\mathbf{a})=0\,,
\end{equation}
which negates that $H$ is an $\varepsilon$-hitting set for 1-branching programs
of width~3 according to (\ref{hitting}). In fact, $P$ implements the
corresponding negated conjunction of DNF and CNF (\ref{DNF&CNF}).

\begin{figure}[htbp]
\centering
\includegraphics[height=18.6cm]{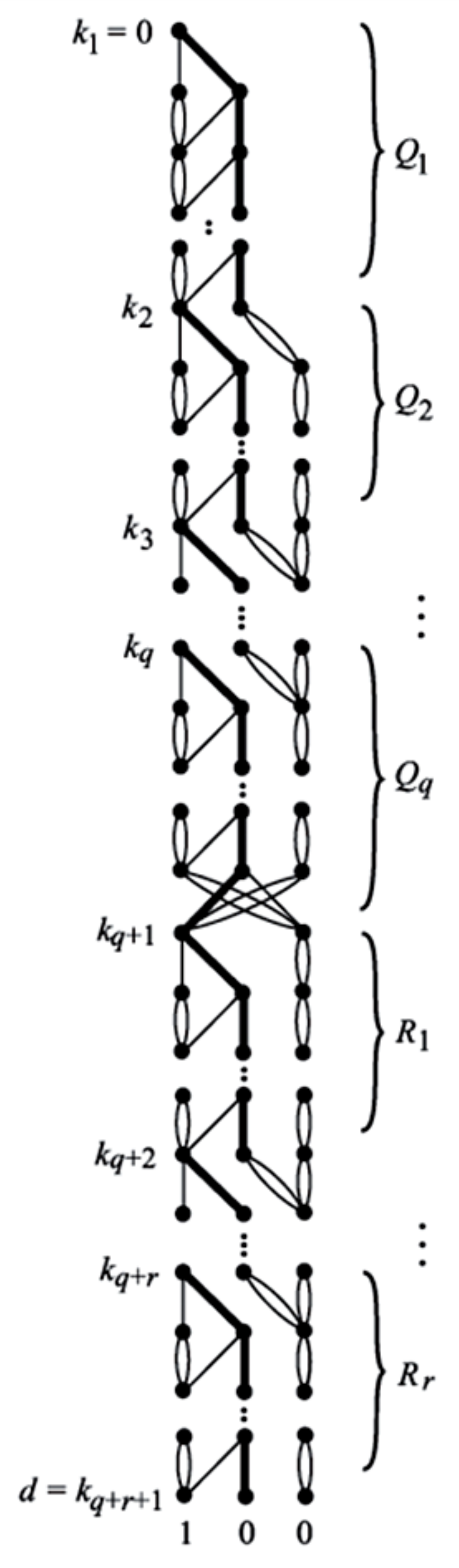}\vspace*{-1mm}
\caption{The Necessary Condition.}
\label{neccon}
\end{figure}

\medskip
As depicted in Figure~\ref{neccon}, branching
program $P$ is composed of $q+r$ consecutive blocks corresponding to the
partition classes $Q_1,\ldots,Q_{q},R_{1},\ldots,R_{r}$ which determine
the indices of variables that are queried within these blocks. For the simplicity of notation, we assume $q\geq 1$, $r\geq 1$, and $|Q_q|>1$, while the proof for $q=0$ or $r=0$ or $|Q_q|=1$ follows the same argument\footnote{The program $P$ in Figure~\ref{neccon} starts or terminates (with the sinks $v_1^{(k_{q+1})}$ and $v_3^{(k_{q+1})}$  labeled with 1 and 0, respectively) at level $k_{q+1}$ if $q=0$ or $r=0$, respectively, while for $|Q_q|=1$ the levels $k_q+1$ and $k_{q+1}$ coincide and the respective transitions are merged as $t_{11}^{(k_{q+1})}=t_{31}^{(k_{q+1})}=\frac{1}{2}$ and $t_{12}^{(k_{q+1})}=t_{13}^{(k_{q+1})}=1$.}. The block associated with $Q_j$ for $j\in\{1,\ldots,q\}$ starts on level $k_j=\sum_{\ell=1}^{j-1}|Q_\ell|$ of $P$ (e.g.\ $k_1=0$) with a transition satisfying $t_{11}^{(k_j+1)}=t_{21}^{(k_j+1)}=\frac{1}{2}$, followed by a sequence of transitions that meet $t_{11}^{(k)}=1$ and $t_{12}^{(k)}=t_{22}^{(k)}=\frac{1}{2}$ for every $k=k_j+2,\ldots,k_j+|Q_j|$, except for the boundary level $k_q+|Q_q|=k_{q+1}$, which is defined below. In addition, there is a parallel double-edge path leading from the node $v_3^{(k_2+1)}$ on level $k_2+1$ up to node $v_3^{(k_{q+1}-1)}$, and thus $t_{33}^{(k)}=1$ for every $k=k_2+2,k_2+3,\ldots,k_{q+1}-1$. This path
is wired up by $q-1$ double edges coming from nodes $v_2^{(k_j)}$, that
is, $t_{32}^{(k_j+1)}=1$ for every $j=2,\ldots,q$.
Finally, a special boundary transition is defined on level $k_{q+1}$
as $t_{31}^{(k_{q+1})}=t_{13}^{(k_{q+1})}=1$ and
$t_{12}^{(k_{q+1})}=t_{32}^{(k_{q+1})}=\frac{1}{2}$. Note that there
are only two nodes $v_1^{(k_{q+1})},v_3^{(k_{q+1})}$ on the boundary
level $k_{q+1}$.

Furthermore, $P$ continues analogously with
blocks corresponding to $R_j$ for $j=1,\ldots,r$, each starting on level
$k_{q+j}=k_{q+1}+\sum_{\ell=1}^{j-1}|R_\ell|$ (e.g.\
$k_{q+r+1}=d$ is the depth of $P$) with the transition satisfying
$t_{11}^{(k_{q+j}+1)}=t_{21}^{(k_{q+j}+1)}=\frac{1}{2}$,
followed by $t_{11}^{(k)}=1$ and $t_{12}^{(k)}=t_{22}^{(k)}=\frac{1}{2}$
for every $k=k_{q+j}+2,\ldots,k_{q+j}+|R_j|$, including the parallel
double-edge path, that is, $t_{33}^{(k)}=1$ for every $k=k_{q+1}+1,\ldots,d$
and $t_{32}^{(k_{q+j}+1)}=1$ for every $j=2,\ldots,r$.
The branching program $P$ then queries the value of each variable
$x_i$ such that $i\in Q_j$ for some $j\in\{1,\ldots,q\}$
or $i\in R_j$ for some $j\in\{1,\ldots,r\}$ only on one level
$k\in\{k_j,\ldots,k_{j+1}-1\}$ or $k\in\{k_{q+j},\ldots,k_{q+j+1}-1\}$,
respectively (i.e.\ the nodes on level $k$ are labeled with $x_i$),
while the single edge leading to $v_2^{(k+1)}$ (or to $v_1^{(k_{q+1})}$
for $k=k_{q+1}-1$) on the subsequent level $k+1$ (indicated by a bold line
in Figure~\ref{neccon}) gets label $c_i$. Finally, the sink $v_1^{(d)}$
gets label $1$, whereas the sinks $v_2^{(d)}$, $v_3^{(d)}$ are labeled
with the output $0$, which completes the construction of $P$.

\medskip
Clearly, $P$ is an (oblivious) read-once branching program of width 3.
The probability that an input reaches the node $v_3^{(k_{q+1})}$ on
the boundary level $k_{q+1}$ can simply be computed as
\begin{equation}
p_3^{(k_{q+1})}=\prod_{j=1}^q\left(1-\frac{1}{2^{|Q_j|}}\right)\,,
\end{equation}
while the probability of the complementary event that an input
reaches $v_1^{(k_{q+1})}$ equals
$p_1^{(k_{q+1})}=$
\linebreak
$1-p_3^{(k_{q+1})}$. Therefore, the probability that
$P$ outputs $1$ can be expressed and lower-bounded by~(\ref{wacond}):
\begin{equation}
\frac{\left|P^{-1}(1)\right|}{2^n}=p_1^{(d)}
=\left(1-\prod_{j=1}^q\left(1-\frac{1}{2^{|Q_j|}}\right)\right)\times
\prod_{j=1}^r\left(1-\frac{1}{2^{|R_j|}}\right)\geq\varepsilon\,.
\end{equation}
Furthermore, we split $H\cap\{0,1\}^n=A_1\cup A_2$ into two parts so that every
$\mathbf{a}\in A_1$ satisfies the first term (\ref{ncond1}) of the underlying
disjunction, whereas every $\mathbf{a}\in A_2=H\setminus A_1$ meets the second
term (\ref{ncond2}). Thus, for any input $\mathbf{a}\in A_1$
and for every $j\in\{1,\ldots,q\}$ the block of $P$ corresponding to $Q_j$
contains a level $k\in\{k_j,\ldots,k_{j+1}-1\}$ where variable $x_i$ is
tested such that $a_i\not=c_i$. This ensures that the computational path for
$\mathbf{a}\in A_1$ reaches $v_3^{(k_{q+1})}$ and further continues through
$v_3^{(k_{q+1}+1)},\ldots,v_3^{(d)}$, which gives $P(\mathbf{a})=0$ for every
$\mathbf{a}\in A_1$. Similarly, for any input $\mathbf{a}\in A_2$ there
exists a block of $P$ corresponding to $R_j$ for some $j\in\{1,\ldots,r\}$
such that the computational path for $\mathbf{a}$ traverses nodes
$v_1^{(k_{q+j})},v_2^{(k_{q+j}+1)},v_2^{(k_{q+j}+2)},\ldots,
v_2^{(k_{q+j}+|R_j|)}$. For $j<r$ this path continues through
$v_3^{(k_{q+j+1}+1)},\ldots,v_3^{(d)}$, whereas for $j=r$ it terminates
at $v_2^{(d)}$, which gives $P(\mathbf{a})=0$ in both cases. Hence,
$P$ satisfies (\ref{nhitt}), which completes the proof.
\end{proof}

\section{The richness condition is sufficient}
\label{suffc}

We say that a set $A\subseteq\{0,1\}^*$ is
\emph{$\varepsilon$-rich} if for sufficiently large $n$, for any index
set $I\subseteq\{1,\ldots,n\}$, and for any partition
$\{R_{1},\ldots,R_{r}\}$ of $I$ ($r\geq 0$) satisfying
\begin{equation}
\label{acond}
\prod_{j=1}^r\left(1-\frac{1}{2^{|R_j|}}\right)\geq\varepsilon\,,
\end{equation}
and for any $Q\subseteq\{1,\ldots,n\}\setminus I$ such that $|Q|\leq\log n$,
for any $\mathbf{c}\in\{0,1\}^n$ there exists $\mathbf{a}\in A\cap\{0,1\}^n$
that meets
\begin{equation}
\label{cond}
(\forall\, i\in Q)\, a_i=c_i\mbox{ and }
(\forall\,j\in\{1,\ldots,r\})\,(\exists\, i\in R_j)\,a_i\not=c_i\,.
\end{equation}
One can observe that an $\varepsilon$-rich set is weakly $\varepsilon$-rich
(see Section~\ref{necc}) since inequality (\ref{wacond}) implies (\ref{acond})
and
\begin{equation}
\label{wacond2}
1-\prod_{j=1}^q\left(1-\frac{1}{2^{|Q_j|}}\right)\geq\varepsilon
\end{equation}
which ensures that there is index $j\in\{1,\ldots,q\}$ of $Q_j=Q$ such
that $|Q|\leq\log n$. Namely, if $|Q_j|>\log n$ for every $j=1,\ldots,q$,
then inequality (\ref{wacond2}) would give
\begin{equation}
1-\varepsilon\geq
\prod_{j=1}^q\left(1-\frac{1}{2^{|Q_j|}}\right)
\geq\left(1-\frac{1}{2^{\log n}}\right)^{\frac{n}{\log n}}
>1-\frac{1}{n}\cdot\frac{n}{\log n}=1-\frac{1}{\log n}
\end{equation}
which is a contradiction for $n>2^{1/\varepsilon}$. Thus, we have (\ref{cond})
which validates the conjunction of (\ref{wcond1}) and (\ref{wcond2})
completing the argument.

\medskip
It follows that any rich set is a hitting set for read-once conjunctions of
DNF and CNF. Also note that formula (\ref{cond}) can be interpreted as
a read-once CNF (cf.~\ref{DNF&CNF})
\begin{equation}
\label{CNF}
\bigwedge_{i\in Q}\ell(x_i)\,\wedge\,
\bigwedge_{j=1}^r\,\bigvee_{i\in R_j}\lnot \ell(x_i)\,,
\quad\mbox{where}\quad
\ell(x_i)=\left\{
\begin{array}{ll}
x_i &\mbox{ for }c_i=1\\
\lnot x_i &\mbox{ for }c_i=0
\end{array}
\right.
\end{equation}
which contains at most logarithmic number of single literals together with
clauses whose sizes satisfy (\ref{acond}). Moreover, Theorem~\ref{richis}
in Section~\ref{kwISrich} proves that any almost $O(\log n)$-wise independent
set satisfies the richness condition.

\medskip
The following theorem shows that the richness condition is, in a certain sense,
sufficient for a set to be a hitting set for 3-width 1-branching programs.
For an input $\mathbf{a}\in\{0,1\}^n$ and an integer constant $c\geq 0$,
denote by $\Omega_c(\mathbf{a})=\{\mathbf{a}'\in\{0,1\}^n\,|\,
h(\mathbf{a},\mathbf{a}')\leq c\}$ the set of so-called
\mbox{\emph{$c$-neighbors}} of $\mathbf{a}$, where $h(\mathbf{a},\mathbf{a}')$
is the Hamming distance between $\mathbf{a}$ and $\mathbf{a}'$
(i.e.\ the number bits in which $\mathbf{a}$ and $\mathbf{a}'$ differ).
We also define $\Omega_c(A)=\bigcup_{\mathbf{a}\in A}\Omega_c(\mathbf{a})$
for a given set $A\subseteq\{0,1\}^*$.
\begin{theorem}
\label{suff}
Let $\varepsilon>\frac{5}{6}$.
If $A$ is $\varepsilon'^{11}$-rich for some $\varepsilon'<\varepsilon$,
then $H=\Omega_3(A)$ is an $\varepsilon$-hitting set for the class of
read-once branching programs of width~3.
\end{theorem}
\begin{proof}
Suppose a read-once branching program $P$ of width 3 with sufficiently many input
variables $n$ meets
\begin{equation}
\label{podmhitt}
\frac{\left|P^{-1}(1)\right|}{2^n}\geq\varepsilon
>\frac{5}{6}\,.
\end{equation}
We will prove that there exists $\mathbf{a}\in H$ such that $P(\mathbf{a})=1$.
On the contrary, we assume that
\begin{equation}
\label{PH0}
P(\mathbf{a})=0\quad\mbox{for every}\quad\mathbf{a}\in H\,,
\end{equation}
which will lead to a contradiction. Without loss of generality, we assume that $P$ is normalized according to Lemma~\ref{lemnorm}.
More precisely, the logical argument goes as follows. The branching program $P$ is transformed to an equivalent branching program $P_1$ which computes the same function as~$P$ (i.e.\ $P_1$ preserves (\ref{podmhitt}) and (\ref{PH0})) and has some additional property (e.g.\ $P_1$ is normalized). In the following proof, several equivalent transformations are employed one after the other in order to achieve various extra properties, which generates a sequence of branching programs $P,P_1,P_2,\ldots,P_c$. After showing that the existence of the last program $P_c$ eventually leads to a contradiction one can conclude that the original program $P$ cannot exist.

\subsection{The plan of the proof}
\label{prplan}

\begin{figure}[htbp]
\vspace*{-2mm}
\centering
\includegraphics[height=20cm]{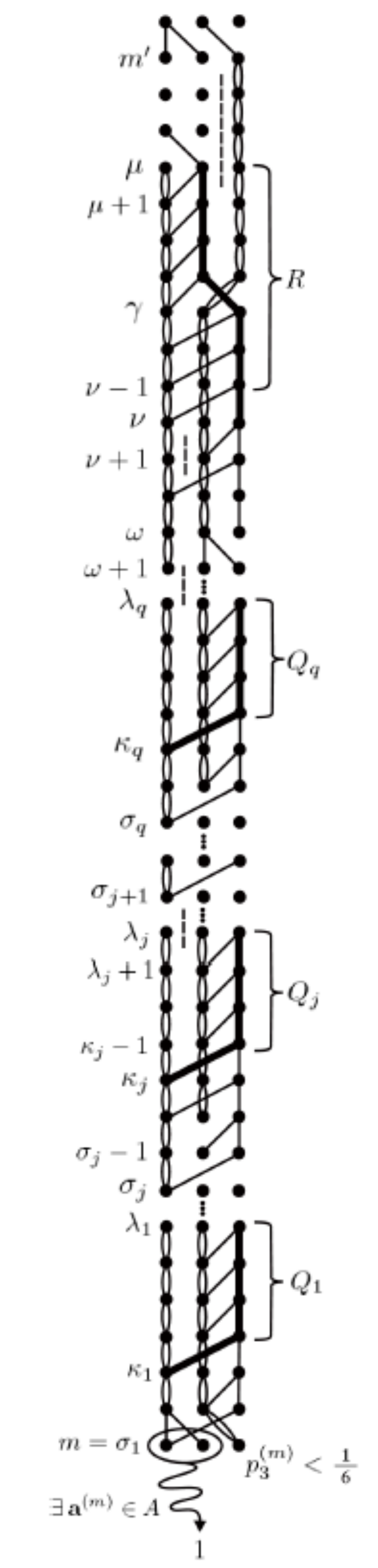}\vspace*{-3mm}
\caption{The structure of a typical block.}
\label{block}
\end{figure}

In this paragraph we will informally explain the main ideas of the proof with the pointers to the subsequent paragraphs and sections where a precise and detailed argument is given. The assumption that branching program $P$ accepts a large fraction of inputs and rejects all the inputs from candidate hitting set $H$ constrains the structure of $P$ severely.
In particular, we inspect the structure of $P$ with respect to (\ref{podmhitt}) and (\ref{PH0}), starting from its last (i.e.\ greatest) level $d$ (containing the sinks) and we proceed in the analysis step by step backwards to smaller levels\footnote{\label{lorder} Recall that we number the levels of $P$ from the least level 0 containing the source up to the greatest last level $d$ which is composed of sinks. Note that in figures, the smaller levels are situated on the top of branching program whereas the greater levels are located at the bottom.}. This analysis reveals that the structure at the end of branching program $P$ can be split into subsequent \emph{blocks} whose typical shape is schematically depicted in Figure~\ref{block} while the following \mbox{Figures~\ref{lemma2}--\ref{nutoo}} focus on particular parts of the block. As the inspection of these blocks proceeds backwards from the last level $d$ towards smaller levels, the blocks are numbered in the order reverse to that of levels, which means the first block represents the tail of~$P$.

\begin{table}[!ht]
\vspace*{-2mm}
\caption{A summary of definitions of block levels and corresponding important conditions on the distribution $p_1^{(k)},p_2^{(k)},p_3^{(k)}$ at these levels, where the notation \mbox{``$a\leq\mathbf{b}\uparrow\,\,\leq c$ :~~$C(b)$''} means $b$ is the \emph{greatest} level such that $a\leq b\leq c$ and condition $C(b)$ is satisfied (similarly, $\downarrow$ denotes the \emph{least} such level).}\label{levsum}\vspace*{-2mm}
\begin{center}
\scalebox{0.9}{
\begin{tabular}{|c||c|c|}
\hline
level & definition & conditions on the distribution\\
\hline\hline\rule{0pt}{1.5em}
$m'$ & $2\leq\mbox{\boldmath $m'$}\uparrow\,\,\leq\mu$ :~~$t_{32}^{(m')}>0$ & $p_3^{(m')}<\frac{1}{6}$\\[0.8ex]
\hline\rule{0pt}{1.5em}
$\mu$ & $2\leq{\boldsymbol\mu}\downarrow\,\,<m$ :~~$t_{11}^{(\ell)}=1$ & $p_1^{(\mu)}+p_2^{(\mu)}=p_1^{(m')}+p_2^{(m')}$\\[0.4ex]
 & for $\ell=\mu+1,\ldots,m-1$ & $p_3^{(\mu)}=p_3^{(m')}<\frac{1}{6}$\\[0.8ex]
\hline\rule{0pt}{1.5em}
$\gamma$ & $\mu\leq{\boldsymbol\gamma}\uparrow\,\,\leq\nu$ :~~$t_{12}^{(\gamma)}>0$ & \\[0.8ex]
\hline\rule{0pt}{1.5em}
$\nu$ & $\mu\leq{\boldsymbol\nu}\uparrow\,\,\leq m$ :~~$t_{12}^{(\ell)}+t_{13}^{(\ell)}>0$ & $p_1^{(\nu)}=p_1^{(\mu)}+p_2^{(\mu)}\left(1-\frac{1}{2^{|R|}}\right)$ \\[0.4ex]
 & for $\ell=\mu+1,\ldots,\nu$ & $p_3^{(\nu+1)}=\frac{p_2^{(\mu)}}{2^{|R|+1}}$\\[0.8ex]
\hline\rule{0pt}{1.5em}
$\omega$ & $\nu-1\leq{\boldsymbol\omega}\uparrow\,\,\leq m$ :~~$\exists$ double-edge path & $p_1^{(\omega+1)}\leq p_1^{(\nu)}+p_3^{(\nu+1)}= p_1^{(\mu)}+p_2^{(\mu)}\left(1-\frac{1}{2^{|R|+1}}\right)$ \\[0.4ex]
& $v_\mu,v_{\mu+1},\ldots,v_\omega$ where $v_\ell\in\left\{v_2^{(\ell)},v_3^{(\ell)}\right\}$ &
$\rightarrow\enspace$ \mbox{\boldmath $p_1^{(\omega+1)}\leq\left(p_1^{(m')}+p_2^{(m')}\right)\left(1-\frac{1}{2^{|R|+3}}\right)$}\\[0.8ex]
\hline\rule{0pt}{1.5em}
$\sigma_{j+1}$ & $\omega+1<\mbox{\boldmath $\sigma_{j+1}$}\uparrow\,\,\leq\lambda_j$ :~~ $t_{13}^{(\sigma_{j+1})}>0$ & $p_2^{(\sigma_j)}+p_3^{(\sigma_j)}\geq\left( p_2^{(\sigma_{j+1})}+p_3^{(\sigma_{j+1})}\right)\left(1-\frac{1}{2^{|Q_j|}}\right)$\\[0.8ex]
\hline\rule{0pt}{1.5em}
$\lambda_j$ & $\omega\leq\mbox{\boldmath $\lambda_j$}\downarrow\,\,<\sigma_j-1$ :~~ $t_{11}^{(\ell)}=t_{22}^{(\ell)}=1$ & $p_2^{(\lambda_q)}+p_3^{(\lambda_q)}=p_2^{(\omega+1)}+p_3^{(\omega+1)}$\\[0.4ex]
& \& $\,t_{33}^{(\ell)}=\frac{1}{2}$ for $\ell=\lambda_j+1,\ldots,\sigma_j-1$ & $p_2^{(\lambda_j)}+p_3^{(\lambda_j)}=p_2^{(\sigma_{j+1})}+p_3^{(\sigma_{j+1})}$\\[0.8ex]
\hline\rule{0pt}{1.5em}
$\kappa_j$ & $\lambda_j+1<\mbox{\boldmath $\kappa_j$}\downarrow\,\,\leq\sigma_j$ :~~ $t_{13}^{(\kappa_j)}>0$ & $p_3^{(\kappa_j-1)}=\frac{p_3^{(\lambda_j)}}{2^{|Q_j|-1}}\leq\frac{p_2^{(\lambda_j)}+p_3^{(\lambda_j)}}{2^{|Q_j|}}$\\[0.8ex]
\hline\rule{0pt}{1.5em}
$m$ & by $m$-conditions & \mbox{\boldmath $p_3^{(m)}\geq\left(p_2^{(\omega+1)}+p_3^{(\omega+1)}\right)\prod_{j=1}^q\left(1-\frac{1}{2^{|Q_j|}}\right)$}\\[0.8ex]
\hline
\end{tabular} }
\end{center}
\end{table}

Figure~\ref{block} contains various parameters denoting certain levels in $P$ which are used to describe the structure of $P$. These parameters cannot be defined formally in advance since their definitions often build on the preceding detailed analysis of $P$. Thus, the formal definitions of levels in $P$ which are indicated in boldface, are scattered below in the proof. Nevertheless, these levels are summarized in Table~\ref{levsum} which also includes the main conditions on the distribution $p_1^{(k)},p_2^{(k)},p_3^{(k)}$ of inputs among three nodes at important levels in the block. Table~\ref{levsum} as well as some details in Figure~\ref{block} which are presented for completeness, will probably be incomprehensible to a reader at this point but they can provide a useful overview when reading the following detailed proof.

\medskip
The \emph{last} (greatest) level of the block is denoted by $m$ and this \textbf{level} \mbox{\boldmath $m$} satisfies the following four so-called \emph{$m$-conditions}:
\begin{enumerate}
\item
\label{mc1}
$t_{11}^{(m)}=t_{21}^{(m)}=\frac{1}{2}$,
\item
\label{mc2}
$t_{32}^{(m)}>0$,
\item
\label{mc3}
$p_3^{(m)}<\frac{1}{6}$,
\item
\label{mc4}
there is $\mathbf{a}^{(m)}\in A$ such that if we put $\mathbf{a}^{(m)}$ at node
$v_1^{(m)}$ or $v_2^{(m)}$, then its onward computational path arrives to the
sink labeled with 1.
\end{enumerate}
Without loss of generality, these $m$-conditions can also be met for $m=d$ (Paragraph~\ref{inicasemd}) which is the last level of $P$. In particular, the sinks $v_1^{(d)}$ and $v_2^{(d)}$ are labeled with 1 according to $m$-condition~\ref{mc4}. Thus, the inspection of the structure of $P$ starts with the analysis of the first block which constitutes the tail of $P$.

\medskip
The \emph{first} (least) level of the block is denoted by $m'$ (its formal definition can be found in Paragraph~\ref{cond1-3}) which, in a typical case, proves also to satisfy the four $m'$-conditions~\ref{mc1}--\ref{mc4}. Thus the block is delimited by levels $m'$ and $m$. The shape of the block is being revealed step by step by the case analysis (Sections~\ref{dfpart} and~\ref{sbelowmu}) which starts from level $m$ and proceeds towards smaller levels down to $m'$.
We will now shortly outline the structure of a typical block as depicted in Figure~\ref{block} which results from this analysis. From the first level $m'$ through $\mu$, there is no edge between the first two columns and the third column, which means there is a double-edge path in the third column from $m'$ through $\mu$ (Paragraph~\ref{cond1-3}). Moreover, there is a double-edge path in the first column starting at level $\mu$ which leads up to level $m-1$ where it is split into vertices $v_1^{(m)}$ and $v_2^{(m)}$ at the next level $m$ (cf.\ $m$-condition~\ref{mc1}). Hence, if $\mathbf{a}^{(m)}\in A$ is redirected to the first column since layer $\mu$, then it will be accepted by $P$ according to $m$-condition~\ref{mc4}.

At the top of Figure~\ref{block}, a single-edge path from $\mu$ to $\nu$ is indicated in boldface which is used to define the partition class $R$ associated with this block (Paragraph~\ref{dfpartR}). In particular, class $R$ contains all the indices of the variables that are queried on this computational path up to level $\nu-1$. Moreover, the edge labels on this path define relevant bits of $\mathbf{c}\in\{0,1\}^n$ so that any input passing through this path that differs from $\mathbf{c}$ in at least one bit location from $R$ turns to the double-edge path in the first column and consequently enters node $v_1^{(m)}$ or $v_2^{(m)}$. This implements one CNF clause $\vee_{i\in R}\,\lnot\ell(x_i)$ from~(\ref{CNF}). Similarly, sets $Q_j$ for $j=1,\ldots,q$ associated with this block are defined (Paragraph~\ref{dfpartQ}) using the single-edge paths from $\lambda_j$ to $\kappa_j$ which are also highlighted in Figure~\ref{block} so that any input that passes through $v_3^{(\lambda_j)}$ and agrees with $\mathbf{c}$ on all the bit locations from $Q_j$ reaches the double-edge path in the first column coming in $v_1^{(m)}$ or $v_2^{(m)}$. This implements DNF monomials $\wedge_{i\in Q_j}\,\ell(x_i)$ in~(\ref{DNF&CNF}) which are candidates for the monomial $\wedge_{i\in Q}\,\ell(x_i)$ in~(\ref{CNF}). Thus, the general block structure corresponds to one CNF clause $\vee_{i\in R}\,\lnot\ell(x_i)$ for class $R$, followed by DNF monomials $\wedge_{i\in Q_j}\,\ell(x_i)$ for sets $Q_1,\ldots,Q_q$.

Under certain assumptions ((\ref{p3mu}) and (\ref{pik45})), one can show that level $m'$ satisfies \mbox{$m'$-condition~\ref{mc1}--\ref{mc3}} (Paragraph~\ref{cond1-3}). In such a case, the first level $m'=m_r$ of the current $r$th block might at the same time represent the last level of the next smaller-level $(r+1)$st block to which the structural analysis could recursively be applied (Section~\ref{recursionm}). It suffices to show that level $m'$ also meets \mbox{$m'$-condition~\ref{mc4}}.
For this purpose, the richness condition (\ref{cond}) is employed for $Q=\emptyset$ and for the partition classes $R_1,\ldots,R_r$ associated with the first $r$ blocks (that have been analyzed so far), provided that this partition satisfies (\ref{acond}). This gives an input $\mathbf{a}^{(m')}\in A$ such that for every block $j=1,\ldots,r$ there is $i\in R_j$ such that $a_i^{(m')}\not=c_i$ according to (\ref{cond}), that is, $\mathbf{a}^{(m')}$ satisfies  $\wedge_{j=1}^r\vee_{i\in R_j}\lnot\ell(x_i)$ from (\ref{CNF}). Hence, if we put this $\mathbf{a}^{(m')}$ at node $v_1^{(m')}$ or $v_2^{(m')}$ ($m'=m_r$), then the block structure in Figure~\ref{block} ensures that $\mathbf{a}^{(m')}$ also traverses $v_1^{(\mu)}$ or $v_2^{(\mu)}$ and reaches the double-edge path in the first column coming in $v_1^{(m)}$ or $v_2^{(m)}$ ($m=m_{r-1}$), by the definition of $R$ and $c_i$ for $i\in R$. This argument is applied recursively to each block $j=r,r-1,\ldots,1$ which implies that $\mathbf{a}^{(m')}$ eventually arrives to the sink $v_1^{(d)}$ or $v_2^{(d)}$ ($m_0=d$) labeled with 1. This proves the $m'$-condition~\ref{mc4} also for level $m'$.
Thus the analysis including the definition of an associated partition class $R=R_{r+1}$ and sets $Q_j=Q_{r+1,j}$ is applied recursively to the next $(r+1)$st block for $m$ replaced with $m'$ etc.

If, on the other hand, the underlying partition does not satisfy (\ref{acond}), then one can prove that there is a set $Q=Q_{bj}$ associated with the $b$th block among the first $r$ blocks (that have been analyzed so far) such that $|Q|\leq\log n$, and the recursive analysis ends (Section~\ref{indend}). In this case, the richness condition (\ref{cond}) for this set $Q$ and for partition $R_1,\ldots,R_{b-1}$ provides $\mathbf{a}\in A$ such that $a_i=c_i$ for each $i\in Q$ and for every block $j=1,\ldots,b-1$ there is $i\in R_j$ such that $a_i\not=c_i$, that is, $\mathbf{a}\in A$ satisfies $\wedge_{i\in Q}\,\ell(x_i)\,\wedge\,\wedge_{j=1}^{b-1}\vee_{i\in R_j}\lnot\ell(x_i)$ according to (\ref{CNF}). Moreover, one can show (Lemma~\ref{gsp}) that there is a \mbox{2-neighbor} $\mathbf{a}'\in\Omega_2(\mathbf{a})\subseteq H$ that differs from this $\mathbf{a}$ in at most two bits so that these bits guarantee that the computational path for $\mathbf{a}'$ in the $b$th block either reaches the double-edge path in the first column, or comes in node $v_3^{(\lambda_j)}$ (see Figure~\ref{block}). In the latter case, $\mathbf{a}'$ further traverses the path corresponding to $Q$ which reaches the double-edge path in the first column anyway by the definition of $Q_j$ and $c_i$ for $i\in Q_j$. In both cases, input $\mathbf{a}'$ traverses node $v_1^{(m_{b-1})}$ or $v_2^{(m_{b-1})}$, and by the above-mentioned recursive argument it eventually arrives to the sink $v_1^{(d)}$ or $v_2^{(d)}$ labeled with 1. This provides the desired contradiction $P(\mathbf{a}')=1$ for $\mathbf{a}'\in H$.

\subsection{The initial case of $m=d$}
\label{inicasemd}

We will first observe that the four $m-$conditions can be met
for $m=d$. Clearly, both edges outgoing from $v_1^{(d-1)}$ lead to the
sink(s) labeled with 1 since $p_1^{(d-1)}>\frac{1}{3}$ due to (\ref{phi1}) and
$|P^{-1}(0)|/2^n<\frac{1}{6}$ according to (\ref{podmhitt}). Hence, we will
assume without loss of generality that $t_{11}^{(d)}=t_{21}^{(d)}=\frac{1}{2}$
($m$-condition~\ref{mc1}) while the remaining edges that originally led to the sinks
labeled with 1 or 0 are possibly redirected to $v_1^{(d)}$ or $v_3^{(d)}$,
respectively, so that the normalization condition
$p_1^{(d)}\geq p_2^{(d)}>\frac{1}{6}>p_3^{(d)}$ ($m$-condition~\ref{mc3})
is preserved by (\ref{podmhitt}). Thus, sinks $v_1^{(d)}$ and $v_2^{(d)}$ are
labeled with 1 ($m$-condition~\ref{mc4}) whereas sink $v_3^{(d)}$ gets label~0.
Finally, we show that $t_{32}^{(d)}>0$ ($m$-condition~\ref{mc2}). On the contrary,
suppose $t_{32}^{(d)}=0$, which implies $t_{33}^{(d)}>0$ and
$H\subseteq P^{-1}(0)\subseteq M(v_3^{(d-1)})$ due to $t_{31}^{(d)}=0$.
In the case of $t_{13}^{(d)}+t_{23}^{(d)}>0$, the computational path for a 1-neighbor $\mathbf{a'}\in\Omega_1(\mathbf{a})$ of
$\mathbf{a}\in A\subseteq H\subseteq M(v_3^{(d-1)})$ that differs from
$\mathbf{a}$ in the $i$th bit that is tested at node $v_3^{(d-1)}$
(i.e.\ $v_3^{(d-1)}$ is labeled with $x_i$), would reach the sink labeled
with~1, and hence $P(\mathbf{a}')=1$ which contradicts the assumption
$H\subseteq P^{-1}(0)$. For $t_{33}^{(d)}=1$, on the other hand, we
could shorten $P$ by removing the last level $d$ while preserving its function
and condition (\ref{podmhitt}), which is in contradiction with the normalization (minimality) of $P$.
This completes the proof that $m$-conditions~\ref{mc1}--\ref{mc4} can be assumed for $m=d$
without loss of generality.

\subsection{A technical lemma}

Let \textbf{level} $\boldsymbol\mu'$ be the least level of $P$ such that $2\leq\mu'<m$
and $t_{11}^{(\ell)}=1$ for every $\ell=\mu'+1,\ldots,m-1$. We define \textbf{level}
$\boldsymbol\mu$ as
\begin{equation}
\label{defmu}
\mu=\left\{
\begin{array}{ll}
\mu'-1&\quad\mbox{if }t_{12}^{(\mu')}=1\mbox{ and }
t_{11}^{(\mu')}=t_{21}^{(\mu')}=\frac{1}{2}\\
\mu'&\quad\mbox{otherwise.}
\end{array}
\right.
\end{equation}
For the analysis of a single block structure
(Sections~\ref{suffc}--\ref{sbelowmu}, \ref{indend}),
we swap $v_1^{(\mu)}$ and $v_2^{(\mu)}$ if $\mu=\mu'-1$ for the notation simplicity
so that $t_{11}^{(\ell)}=1$ for every $\ell=\mu+1,\ldots,m-1$ at the cost of violating
condition $p_1^{(\mu)}\geq p_2^{(\mu)}$ given by (\ref{normdis}). Thus, for $\mu=\mu'-1$,
assume $p_1^{(\mu)}<p_2^{(\mu)}$, $t_{11}^{(\mu+1)}=1$, and
$t_{12}^{(\mu+1)}=t_{22}^{(\mu+1)}=\frac{1}{2}$.  For the recursion
(Section~\ref{recursionm}) when the last level $m$ in the next (smaller-level) block
may coincide with level $\mu$ of the current block we will nevertheless assume the
original node order and $p_1^{(\mu)}\geq p_2^{(\mu)}$.

\medskip
The following lemma represents a technical tool which will be used for the analysis
of the block from level $\mu$ through $m$. For this purpose, define a so-called \emph{switching path} starting from $v\in\{v_2^{(k)},v_3^{(k)}\}$ at level $k$, where $\mu\leq k<m$, to be a computational path of length at most 3 edges leading from $v$ to $v_1^{(\ell)}$ at level $\ell$ such that $k<\ell\leq\min(k+3,m)$ or possibly to $v_2^{(m)}$ for $m\leq k+3$.
\begin{lemma}~
\label{aboveedge}
\begin{enumerate}
\renewcommand\labelenumi{\theenumi}
\renewcommand{\theenumi}{(\roman{enumi})}
\item
$\mu>3$.
\item
There are no two switching paths starting from $v_2^{(k)}$ and
$v_3^{(k)}$, respectively, at any level $k$ such that $\mu\leq k<m$.
\item
If $t_{12}^{(k+1)}>0$ for some level $k$ such that $\mu\leq k<m$, then $t_{11}^{(\ell)}=t_{33}^{(\ell)}=1$,
$t_{12}^{(\ell)}=t_{22}^{(\ell)}=\frac{1}{2}$ for every $\ell=\mu+1,\ldots,k$,
and $t_{12}^{(k+1)}=\frac{1}{2}$ (see Figure~\ref{lemma2}).
\item
If $t_{13}^{(k+1)}>0$ for some level $k$ such that $\mu<k<m$, then one of
the following four cases occurs:
\begin{enumerate}
\renewcommand\labelenumii{\theenumii}
\renewcommand{\theenumii}{\arabic{enumii}.}
\item\label{c1}
$t_{11}^{(k)}=t_{23}^{(k)}=1$ and $t_{12}^{(k)}=t_{32}^{(k)}=\frac{1}{2}$,
\item\label{c2}
$t_{11}^{(k)}=t_{23}^{(k)}=1$ and $t_{22}^{(k)}=t_{32}^{(k)}=\frac{1}{2}$,
\item\label{c3}
$t_{11}^{(k)}=t_{22}^{(k)}=1$ and $t_{13}^{(k)}=t_{33}^{(k)}=\frac{1}{2}$,
\item\label{c4}
$t_{11}^{(k)}=t_{22}^{(k)}=1$ and $t_{23}^{(k)}=t_{33}^{(k)}=\frac{1}{2}$.
\end{enumerate}
In addition, if $t_{23}^{(k)}=1$ (case~1 or~2),
then $t_{11}^{(\ell)}=t_{33}^{(\ell)}=1$
and $t_{12}^{(\ell)}=t_{22}^{(\ell)}=\frac{1}{2}$ for every
$\ell=\mu+1,\ldots,k-1$ (see Figure~\ref{lemma2}).
\end{enumerate}
\end{lemma}
\begin{proof}

\vspace*{-6mm}
\begin{enumerate}
\renewcommand\labelenumi{\theenumi}
\renewcommand{\theenumi}{(\roman{enumi})}
\item
Suppose $\mu\leq 3$ and let $\mathbf{a}^{(m)}\in A$ be the input from
$m$-condition~\ref{mc4}. Then there exists a \mbox{3-neighbor}
$\mathbf{a}'\in\Omega_3(\mathbf{a}^{(m)})$ of $\mathbf{a}^{(m)}$ whose
computational path starting from source $v_1^{(0)}$ reaches $v_1^{(\mu)}$.
Hence, $P(\mathbf{a}')=1$ for $\mathbf{a}'\in H$ follows from
$M(v_1^{(\mu)})\subseteq M(v_1^{(m)})\cup M(v_2^{(m)})$ and
\mbox{$m$-condition}~\ref{mc4}, which is a contradiction, and thus $\mu>3$.
\item
Suppose there are two switching paths starting from $v_2^{(k)}$
and $v_3^{(k)}$, respectively, at some level $k$ such that $\mu\leq k<m$, and let
$\mathbf{a}^{(m)}\in A$ be the input satisfying $m$-condition~\ref{mc4}. Clearly,
$\mathbf{a}^{(m)}\not\in M(v_1^{(k)})\subseteq M(v_1^{(m)})\cup M(v_2^{(m)})$ since
otherwise $P(\mathbf{a}^{(m)})=1$ for $\mathbf{a}^{(m)}\in H$. Thus, assume
$\mathbf{a}^{(m)}\in M(v)$ for $v\in\{v_2^{(k)},v_3^{(k)}\}$. Then there is
a 3-neighbor $\mathbf{a}'\in\Omega_3(\mathbf{a}^{(m)})\cap M(v)$ of $\mathbf{a}^{(m)}$
whose computational path follows the switching path starting from $v$.
Hence, $\mathbf{a}'\in M(v_1^{(m)})\cup M(v_2^{(m)})$ implying
$P(\mathbf{a}')=1$ for $\mathbf{a}'\in H$ due to $P$ is read-once, which is a contradiction. This completes the proof of (ii).
\end{enumerate}
\begin{figure}[h]
\vspace*{-4mm}
\centering
\includegraphics[width=10.1cm]{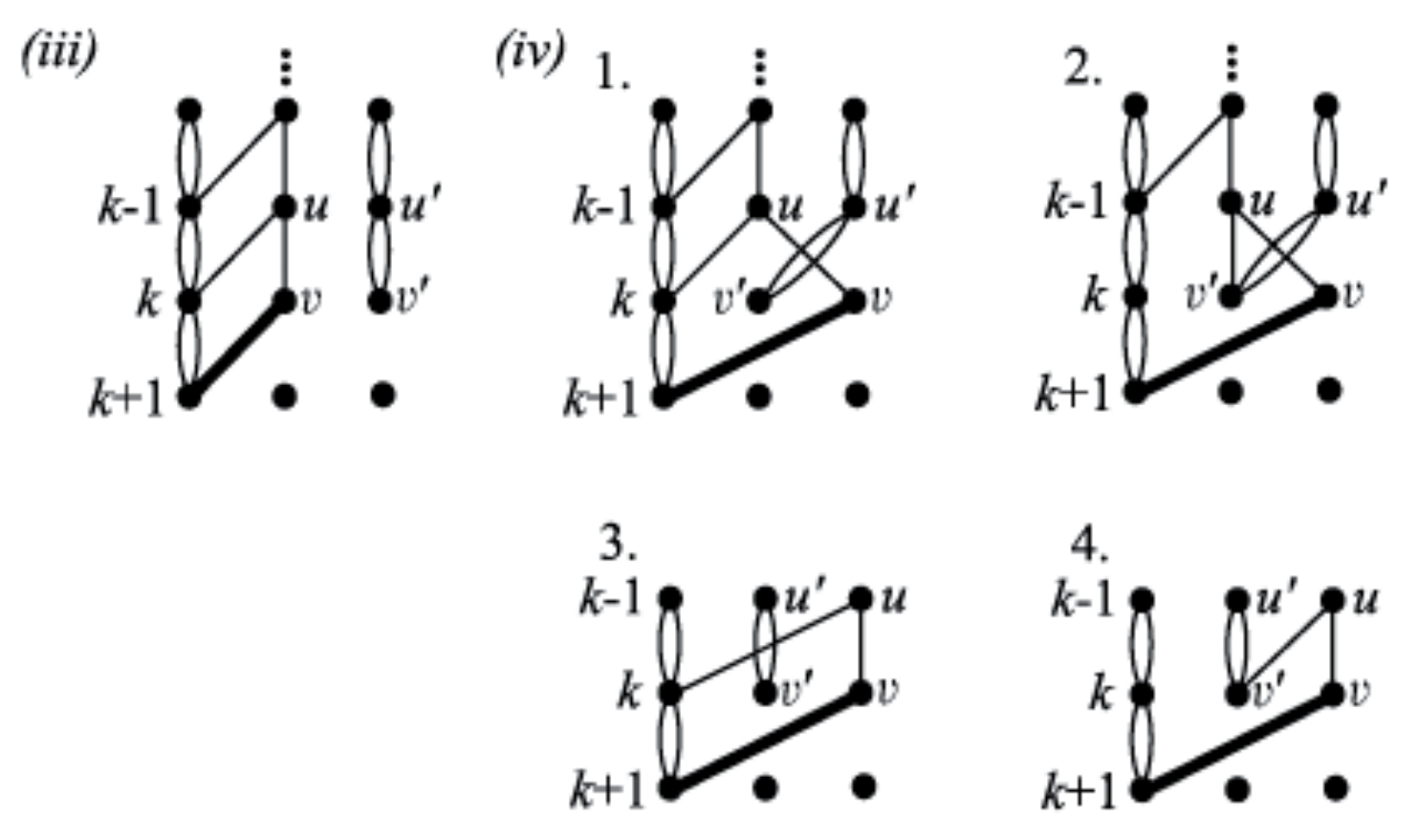}\vspace*{-2mm}
\caption{Lemma~\ref{aboveedge}.iii and iv.}
\label{lemma2}\vspace*{-1mm}
\end{figure}
As depicted in Figure~\ref{lemma2}, at level $k$ such that $\mu<k<m$,
denote by $v\in\{v_2^{(k)},v_3^{(k)}\}$ a node with the edge outgoing to $v_1^{(k+1)}$, and let $u$ be a node on level $k-1$ from which an edge
leads to $v$, while $v'\in\{v_2^{(k)},v_3^{(k)}\}\setminus\{v\}$ and
$u'\in\{v_2^{(k-1)},v_3^{(k-1)}\}\setminus\{u\}$ denote the other nodes.
It follows from (ii) there is no edge from $u'$ to $v$ nor to $v_1^{(k)}$,
which would establish two switching paths starting from
$v_2^{(k-1)}$ and $v_3^{(k-1)}$, respectively. Hence, there must be
a double edge from $u'$ to $v'$. Since $P$ is normalized, $u'=v_2^{(k-1)}$
and $v'=v_3^{(k)}$ cannot happen simultaneously. Moreover, the second
edge from $u$ may lead either to $v_1^{(k)}$ or to $v'$ if
$v'\not=v_3^{(k)}$. Now, the possible cases can be summarized:
\begin{enumerate}
\renewcommand\labelenumi{\theenumi}
\renewcommand{\theenumi}{(\roman{enumi})}
\item[(iii)]
For $t_{12}^{(k+1)}>0$ we know $v=v_2^{(k)}$ and $v'=v_3^{(k)}$, which implies
$t_{11}^{(k)}=t_{33}^{(k)}=1$ and $t_{12}^{(k)}=t_{22}^{(k)}=\frac{1}{2}$.
The proposition follows when this argument is applied recursively for $k$
replaced with $k-1$ etc. In addition, we will prove that $t_{12}^{(k+1)}<1$
for $\mu\leq k<m$. Clearly, $t_{12}^{(m)}<1$ from $m$-condition~\ref{mc2}, and
hence suppose $k<m-1$. Also for $k=\mu=\mu'-1$ we
know $t_{12}^{(\mu+1)}=\frac{1}{2}$ and thus we further assume $k\geq\mu'$.
On the contrary, suppose $t_{12}^{(k+1)}=1$ which
implies $t_{23}^{(k+1)}=t_{33}^{(k+1)}=\frac{1}{2}$.

For $k>\mu$, one could shorten $P$ by merging level $k$ with $\mu$ without changing its function as it is shown in Figure~\ref{contr}, which is a contradiction with the normalization (minimality) of~$P$.

\begin{figure}[h]
\vspace*{2mm}
\centering
\includegraphics[height=4cm]{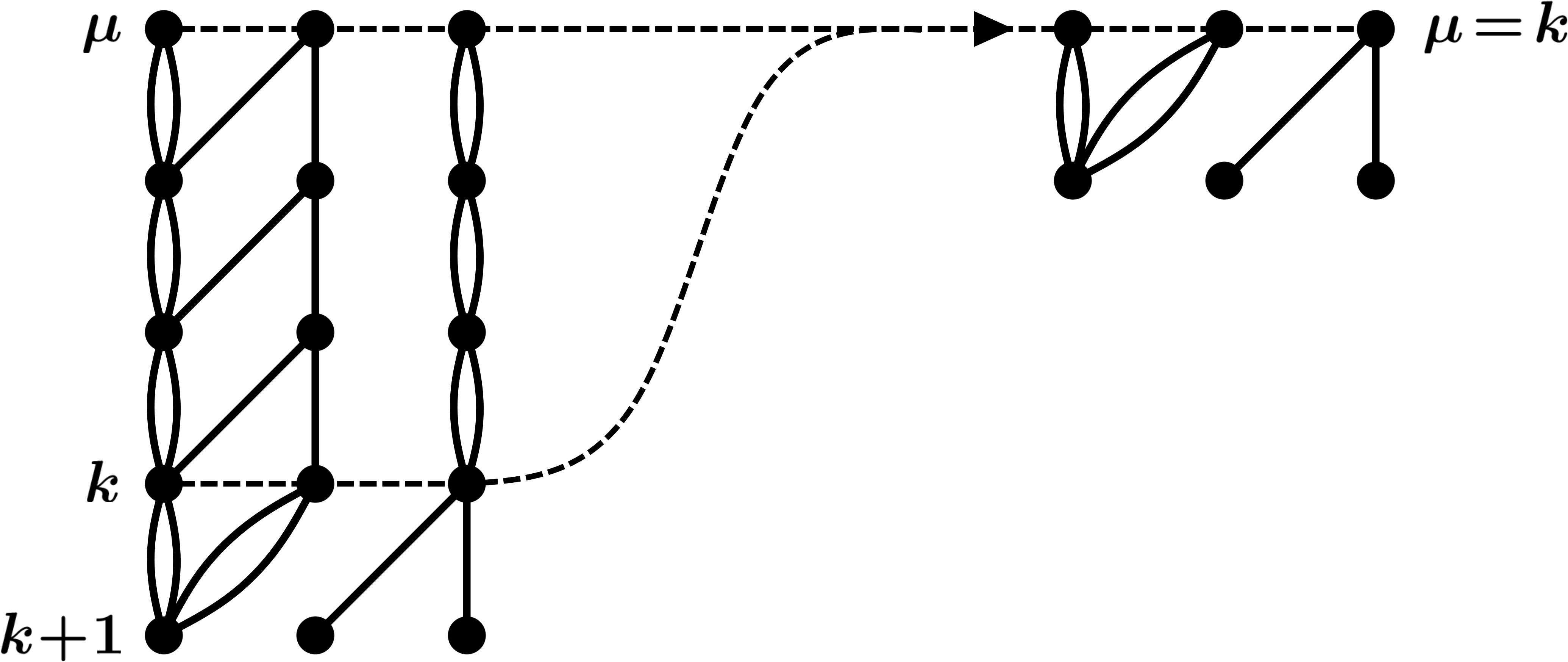}
\caption{Merging level $k$ with $\mu$ ($k>\mu$) while preserving the function of $P$.}
\label{contr}\vspace*{-2mm}
\end{figure}

For $k=\mu$, on the other hand, we know that $\mu=\mu'>3$ from the
definition of $\mu$ and we will first observe that there are at least two edges leading to $v_3^{(\mu)}$. Suppose that only one edge leads to $v_3^{(\mu)}$ from $u\in\{v_1^{(\mu-1)},v_2^{(\mu-1)},v_3^{(\mu-1)}\}$.
If $\mathbf{a}^{(m)}\not\in M(u)$, then
$\mathbf{a}^{(m)}\in M(v_1^{(\mu)})\cup M(v_2^{(\mu)})=M(v_1^{(\mu+1)})\subseteq
M(v_1^{(m)})\cup M(v_2^{(m)})$ implying $P(\mathbf{a}^{(m)})=1$ according to
$m$-condition~\ref{mc4}. If $\mathbf{a}^{(m)}\in M(u)$, then a 1-neighbor
$\mathbf{a}'\in\Omega_1(\mathbf{a}^{(m)})\cap M(u)\subseteq H$ of $\mathbf{a}^{(m)}$
exists which differs from $\mathbf{a}^{(m)}$ in the variable that is tested
at $u$ and thus $\mathbf{a}'\in M(v_1^{(\mu)})\cup M(v_2^{(\mu)})$ implying
$P(\mathbf{a}')=1$. Now, with the two edges leading to $v_3^{(\mu)}$, we could
split $v_3^{(\mu)}$ into two nodes and merge $v_1^{(\mu)}$ and $v_2^{(\mu)}$
while preserving the function of $P$. Thus, for $t_{12}^{(\mu+1)}=1$ we can construct an equivalent branching program with $t_{12}^{(\mu+1)}=0$.
\item[(iv)]
For $t_{13}^{(k+1)}>0$ we know $v=v_3^{(k)}$ and $v'=v_2^{(k)}$ and the four
cases listed in the proposition are obtained when the choice of
$u\in\{v_2^{(k-1)},v_3^{(k-1)}\}$ is combined with whether the second edge
from $u$ leads to $v_1^{(k)}$ or $v'$. In addition, the remaining part
for case~1 and~2 follows from (iii) when $k+1$ is replaced
with $k$. In particular, we know $t_{12}^{(k)}>0$ in case~1, while in case~2
there is a switching path from $v_2^{(k-1)}$ to $v_1^{(k+1)}$ via $v_3^{(k)}$
(substituting for $t_{12}^{(k)}>0$) and a similar analysis applies to
$v=v_2^{(k-1)}$ excluding two switching paths starting from
$v_2^{(k-2)}$ and $v_3^{(k-2)}$, respectively.
\end{enumerate}
\end{proof}

\section{Definition of partition class $R$ and sets $Q_1,\ldots,Q_q$}
\label{dfpart}

\subsection{The block structure from $\mu$ to $\nu$ (definition of $R$)}
\label{dfpartR}

In the following corollary, we summarize the block structure from level
$\mu$ through \textbf{level}~$\boldsymbol\nu$ by using Lemma~\ref{aboveedge},
where $\nu$ is the greatest level such that $\mu\leq\nu\leq m$ and
$t_{12}^{(\ell)}+t_{13}^{(\ell)}>0$ for every $\ell=\mu+1,\ldots,\nu$.
Note that $\nu=\mu$ for $t_{12}^{(\mu+1)}=t_{13}^{(\mu+1)}=0$.
In addition, let \textbf{level}~$\boldsymbol\gamma$ be the greatest
level such that $\mu\leq\gamma\leq\nu$ and $t_{12}^{(\gamma)}>0$ if such $\gamma$
exists, otherwise set $\gamma=\mu$.
\begin{corollary}~
\label{belowmu}
\begin{enumerate}
\item\label{podm1}
$t_{11}^{(\ell)}=t_{33}^{(\ell)}=1$ and $t_{12}^{(\ell)}=t_{22}^{(\ell)}=\frac{1}{2}$
for $\ell=\mu+1,\ldots,\gamma-1$ (Lemma~\ref{aboveedge}.iii),
\item
$t_{11}^{(\gamma)}=t_{23}^{(\gamma)}=1$ and
$t_{12}^{(\gamma)}=t_{32}^{(\gamma)}=\frac{1}{2}$ if $\mu<\gamma<\nu$ (case~1 of
Lemma~\ref{aboveedge}.iv),
\item\label{podm3}
$t_{11}^{(\ell)}=t_{22}^{(\ell)}=1$ and $t_{33}^{(\ell)}=\frac{1}{2}$
for $\ell=\gamma+1,\ldots,\nu-1$ (case~3 of Lemma~\ref{aboveedge}.iv),
\item\label{podm3-4}
if $\nu>\mu$, then $t_{12}^{(\nu)}<1$ (Lemma~\ref{aboveedge}.iii) and
$t_{13}^{(\nu)}<1$ for $\nu<m$ (similarly to $t_{12}^{(\nu)}<1$),
\item\label{podm4}
$t_{12}^{(\ell)}=0$ for $\ell=\nu+1,\ldots,m$ (Lemma~\ref{aboveedge}.iii).
\end{enumerate}
\end{corollary}
\begin{figure}[h]
\vspace*{-2mm}
\centering
\includegraphics[width=11.2cm]{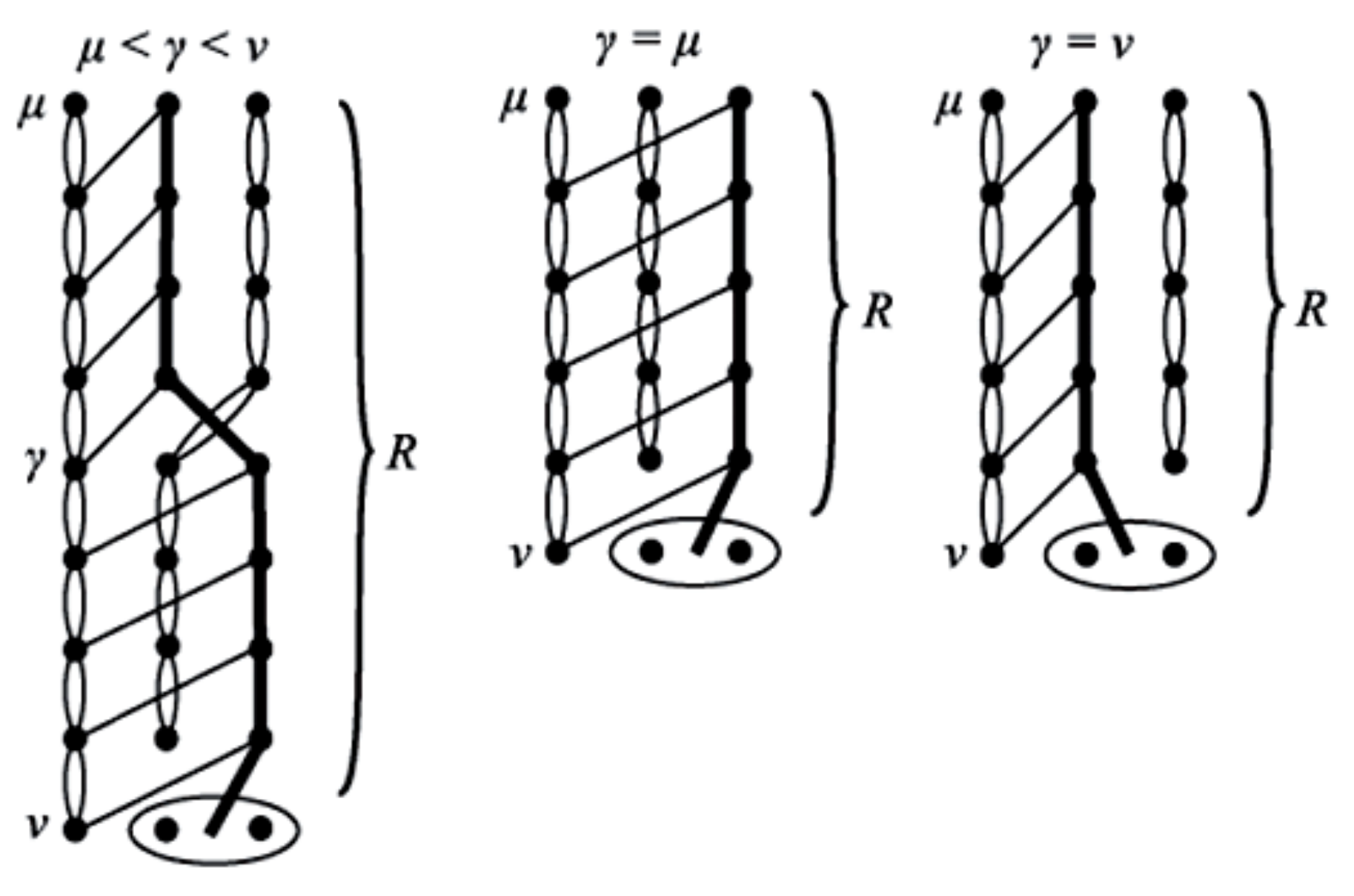}\vspace*{-2mm}
\caption{The block structure from level $\mu$ through $\nu<m$ according
to Corollary~\ref{belowmu}.}
\label{mutonu}
\end{figure}

Figure~\ref{mutonu} shows a typical structure of the block from level
$\mu$ through $\nu$ for the case of $\nu<m$, which comes out of
Corollary~\ref{belowmu}. In particular, there are two disjoint double-edge paths starting at level $\mu$. One follows the first column from $v_1^{(\mu)}$ through $v_1^{(\nu)}$. For $\mu<\gamma<\nu$, the other double-edge path starts at $v_3^{(\mu)}$, follows the third column and turns to $v_2^{(\gamma)}$ on level $\gamma$, and further continues through the second column up to $v_2^{(\nu-1)}$. For $\gamma=\mu$, this double-edge path follows only the second column leading from $v_2^{(\mu)}$ through $v_2^{(\nu-1)}$, whereas for $\gamma=\nu$, it follows the third column from $v_3^{(\mu)}$ through $v_3^{(\nu-1)}$. In addition, there is a node left on each level from $\mu$ through $\nu-1$ that does not lay on the underlying two disjoint double-edge paths. These remaining nodes are connected in a single-edge path from level $\mu$ through $\nu-1$ extended with an edge to $v_2^{(\nu)}$ or $v_3^{(\nu)}$. For each node on this single-edge path the other outgoing edge leads to the double-edge path in the first column.

Furthermore, we shortly analyze level $m$ for the special case of $\nu=m$ as depicted in Figure~\ref{nuem}. Recall that $t_{11}^{(m)}=t_{21}^{(m)}=\frac{1}{2}$ and $t_{32}^{(m)}>0$
by $m$-condition~\ref{mc1} and~\ref{mc2}, respectively. Moreover, either $t_{12}^{(m)}=\frac{1}{2}$ (i.e.\ $\nu=\gamma$) or $t_{13}^{(m)}>0$ (i.e.\ $\nu>\gamma$) by the definition of $\nu$. It follows from Lemma~\ref{aboveedge}.ii for $k=m-1$ that either $t_{33}^{(m)}=1$ for $\nu=\gamma$ or $t_{32}^{(m)}=1$ for $\nu>\gamma$, since otherwise either $t_{13}^{(m)}+t_{23}^{(m)}>0$ and $t_{12}^{(m)}=\frac{1}{2}$ for $\nu=\gamma$, or $t_{12}^{(m)}+t_{22}^{(m)}>0$ and $t_{13}^{(m)}>0$ for $\nu>\gamma$, would provide two switching paths that coincide with two edges from $v_2^{(m-1)}$ and $v_3^{(m-1)}$, respectively, leading to $v_1^{(m)}$ or $v_2^{(m)}$. In the latter case of $\nu>\gamma$, the other edge from $v_3^{(m-1)}$ may lead either to $v_3^{(m)}$ (i.e.\ $t_{13}^{(m)}=t_{33}^{(m)}=\frac{1}{2}$) or to $v_1^{(m)}$ (i.e.\ $t_{13}^{(m)}=1$) or $v_2^{(m)}$ (i.e.\ $t_{13}^{(m)}=t_{23}^{(m)}=\frac{1}{2}$). This completes the analysis of level $\nu=m$. We say that the underlying block is an \emph{empty block} if $\nu=m$ and $t_{33}^{(m)}=0$ (i.e.\ $t_{13}^{(m)}+t_{23}^{(m)}=1$ and $t_{32}^{(m)}=1$).
\begin{figure}[h]
\vspace*{-4mm}
\centering
\includegraphics[width=8.2cm]{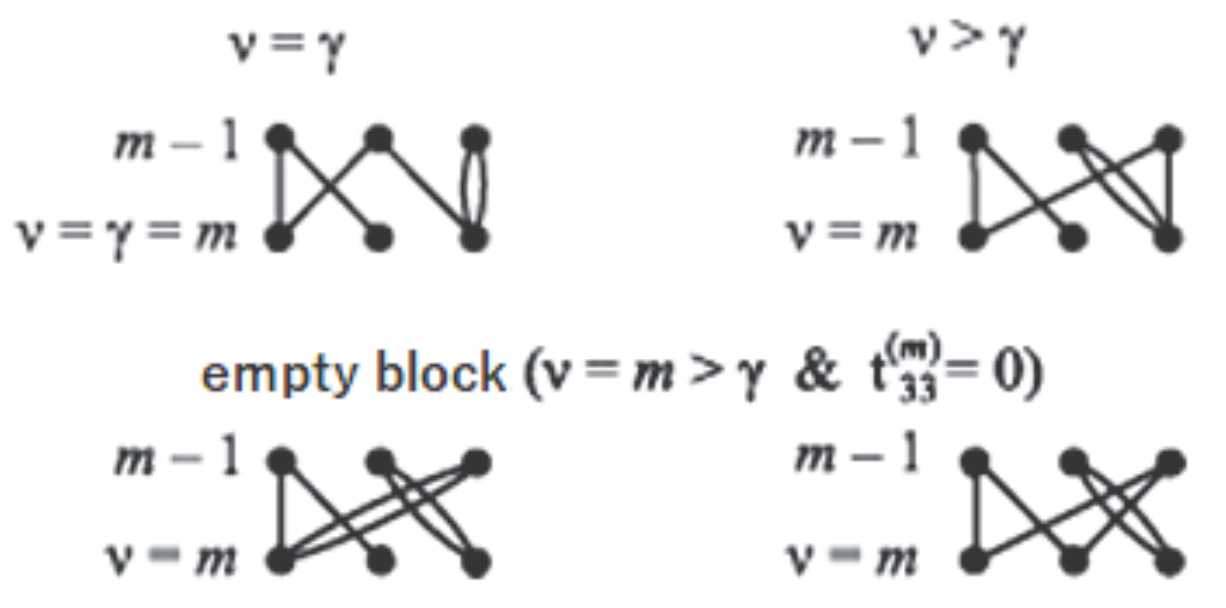}\vspace*{-2mm}
\caption{Level $m$ for $\nu=m$.}
\label{nuem}\vspace*{-1mm}
\end{figure}

Corollary~\ref{belowmu} will be used for the definition of partition class $R$ associated
with the current block, if this block is not empty, which is illustrated in Figure~\ref{mutonu}. Moreover, class $R$ is neither defined for $\nu=\mu$ when only sets
$Q_1,\ldots,Q_q$ are associated with the block (see Paragraph~\ref{dfpartQ} and
Lemma~\ref{nRyQ} in particular). Thus, for a non-empty block and $\nu>\mu$, we define the partition class $R$ to be a set of indices of the variables that are tested on the single-edge computational path $v_2^{(\mu)},v_2^{(\mu+1)},\ldots,v_2^{(\gamma-1)},$ $v_3^{(\gamma)},v_3^{(\gamma+1)},\ldots,v_3^{(\nu'-1)}$ (or $v_3^{(\mu)},v_3^{(\mu+1)},\ldots,v_3^{(\nu'-1)}$ if $\gamma=\mu$
or $v_2^{(\mu)},v_2^{(\mu+1)},\ldots,v_2^{(\nu'-1)}$ if
$\gamma=\nu$) where \textbf{level}~$\boldsymbol\nu'$ is defined as
\begin{equation}
\label{defnup}
\nu'=\min(\nu,m-1)\,.
\end{equation}
For the future use of condition (\ref{cond}) we also define relevant bits of
string $\mathbf{c}\in\{0,1\}^n$. Thus, let $c_i^R$
be the corresponding labels of the edges creating this computational path
(indicated by a bold line in Figure~\ref{mutonu}) including the edge outgoing
from the last node $v_3^{(\nu'-1)}$ (or $v_2^{(\nu'-1)}$ if $\gamma=\nu$) that
leads to \mbox{$v_2^{(\nu')}$ or $v_3^{(\nu')}$.}

\subsection{The block structure from $\omega$ to $m$
(definition of $Q_1,\ldots,Q_q$)}
\label{dfpartQ}

Furthermore, we define \textbf{level}~$\boldsymbol\omega$ to be the
greatest level such that $\mu\leq\omega\leq m$ and
there is a double-edge path from $\mu$ through $\omega$ containing only
nodes $v_\ell\in\{v_2^{(\ell)},v_3^{(\ell)}\}$ for every
$\ell=\mu,\ldots,\omega$. Note that this path possibly extends the
double-edge path from Corollary~\ref{belowmu} (see Figure~\ref{mutonu})
leading from $v_2^{(\mu)}$ to $v_2^{(\nu-1)}$ (for $\gamma=\mu<\nu$) or from
$v_3^{(\mu)}$ to $v_2^{(\nu-1)}$ (for $\mu<\gamma<\nu$) or from $v_3^{(\mu)}$
to $v_3^{(\nu-1)}$ (for $\gamma=\nu>\mu$). Hence,
\begin{equation}
\label{omnu}
\omega\geq\max(\nu-1,\mu)\,.
\end{equation}
For the special case of $\omega=m$ (including the empty block) when this
double-edge path reaches level $m$, no sets $Q_1,\ldots,Q_q$ are associated with
the current block and we set $q=0$. In this case, we will observe in the following
lemma that $\nu>\mu$, which ensures that at least class $R$ is defined
for a non-empty block (Paragraph~\ref{dfpartR}) when $\omega=m$.
\begin{lemma}
\label{nRyQ}
If $\omega=m$, then $\nu>\mu$.
\end{lemma}
\begin{proof}
On the contrary, suppose $\omega=m$ and $\nu=\mu$.
It follows from Corollary~\ref{belowmu}.5 that $t_{12}^{(m)}=0$.
Moreover, $t_{22}^{(m)}=0$ since $t_{22}^{(m)}>0$ would require
$t_{13}^{(m)}>0$ by the normalization of $P$, which contradicts
Lemma~\ref{aboveedge}.ii, and hence, $t_{32}^{(m)}=1$. If
$t_{21}^{(m-1)}+t_{31}^{(m-1)}>0$, then $p_3^{(m)}\geq p_2^{(m-1)}
\geq p_1^{(m-2)}/2>\frac{1}{6}$ due to (\ref{phi1}) which is
in contradiction to $m$-condition~\ref{mc3}. Hence, $t_{11}^{(m-1)}=1$
which means $\mu'<m-1$. Furthermore, $t_{12}^{(\mu+1)}=t_{13}^{(\mu+1)}=0$
by the definition of $\nu$ implying $\mu=\mu'$. Since $P$ is normalized,
we know $t_{22}^{(\mu+1)}>0$ and either $t_{22}^{(\mu+1)}=1$ or
$t_{23}^{(\mu+1)}=1$ due to $\omega=m$, which implies $t_{22}^{(\ell)}=1$
for $\ell=\mu+2,\ldots,m-1$. It follows that
$p_2^{(\mu+1)}\leq p_3^{(m)}<\frac{1}{6}$ according to
$m$-condition~\ref{mc3}.

On the other hand, we know
$t_{21}^{(\mu)}+t_{31}^{(\mu)}>0$ by the definition of $\mu'$, which
implies $p_2^{(\mu)}>\frac{1}{6}$ due to (\ref{phi1}).
Hence, $t_{22}^{(\mu+1)}=t_{32}^{(\mu+1)}=\frac{1}{2}$ and
$t_{23}^{(\mu+1)}=1$ because of $p_2^{(\mu+1)}<\frac{1}{6}$. This
ensures $t_{11}^{(\mu)}=t_{21}^{(\mu)}=\frac{1}{2}$
since $p_1^{(\mu-1)}>\frac{1}{3}$. Thus,
$\frac{1}{6}>p_2^{(\mu+1)}>p_1^{(\mu-1)}/4$ which rewrites as
$p_1^{(\mu-1)}<\frac{2}{3}$ implying
$p_2^{(\mu-1)}+p_3^{(\mu-1)}>\frac{1}{3}$, and hence
$p_2^{(\mu-1)}>\frac{1}{6}$ due to $p_2^{(\mu-1)}\geq p_3^{(\mu-1)}$.
Clearly, $t_{32}^{(\mu)}=0$ since otherwise we get
a contradiction $\frac{1}{6}>p_2^{(\mu+1)}\geq
\frac{1}{4}p_1^{(\mu-1)}+\frac{1}{2}p_2^{(\mu-1)}>
\frac{1}{4}\cdot\frac{1}{3}+\frac{1}{2}\cdot\frac{1}{6}=\frac{1}{6}$.
Similarly, $t_{22}^{(\mu)}=1$ produces a contradiction $\frac{1}{6}>
p_2^{(\mu+1)}\geq\frac{1}{4}p_1^{(\mu-1)}+\frac{1}{2}p_2^{(\mu-1)}
>\frac{1}{6}$. It follows that $t_{12}^{(\mu)}>0$ whereas
$t_{12}^{(\mu)}=1$ contradicts $\mu=\mu'$ according to (\ref{defmu}),
and hence $t_{12}^{(\mu)}=t_{22}^{(\mu)}=\frac{1}{2}$ and
$t_{33}^{(\mu)}>0$. This gives a contradiction $\frac{1}{6}>
p_2^{(\mu+1)}\geq\frac{1}{4}(p_1^{(\mu-1)}+p_2^{(\mu-1)})
+\frac{1}{2}p_3^{(\mu-1)}>\frac{1}{4}(p_1^{(\mu-1)}
+p_2^{(\mu-1)}+p_3^{(\mu-1)})=\frac{1}{4}$.
\end{proof}

Hereafter in this Section~\ref{dfpart}, we will consider only the case of $\mathbf{\boldsymbol\omega<m}$ which allows the definition of sets $Q_1,\ldots,Q_q$, while in the following sections, the case of $\omega=m$ is also taken into account. This implies $t_{12}^{(m)}=0$ since otherwise $t_{12}^{(m)}=t_{32}^{(m)}=\frac{1}{2}$ ($m$-condition~\ref{mc2}) forces $t_{33}^{(m)}=1$ by Lemma~\ref{aboveedge}.ii which would prolong the double-edge path from $v_3^{(\mu)}$ up to $v_3^{(m)}$ according to Lemma~\ref{aboveedge}.iii.

We will show that one can assume $t_{13}^{(m)}>0$
without loss of generality. Suppose that $t_{13}^{(m)}=0$, which implies $t_{22}^{(m)}=t_{23}^{(m)}=0$ due to $P$ is normalized, and hence $t_{32}^{(m)}=t_{33}^{(m)}=1$. Moreover, we know $t_{11}^{(m)}=t_{21}^{(m)}=\frac{1}{2}$ by $m$-condition~\ref{mc1} and $m$-condition~\ref{mc3} ensures $t_{11}^{(m-1)}=1$.
If~$t_{12}^{(m-1)}=t_{13}^{(m-1)}=0$, then $v_2^{(m-1)}$ and $v_3^{(m-1)}$
can be merged and replaced by $v_3^{(m)}$, while $v_1^{(m-1)}$ replaces
$v_1^{(m-2)}$, which shortens $P$ without changing its function. Hence,
either $t_{12}^{(m-1)}>0$ or $t_{13}^{(m-1)}>0$ by Lemma~\ref{aboveedge}.ii.
In fact, $t_{12}^{(m-1)}>0$ contradicts $\omega<m$ according to
Lemma~\ref{aboveedge}.iii since
\linebreak
$t_{23}^{(m-1)}+t_{33}^{(m-1)}=t_{32}^{(m)}=t_{33}^{(m)}=1$
can, without loss of generality, prolong the double-edge path from $v_3^{(\mu)}$
through $v_3^{(m-2)}$ up to $v_3^{(m)}$. For $t_{13}^{(m-1)}>0$, on the other hand,
$v_2^{(m-1)}$ and $v_3^{(m-1)}$ can be merged while
$v_1^{(m-1)}$ is split into two its copies, which produces
$t_{11}^{(m-1)}=t_{21}^{(m-1)}=\frac{1}{2}$, $t_{32}^{(m-1)}=1$, and
$t_{11}^{(m)}=t_{21}^{(m)}=t_{12}^{(m)}=t_{22}^{(m)}=\frac{1}{2}$, $t_{33}^{(m)}=1$.
After this modification, level $m-1$ satisfies the four
$m$-conditions~\ref{mc1}--\ref{mc4} (see Paragraph~\ref{prplan}) and thus, it
can serve as a new level $m$ while the original level $m>d$ (for $m=d$ program $P$
could be shortened by removing its last level) is included in the previous greater-level
neighboring block, which is consistent with its structure (see
Paragraph~\ref{cond1-3} and Figure~\ref{fmptomu} in particular).

Thus, we assume
$t_{13}^{(m)}>0$ without loss of generality, which implies $t_{32}^{(m)}=1$ by
Lemma~\ref{aboveedge}.ii and $t_{11}^{(m-1)}=1$ according to
$m$-condition~\ref{mc3}. Then Lemma~\ref{aboveedge}.iv can be
employed for $k=m-1$ where only case~3 and~4 may occur due
to $\omega<m$ is assumed, which even implies $\omega<m-1$. In case~3,
$t_{13}^{(m-1)}>0$ and Lemma~\ref{aboveedge}.iv can again be applied recursively
to $k=m-2$ etc.

In general, we start with \textbf{level}~$\mathbf{\boldsymbol\sigma_1}\!=\!m$ that meets
$t_{13}^{(\sigma_j)}\!>\!0$ for $j=1$. We proceed to smaller~levels and inspect
recursively the structure of subblocks indexed as $j$ from \textbf{level}~$\mathbf{\boldsymbol\lambda_j}$ through $\sigma_j$ where $\lambda_j$
is the least level such that $\omega\leq\lambda_j<\sigma_j-1$ and the
transitions from case~3 or~4 of Lemma~\ref{aboveedge}.iv occur for all
levels $\ell\!=\!\lambda_j+1,\ldots,\sigma_j-1$ as depicted
in Figure~\ref{nutom}. This means $t_{11}^{(\ell)}\!=\!t_{22}^{(\ell)}\!=\!1$ and $t_{33}^{(\ell)}\!=\!\frac{1}{2}$ for every $\ell=\lambda_j+1,\ldots,\sigma_j-1$.
Note that $\lambda_j>\mu$ because $\lambda_j=\mu$ ensures $t_{22}^{(\mu+1)}=1$
implying $\omega>\mu=\lambda_j$ by the definition of $\omega$, which contradicts
$\omega\leq\lambda_j$. In addition, we will observe that case~4 from
Lemma~\ref{aboveedge}.iv occurs at level $\lambda_j+1$, that is
$t_{23}^{(\lambda_j+1)}=\frac{1}{2}$. On the contrary, suppose that
$t_{13}^{(\lambda_j+1)}=\frac{1}{2}$ (case~3). For $\lambda_j>\omega$,
this means case~1 or~2 occurs at level $\lambda_j>\mu$ by the
definition of $\lambda_j$, which would be in contradiction to
$\omega\leq\lambda_j$ according to Lemma~\ref{aboveedge}.iv. For
$\lambda_j=\omega$, on the other hand, $t_{13}^{(\omega+1)}=\frac{1}{2}$
contradicts the definition of $\omega$ by Lemma~\ref{aboveedge}.iv.
This completes the argument for $t_{23}^{(\lambda_j+1)}=\frac{1}{2}$.

Furthermore, let \textbf{level}~${\mathbf{\boldsymbol\kappa_j}}$
be the least level such that $\lambda_j+1<\kappa_j\leq \sigma_j$ and
$t_{13}^{(\kappa_j)}>0$, which exists since at least $t_{13}^{(\sigma_j)}>0$.
Now we can define $Q_j$ associated with the current block
(a candidate for Q in the richness condition (\ref{cond})) to be a set of
indices of the variables that are tested on the computational path
$v_3^{(\lambda_j)},v_3^{(\lambda_j+1)},\ldots,v_3^{(\kappa_j-1)}$, and let
$c_i^{Q_j}$ be the corresponding labels of the edges creating this path
including the edge from $v_3^{(\kappa_j-1)}$ to $v_1^{(\kappa_j)}$
(indicated by a bold line in Figure~\ref{nutom}). This extends
the definition of $\mathbf{c}\in\{0,1\}^n$ associated with $R$ and $Q_k$
for $1\leq k<j$, which are usually pairwise disjoint due to $P$ is read-once.
Nevertheless, the definition of $\mathbf{c}$ may not be unique for indices
from their nonempty intersections in some very special cases (including those
corresponding to neighboring blocks) but the richness condition will only be
used for provably disjoint sets (see Section~\ref{recursionm}).

\medskip
Finally, define next
\textbf{level}~${\mathbf{\boldsymbol\sigma_{j+1}}}$ to be the greatest level such
that $\omega+1<\sigma_{j+1}\leq\lambda_j$ and $t_{13}^{(\sigma_{j+1})}>0$, if such
$\sigma_{j+1}$ exists, and continue in the recursive definition of
$\lambda_{j+1},\kappa_{j+1},Q_{j+1}$ with $j$ replaced by $j+1$ etc. If such
$\sigma_{j+1}$ does not exist, then set $q=j$ and the definition
of sets $Q_1,\ldots,Q_q$ associated with the current block is complete.

\begin{figure}[ht]
\centering
\includegraphics[height=12.9cm]{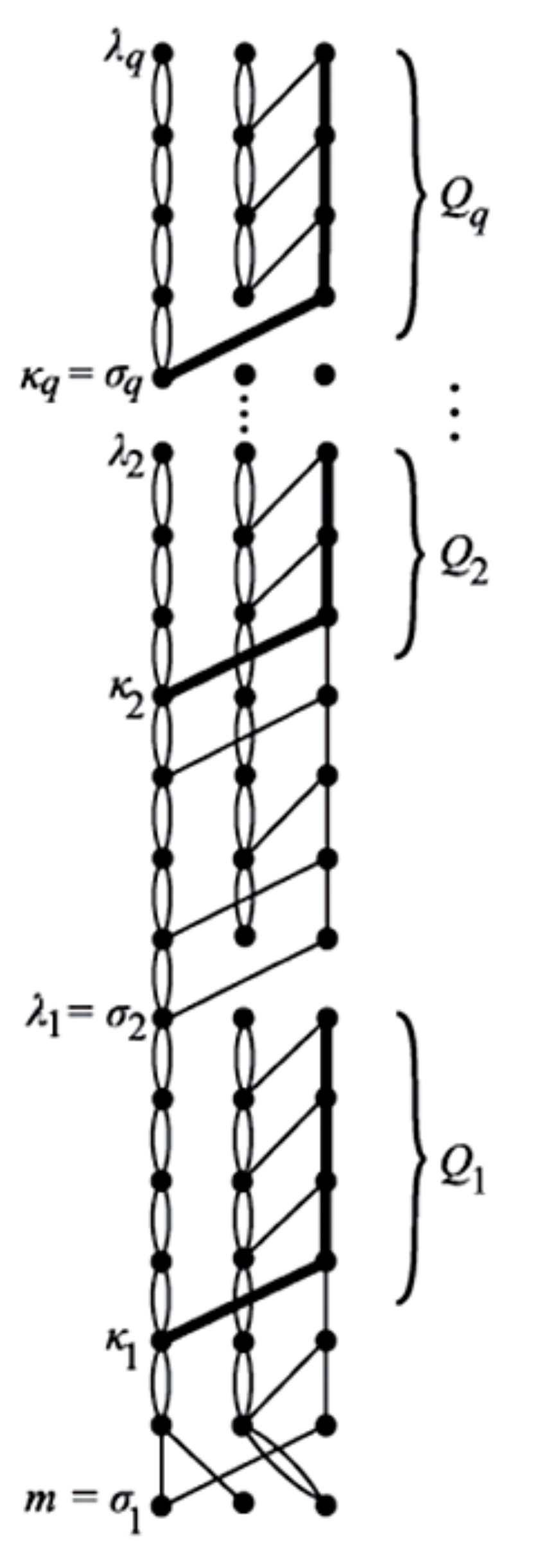}\vspace*{-2mm}
\caption{The definition of $Q_1,\ldots,Q_q$.}
\label{nutom}\vspace*{-3mm}
\end{figure}

\subsection{An upper bound on $p_1^{(m)}+p_2^{(m)}$ in terms of $p_1^{(\omega+1)}$}
\label{ubp1mp2mp1o}

In this paragraph, we will upper-bound $p_1^{(m)}+p_2^{(m)}$ in terms of
$p_1^{(\omega+1)}$ which will later be used for verifying the condition (\ref{acond}).
For any $1\leq j\leq q$, we know that
$t_{11}^{(\ell)}=t_{22}^{(\ell)}=1$ and
$t_{23}^{(\ell)}=t_{33}^{(\ell)}=\frac{1}{2}$ for every
$\ell=\lambda_j+1,\ldots,\kappa_j-1$ (see Figure~\ref{nutom}), which gives
\begin{eqnarray}
\label{kaplam}
p_2^{(\kappa_j-1)}+p_3^{(\kappa_j-1)}&=&
p_2^{(\lambda_j)}+p_3^{(\lambda_j)}\,,\\
\label{p3kapj1}
p_3^{(\kappa_j-1)}=\frac{p_3^{(\lambda_j)}}{2^{|Q_j|-1}}&\leq&
\frac{p_2^{(\lambda_j)}+p_3^{(\lambda_j)}}{2^{|Q_j|}}
\end{eqnarray}
because $p_3^{(\lambda_j)}\leq\frac{1}{2}(p_2^{(\lambda_j)}+p_3^{(\lambda_j)})$.
We know $t_{12}^{(\ell)}=0$ for every $\ell=\omega+2,\ldots,m$ according to
Corollary~\ref{belowmu}.\ref{podm4} where $\nu+1\leq\omega+2$ by (\ref{omnu}).
Moreover, it follows from the definition of
$\sigma_{j+1}>\omega+1$ that $t_{13}^{(\ell)}=0$ for every
$\ell=\sigma_{j+1}+1,\ldots,\lambda_j$ for any $1\leq j<q$, and
$t_{13}^{(\ell)}=0$ for every $\ell=\omega+2,\ldots,\lambda_q$. Hence,
\begin{equation}
\label{lammj}
p_2^{(\sigma_{j+1})}+p_3^{(\sigma_{j+1})}=p_2^{(\lambda_j)}+p_3^{(\lambda_j)}
=p_2^{(\kappa_j-1)}+p_3^{(\kappa_j-1)}
\end{equation}
for $1\leq j<q$ and
\begin{equation}
\label{lamom}
p_2^{(\omega+1)}+p_3^{(\omega+1)}=p_2^{(\lambda_q)}+p_3^{(\lambda_q)}
=p_2^{(\kappa_q-1)}+p_3^{(\kappa_q-1)}
\end{equation}
according to (\ref{kaplam}). Note that equation (\ref{lamom}) holds trivially
for $\lambda_q=\omega+1$ and it is also valid for
$\lambda_q=\omega$ (recall $\lambda_q\geq\omega$ from the definition of
$\lambda_j$) because $t_{11}^{(\lambda_q+1)}=t_{22}^{(\lambda_q+1)}=1$
and $t_{23}^{(\lambda_q+1)}=t_{33}^{(\lambda_q+1)}=\frac{1}{2}$
(case~4 of Lemma~\ref{aboveedge}.iv).
Furthermore, we know $t_{22}^{(\ell)}=1$ for every
$\ell=\kappa_j,\ldots,\sigma_j-1$ and $t_{12}^{(\sigma_j)}=0$, which
implies
\begin{eqnarray}
\label{p2p3mj1}
p_2^{(\sigma_j)}+p_3^{(\sigma_j)}&\geq&
p_2^{(\kappa_j-1)}+p_3^{(\kappa_j-1)}-p_3^{(\kappa_j-1)}\nonumber\\
&\geq&p_2^{(\kappa_j-1)}+p_3^{(\kappa_j-1)}
-\frac{p_2^{(\lambda_j)}+p_3^{(\lambda_j)}}{2^{|Q_j|}}\nonumber\\
&=&\left(p_2^{(\sigma_{j+1})}+p_3^{(\sigma_{j+1})}\right)
\left(1-\frac{1}{2^{|Q_j|}}\right)
\end{eqnarray}
for $1<j<q$ according to (\ref{p3kapj1}) and (\ref{lammj}), while formula
(\ref{p2p3mj1}) reads
\begin{equation}
\label{p2p3m2}
p_3^{(m)}=p_3^{(\sigma_1)}\geq\left(p_2^{(\sigma_2)}+p_3^{(\sigma_2)}\right)
\left(1-\frac{1}{2^{|Q_1|}}\right)
\end{equation}
for $j=1<q$ due to $t_{32}^{(m)}=1$, whereas (\ref{p2p3mj1}) is rewritten as
\begin{equation}
\label{p2p3lamp}
p_2^{(\sigma_q)}+p_3^{(\sigma_q)}\geq
\left(p_2^{(\omega+1)}+p_3^{(\omega+1)}\right)
\left(1-\frac{1}{2^{|Q_q|}}\right)
\end{equation}
for $j=q>1$ according to (\ref{lamom}). Thus starting with (\ref{p2p3m2}),
inequality (\ref{p2p3mj1}) is applied recursively for $j=2,\ldots,q-1$,
and, in the end, formula (\ref{p2p3lamp}) is employed, leading to
\begin{equation}
\label{p3mr1}
p_3^{(m)}\geq\left(p_2^{(\omega+1)}+p_3^{(\omega+1)}\right)
\prod_{j=1}^q\left(1-\frac{1}{2^{|Q_j|}}\right)
\end{equation}
which is also obviously valid for the special case of $q=1$. This can be
rewritten as
\begin{equation}
\label{p1p2m}
p_1^{(m)}+p_2^{(m)}\leq 1-\left(1-p_1^{(\omega+1)}\right)
\prod_{j=1}^q\left(1-\frac{1}{2^{|Q_j|}}\right)
\end{equation}
which represents the desired upper bound on $p_1^{(m)}+p_2^{(m)}$ in terms of
$p_1^{(\omega+1)}$.

\section{The conditional block structure before level $\mu$}
\label{sbelowmu}

\subsection{Assumptions and level $\mu+1$}

Throughout this Section~\ref{sbelowmu}, we will assume
\begin{eqnarray}
\label{p3mu}
p_3^{(\mu)}&<&\frac{1}{6}\,,\\
\label{pik45}
\prod_{j=1}^q\left(1-\frac{1}{2^{|Q_j|}}\right)&>&\frac{2}{3}
\end{eqnarray}
where the product in (\ref{pik45}) equals 1 for $q=0$. Based on these
assumption, we will further analyze the block structure
before level $\mu$ in order
to satisfy the $m'$-conditions~\mbox{\ref{mc1}--\ref{mc4}}
(see Paragraph~\ref{prplan}) also for the first block level $m'$ (the formal
definition of $m'$ appears at the beginning of Paragraph~\ref{cond1-3})
so that the analysis can be applied recursively when inequalities
(\ref{p3mu}) and (\ref{pik45}) hold (Section~\ref{recursionm}). For this purpose,
we still analyze level $\mu+1$ in the following lemma which implies $\nu>\mu$ and
thus guarantees that partition class $R$ is defined for the underlying block if not
empty.
\begin{lemma}
\label{levelmu}
$t_{12}^{(\mu+1)}=\frac{1}{2}$.
\end{lemma}
\begin{proof}
For $\mu=\mu'-1$, the proposition follows from the definition of $\mu$,
and thus assume $\mu=\mu'$. Consider first the special case of $\mu+1=m$
and on the contrary suppose $t_{12}^{(m)}=0$ which implies $t_{32}^{(m)}=1$ by using $m$-condition~\ref{mc2}, Lemma~\ref{aboveedge}.ii, and the normalization of $P$. It follows from $t_{32}^{(m)}=1$ that $p_3^{(m)}\geq p_2^{(m-1)}$. We have $t_{11}^{(m-1)}<1$ from the definition of $\mu'$, which means either $t_{21}^{(m-1)}>0$ implying $p_2^{(m-1)}\geq\frac{1}{2}p_1^{(m-2)}$, or
$t_{31}^{(m-1)}>0$ giving the same $p_2^{(m-1)}\geq p_3^{(m-1)}\geq\frac{1}{2}p_1^{(m-2)}$ by the normalization of $P$. Moreover, we know $p_3^{(m)}<\frac{1}{6}$ from $m$-condition~\ref{mc3}, and $\frac{1}{2}p_1^{(m-2)}>\frac{1}{6}$ due to (\ref{phi1}). Altogether, we get the contradiction $\frac{1}{6}>p_3^{(m)}\geq
p_2^{(m-1)}\geq\frac{1}{2}p_1^{(m-2)}>\frac{1}{6}$. Thus further assume $\mu<m-1$. Clearly,
$t_{32}^{(\mu+1)}<1$ by the normalization of $P$ whereas $t_{33}^{(\mu+1)}=1$
implies $t_{12}^{(\mu+1)}=\frac{1}{2}$, and thus, further consider the case
when no double edge leads to $v_3^{(\mu+1)}$. If $t_{12}^{(\mu+1)}>0$,
then $t_{12}^{(\mu+1)}=\frac{1}{2}$ by Lemma~\ref{aboveedge}.iii for
$k=\mu$. On the contrary, suppose $t_{12}^{(\mu+1)}=0$, which gives
$t_{22}^{(\mu+1)}>0$ due to $t_{32}^{(\mu+1)}<1$.
Assumption (\ref{p3mu}) ensures $t_{31}^{(\mu)}=0$
which implies $t_{21}^{(\mu)}>0$ by the definition of $\mu'$.

\medskip
We will first show that
\begin{equation}
\label{p2m114}
p_2^{(\mu+1)}<\frac{1}{4}\,.
\end{equation}
For $\omega<m$, assumption (\ref{pik45}) together with $m$-condition~\ref{mc3} ensures
\begin{equation}
\label{p2om1p3om1548}
p_2^{(\omega+1)}+p_3^{(\omega+1)}<\frac{1}{4}
\end{equation}
according to (\ref{p3mr1}), which gives (\ref{p2m114}) for $\omega=\mu$.
For $\omega>\mu$, we know by the definition of $\omega$
that there is a double-edge path starting from $v_2^{(\mu)}$
or $v_3^{(\mu)}$ and traversing $v_2^{(\mu+1)}$ as we assume no double edge
to $v_3^{(\mu+1)}$. For $\omega<m$, we have
$t_{22}^{(\ell)}=1$ for $\ell=\mu+2,\ldots,\omega$, and $t_{12}^{(\omega+1)}=0$
according to Lemma~\ref{aboveedge}.iii, and hence,
$p_2^{(\mu+1)}\leq p_2^{(\omega+1)}+p_3^{(\omega+1)}<\frac{1}{4}$ due to
(\ref{p2om1p3om1548}). Similarly, $p_2^{(\mu+1)}\leq p_3^{(m)}<\frac{1}{6}$
for $\omega=m$ by $m$-condition~\ref{mc3}, which
\mbox{completes the argument for (\ref{p2m114}).}

\medskip
Suppose first that $t_{21}^{(\mu)}=1$, which together with
$p_1^{(\mu-1)}>\frac{1}{3}$ implies $t_{22}^{(\mu+1)}=t_{32}^{(\mu+1)}=\frac{1}{2}$
according to (\ref{p2m114}). Obviously,
$\frac{1}{2}<t_{12}^{(\mu)}+t_{13}^{(\mu)}<2$ by the normalization of $P$.
For $t_{12}^{(\mu)}+t_{13}^{(\mu)}=1$, either $t_{12}^{(\mu)}=t_{33}^{(\mu)}=1$
or $t_{32}^{(\mu)}=t_{13}^{(\mu)}=1$ when $P$ could be shortened without changing
its function, or $t_{12}^{(\mu)}=t_{13}^{(\mu)}=t_{32}^{(\mu)}=t_{33}^{(\mu)}=\frac{1}{2}$
implying $p_1^{(\mu)}=p_2^{(\mu)}=p_3^{(\mu)}=\frac{1}{3}$ which contradicts (\ref{phi1}).
Hence, $t_{12}^{(\mu)}+t_{13}^{(\mu)}=\frac{3}{2}$.
Denote $i\in\{2,3\}$ so that $t_{1i}^{(\mu)}=1$ whereas $j\in\{2,3\}$
satisfies $t_{1j}^{(\mu)}=t_{3j}^{(\mu)}=\frac{1}{2}$.
If $t_{13}^{(\mu+1)}=1$, then we could shorten $P$ while preserving its
function, and hence $t_{23}^{(\mu+1)}>0$ due to $t_{33}^{(\mu+1)}<1$. It follows that
$p_2^{(\mu+1)}\geq\frac{1}{2}p_1^{(\mu-1)}+\frac{1}{4}p_j^{(\mu-1)}=
\frac{1}{4}(2p_1^{(\mu-1)}+p_j^{(\mu-1)})=
\frac{1}{4}(1-p_i^{(\mu-1)}+p_1^{(\mu-1)})\geq\frac{1}{4}$ which
contradicts (\ref{p2m114}). Hence,
$t_{11}^{(\mu)}=t_{21}^{(\mu)}=\frac{1}{2}$ due to $t_{11}^{(\mu)}<1$ and
$t_{31}^{(\mu)}=0$, which implies $t_{12}^{(\mu)}<1$ since $\mu=\mu'$.

In addition, there are no `switching paths' (cf. Lemma~\ref{aboveedge}.ii)
starting simultaneously from all three vertices
$v_1^{(\mu-1)},v_2^{(\mu-1)},v_3^{(\mu-1)}$ and leading to $v_1^{(\mu)}$ or
$v_1^{(\mu+1)}$ since otherwise a 2-neighbor
$\mathbf{a}'\in\Omega_2(\mathbf{a}^{(m)})\cap M(v_i^{(\mu-1)})\subseteq H$ of
$\mathbf{a}^{(m)}\in M(v_i^{(\mu-1)})$ from
$m$-condition~\ref{mc4} would exist for some $i\in\{1,2,3\}$ such that
$\mathbf{a}'\in M(v_1^{(\mu+1)})$ implying $P(\mathbf{a}')=1$.
Recall we still need to contradict $t_{12}^{(\mu+1)}=0$, provided that
$t_{11}^{(\mu)}=t_{21}^{(\mu)}=\frac{1}{2}$,
$t_{12}^{(\mu)}<1$, $t_{11}^{(\mu+1)}=1$, $t_{22}^{(\mu+1)}>0$, and
$t_{33}^{(\mu+1)}<1$.

\medskip
We will first consider the case of $t_{12}^{(\mu)}=\frac{1}{2}$ which
implies $t_{13}^{(\mu)}=0$ since three switching paths starting
from level $\mu-1$ are excluded. Suppose that $t_{33}^{(\mu)}>0$ which also
gives $t_{13}^{(\mu+1)}=0$ because of ruling out the three switching paths,
and hence $t_{23}^{(\mu+1)}>0$ due to $t_{33}^{(\mu+1)}<1$. In addition, we
know $t_{22}^{(\mu)}+t_{32}^{(\mu)}=\frac{1}{2}$ since we assume
$t_{12}^{(\mu)}=\frac{1}{2}$.
It follows that $p_2^{(\mu+1)}\geq\frac{1}{4}p_1^{(\mu-1)}+
\frac{1}{4}p_2^{(\mu-1)}+\frac{1}{4}p_3^{(\mu-1)}=\frac{1}{4}$
which contradicts (\ref{p2m114}). Hence, $t_{33}^{(\mu)}=0$
implying $t_{23}^{(\mu)}=1$ due to $t_{13}^{(\mu)}=0$, which gives
$t_{32}^{(\mu)}=\frac{1}{2}$. For $t_{23}^{(\mu+1)}>0$, we would again get
a contradiction $p_2^{(\mu+1)}\geq\frac{1}{4}p_1^{(\mu-1)}+
\frac{1}{4}p_2^{(\mu-1)}+\frac{1}{2}p_3^{(\mu-1)}\geq\frac{1}{4}$,
and hence we have $t_{23}^{(\mu+1)}=0$ and $t_{13}^{(\mu+1)}>0$ because of
$t_{33}^{(\mu+1)}<1$. We can assume without loss of
generality that $t_{13}^{(\mu+1)}=\frac{1}{2}$
since otherwise $t_{12}^{(\mu)}=t_{32}^{(\mu)}=\frac{1}{2}$ and
$t_{13}^{(\mu+1)}=1$ (implying $t_{22}^{(\mu+1)}=t_{32}^{(\mu+1)}=\frac{1}{2}$)
could be replaced with $t_{12}^{(\mu)}=1$ while $t_{23}^{(\mu)}=1$
is replaced with
$t_{22}^{(\mu)}=t_{32}^{(\mu)}=t_{23}^{(\mu)}=t_{33}^{(\mu)}=\frac{1}{2}$
and $t_{23}^{(\mu+1)}=t_{33}^{(\mu+1)}=\frac{1}{2}$ where $v_3^{(\mu)}$
is a copy of $v_2^{(\mu)}$, which redefines level $\mu$.
Thus, it follows from $t_{13}^{(\mu+1)}=\frac{1}{2}$ and $t_{23}^{(\mu+1)}=0$
that $t_{33}^{(\mu+1)}=\frac{1}{2}$ and $t_{22}^{(\mu+1)}=1$ by the
normalization of $P$.

Recall once more we have
$t_{11}^{(\mu)}=t_{21}^{(\mu)}=t_{12}^{(\mu)}=t_{32}^{(\mu)}=\frac{1}{2}$,
$t_{23}^{(\mu)}=1$, $t_{11}^{(\mu+1)}=t_{22}^{(\mu+1)}=1$, and
$t_{13}^{(\mu+1)}=t_{33}^{(\mu+1)}=\frac{1}{2}$.
We know $p_1^{(\mu-1)}\leq 2p_2^{(\mu+1)}<\frac{1}{2}$
due to (\ref{p2m114}) and $p_2^{(\mu-1)}=2p_3^{(\mu)}<\frac{1}{3}$
by (\ref{p3mu}), which implies
$p_1^{(\mu+1)}=\frac{1}{2}p_1^{(\mu-1)}+\frac{3}{4}p_2^{(\mu-1)}<\frac{1}{2}$.
This gives a contradiction
$p_2^{(\mu+1)}\geq\frac{1}{2}(p_2^{(\mu+1)}+p_3^{(\mu+1)})=
\frac{1}{2}(1-p_1^{(\mu+1)})>\frac{1}{4}$ according to (\ref{p2m114}), which
completes the argument for $t_{12}^{(\mu)}=\frac{1}{2}$.

Further consider the case of $t_{13}^{(\mu)}>0$ which ensures
$t_{12}^{(\mu)}=0$ or equivalently $t_{22}^{(\mu)}+t_{32}^{(\mu)}=1$. We know
$t_{22}^{(\mu)}<1$ by the normalization of $P$, and hence $t_{32}^{(\mu)}>0$,
which also ensures $t_{13}^{(\mu+1)}=0$ since three switching
paths starting from level $\mu-1$ are excluded. It follows that
$t_{23}^{(\mu+1)}>0$ due to \mbox{$t_{33}^{(\mu+1)}<1$.} Thus, we get a contradiction
$p_2^{(\mu+1)}\geq\frac{1}{4}p_1^{(\mu-1)}+\frac{1}{2}p_2^{(\mu-1)}\geq$
$\frac{1}{4}p_1^{(\mu-1)}+\frac{1}{4}p_2^{(\mu-1)}
+\frac{1}{4}p_3^{(\mu-1)}=\frac{1}{4}$ according to (\ref{p2m114}).

Similarly, for the remaining case of $t_{12}^{(\mu)}=t_{13}^{(\mu)}=0$
we obtain $t_{32}^{(\mu)}=t_{33}^{(\mu)}=1$ by the
\linebreak
normalization of $P$,
which again ensures $t_{13}^{(\mu+1)}=0$ implying $t_{23}^{(\mu+1)}>0$.
Hence, $p_2^{(\mu+1)}\geq$
\linebreak
$\frac{1}{4}p_1^{(\mu-1)}+\frac{1}{2}p_2^{(\mu-1)}
+\frac{1}{2}p_3^{(\mu-1)}\geq\frac{1}{4}$, which contradicts (\ref{p2m114}).
This completes the proof of the lemma.
\end{proof}

\subsection{The block structure from $m'$ to $\mu$  ($m'$-conditions~\ref{mc1}--\ref{mc3})}
\label{cond1-3}

We define the first \textbf{level}~${\mathbf{\boldsymbol m'}}$ of the underlying block
to be the greatest level such that $2\leq m'\leq\mu$ and $t_{32}^{(m')}>0$
($m'$-condition~\ref{mc2}), which exists since at least $t_{32}^{(2)}>0$.
In the following lemma, we will analyze the initial block structure from
level $m'$ through $\mu$, which is illustrated in Figure~\ref{fmptomu} (where
the dashed line shows that there is no edge from $v_1^{(k-1)}$ or $v_2^{(k-1)}$
to $v_3^{(k)}$ for any $m'<k\leq\mu$).

\begin{figure}[!h]
\centering
\includegraphics[height=5.3cm]{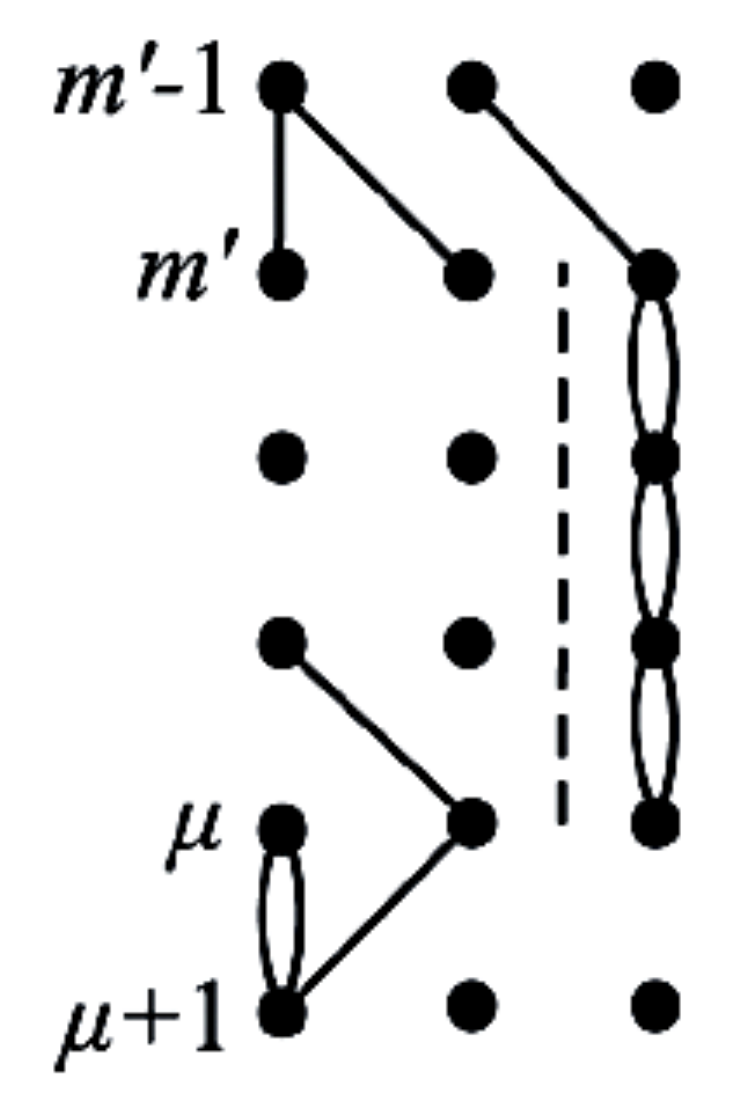}\vspace*{-3mm}
\caption{The block structure from $m'$ to $\mu$.}
\label{fmptomu}\vspace*{-3mm}
\end{figure}

\begin{lemma}
\label{mptomu}
$t_{31}^{(k)}=t_{32}^{(k)}=0$ and $t_{33}^{(k)}=1$ for every $k=m'+1,\ldots,\mu$.
\end{lemma}
\begin{proof}
On the contrary, let $k$ be the greatest level such that $m'<k\leq\mu$ and
$t_{33}^{(k)}<1$, that is $t_{33}^{(\ell)}=1$ for $\ell=k+1,\ldots,\mu$.
Obviously, $t_{33}^{(k)}>0$ because $t_{32}^{(\ell)}=0$ for
every $\ell=m'+1,\ldots,k,\ldots,\mu$ by the definition of $m'$,
and $t_{31}^{(\ell)}=0$ for every $\ell=k,\ldots,\mu$ since otherwise
$p_3^{(\mu)}\geq p_3^{(\ell)}>\frac{1}{6}$, which contradicts (\ref{p3mu}).
Hence, $t_{33}^{(k)}=\frac{1}{2}$ and the edge from $v_3^{(k-1)}$ to
$v_3^{(k)}$ is the only edge that leads to $v_3^{(k)}$ due to
$t_{31}^{(k)}=t_{32}^{(k)}=0$, while the other edge from $v_3^{(k-1)}$
goes either to $v_1^{(k)}$ or to $v_2^{(k)}$. Thus, either
$\mathbf{a}^{(m)}\in M(v_1^{(k)})\cup M(v_2^{(k)})$ for $\mathbf{a}^{(m)}$
satisfying $m$-condition~\ref{mc4} (Paragraph~\ref{prplan}), or
a 1-neighbor $\mathbf{a}'\in\Omega_1(\mathbf{a}^{(m)})\cap M(v_3^{(k-1)})$
of $\mathbf{a}^{(m)}$ exists that differs from $\mathbf{a}^{(m)}$ in the
variable that is tested at $v_3^{(k-1)}$ so that also
$\mathbf{a}'\in M(v_1^{(k)})\cup M(v_2^{(k)})$.
Since $M(v_1^{(k)})\cup M(v_2^{(k)})=M(v_1^{(\mu)})\cup M(v_2^{(\mu)})$ and
$t_{12}^{(\mu+1)}=\frac{1}{2}$ by Lemma~\ref{levelmu}, there is a 2-neighbor
$\mathbf{a}''\in \Omega_2(\mathbf{a}^{(m)})\cap M(v_1^{(\mu+1)})\subseteq H$ of
$\mathbf{a}^{(m)}$ such that $P(\mathbf{a}'')=1$ by $m$-condition~\ref{mc4}
since $M(v_1^{(\mu+1)})\subseteq M(v_1^{(m)})\cup M(v_2^{(m)})$, which is a
contradiction. Thus $t_{33}^{(k)}=1$ for $k=m'+1,\ldots,\mu$.
\end{proof}

Lemma~\ref{mptomu} together with assumption (\ref{p3mu}) gives
\begin{eqnarray}
\label{Mm'mu}
p_1^{(m')}+p_2^{(m')}&=&p_1^{(\mu)}+p_2^{(\mu)}\,,\\
\label{p3m'}
p_3^{(m')}&=&p_3^{(\mu)}<\frac{1}{6}
\end{eqnarray}
which verifies $m'$-condition~\ref{mc3} for the first block level $m'$. Note that
inequality (\ref{p3m'}) ensures $m'\geq 3$ due to $p_3^{(2)}\geq\frac{1}{4}$.
Finally, the following lemma shows $m'$-condition~\ref{mc1}.
\begin{lemma}
\label{mpc1}
$t_{11}^{(m')}=t_{21}^{(m')}=\frac{1}{2}$ \emph{($m'$-condition~\ref{mc1}).}
\end{lemma}
\begin{proof}
Obviously, $t_{31}^{(m')}=0$ since otherwise $p_3^{(m')}>\frac{1}{6}$ which
contradicts (\ref{p3m'}). For $t_{11}^{(m')}=1$ or $t_{21}^{(m')}=1$ we
obtain $t_{12}^{(m')}+t_{22}^{(m')}>0$ and $t_{13}^{(m')}+t_{23}^{(m')}>0$
by the normalization of $P$. Thus either
$\mathbf{a}^{(m)}\in M(v_1^{(m'-1)})\subseteq M(v_1^{(m')})\cup M(v_2^{(m')})$
or a 1-neighbor $\mathbf{a}'\in\Omega_1(\mathbf{a}^{(m)})\cap
(M(v_2^{(m'-1)})\cup M(v_3^{(m'-1)}))$ of
$\mathbf{a}^{(m)}$ exists such that
$\mathbf{a}'\in M(v_1^{(m')})\cup M(v_2^{(m')})$.
Since $M(v_1^{(m')})\cup M(v_2^{(m')})=M(v_1^{(\mu)})\cup M(v_2^{(\mu)})$ and
$t_{12}^{(\mu+1)}=\frac{1}{2}$ by Lemma~\ref{levelmu}, there is a 2-neighbor
$\mathbf{a}''\in \Omega_2(\mathbf{a}^{(m)})\cap M(v_1^{(\mu+1)})\subseteq H$ of
$\mathbf{a}^{(m)}$ such that $P(\mathbf{a}'')=1$ which is a contradiction.
The last possibility $t_{11}^{(m')}=t_{21}^{(m')}=\frac{1}{2}$ follows.
\end{proof}

\subsection{An upper bound on $p_1^{(\omega+1)}$ in terms of $p_1^{(m')}+p_2^{(m')}$}

In Paragraph~\ref{ubp1mp2mp1o}, we have upper-bounded $p_1^{(m)}+p_2^{(m)}$ at the last
block level $m$ in terms of $p_1^{(\omega+1)}$ provided that $\omega<m$. In this
paragraph, we will extend this estimate by upper-bounding $p_1^{(\omega+1)}$
(or $p_1^{(m)}+p_2^{(m)}$ for $\omega=m$) in terms of $p_1^{(m')}+p_2^{(m')}$ from
the first block level $m'$. Putting these two bounds together, we will obtain a recursive
formula for an upper bound on $p_1^{(m)}+p_2^{(m)}$ in terms of $p_1^{(m')}+p_2^{(m')}$
which will be used in Section~\ref{recursionm} for verifying condition (\ref{acond}).

We first resolve the case of the empty block when $\nu=m=\omega$, $t_{33}^{(m)}=0$,
$t_{13}^{(m)}+t_{23}^{(m)}=1$, and $t_{32}^{(m)}=1$ (see Figure~\ref{nuem}). It follows
from Corollary~\ref{belowmu} and Lemma~\ref{mptomu}
(see Figures~\ref{mutonu} and~\ref{fmptomu}, respectively) that
$M(v_1^{(m')})\cup M(v_2^{(m')})=M(v_1^{(m)})\cup M(v_2^{(m)})$ which ensures
$m'$-condition~\ref{mc4} ($m'$-conditions~\ref{mc1}--\ref{mc3} have already been checked
in Paragraph~\ref{cond1-3}) and $p_1^{(m')}+p_2^{(m')}=p_1^{(m)}+p_2^{(m)}$. Hence, the
empty block can be skipped in our analysis by replacing $m'$ with $m$, and we will
further consider only the non-empty blocks.

\medskip
It follows from the definition of partition class
$R$ (see Figure~\ref{mutonu}) and Lemma~\ref{levelmu} that
\begin{equation}
\label{p1nu}
p_1^{(\nu)}=p_1^{(\mu)}+p_2^{(\mu)}\left(1-\frac{1}{2^{|R|}}\right)
\qquad\mbox{for }\nu<m\,.
\end{equation}
For $\nu=m$ when $\nu'=\nu-1$, we know $t_{33}^{(m)}>0$ because we assume
a non-empty block, and hence, either $t_{12}^{(m)}=t_{32}^{(m)}=\frac{1}{2}$
and $t_{33}^{(m)}=1$, or $t_{13}^{(m)}=t_{33}^{(m)}=\frac{1}{2}$ and
$t_{32}^{(m)}=1$ (see Figure~\ref{nuem}) by the definition of $\nu$,
Lemma~\ref{aboveedge}.ii, and $m$-conditions~\ref{mc1} and~\ref{mc2}, which
also ensures $\omega=m$ in both cases. Thus,
\begin{equation}
\label{p1nub}
p_1^{(m)}+p_2^{(m)}=p_1^{(\mu)}+p_2^{(\mu)}\left(1-\frac{1}{2^{|R|+1}}\right)
\qquad\mbox{for }\nu=m=\omega
\end{equation}
according to (\ref{defnup}).
For $\nu=m-1$ we know $t_{12}^{(m)}=t_{13}^{(m)}=0$ leading to
$t_{32}^{(m)}=t_{33}^{(m)}=1$, for which $\omega=m$ can be assumed without
loss of generality.

\medskip
Further assume $\nu<m-1$, while the resulting formula for $\nu<m$ will also be
verified for the case of $\nu=m-1$ (when $\omega=m$) below in (\ref{p1om12}).
We know by the definition of $\nu$ that $t_{12}^{(\nu+1)}=t_{13}^{(\nu+1)}=0$,
which excludes $t_{32}^{(\nu+1)}=1$ and $t_{33}^{(\nu+1)}=1$ since $P$ is
normalized. First consider the case of $\omega>\nu$ excluding
$\omega=\nu-1\geq\mu$ and $\omega=\nu$ for now, cf.\ (\ref{omnu}).
Then the double-edge path from the definition of
$\omega$ passes through a double edge from $v\in\{v_2^{(\nu)},v_3^{(\nu)}\}$
to $v_2^{(\nu+1)}$, while the two edges from the other
node $v'\in\{v_2^{(\nu)},v_3^{(\nu)}\}\setminus\{v\}$ lead
to $v_2^{(\nu+1)}$ and $v_3^{(\nu+1)}$, respectively, as depicted in
Figure~\ref{nutoo}. For $\ell=\nu+2,\ldots,\omega$, we have either
$t_{22}^{(\ell)}=1$ implying $t_{33}^{(\ell)}=\frac{1}{2}$ if $\ell<m$, or
$t_{32}^{(\ell)}=1$ if $\ell=m$. Moreover, $t_{12}^{(\omega+1)}=0$
for $\omega<m$ by Corollary~\ref{belowmu}.\ref{podm4}. Hence,
$p_3^{(\nu+1)}=p_2^{(\mu)}/2^{|R|+1}$ (cf.\ Figure~\ref{mutonu} and
Lemma~\ref{levelmu}) upper-bounds the fraction of all the inputs
whose computational path traverses nodes
$v',v_3^{(\nu+1)},v_3^{(\nu+2)},\ldots,v_3^{(\ell)},v_1^{(\ell+1)}$
for some $\nu+1\leq\ell\leq\min(\omega,m-1)$. It follows that
\begin{equation}
\label{p1om1}
p_1^{(\omega+1)}\leq p_1^{(\nu)}+\frac{p_2^{(\mu)}}{2^{|R|+1}}\qquad\mbox{for }\omega<m
\end{equation}
which is even valid for any $\max(\nu-1,\mu)\leq\omega<m$ since obviously
$p_1^{(\omega+1)}=p_1^{(\nu)}$ for $\omega=\nu-1\geq\mu$ as well as for
$\omega=\nu<m$, while
\begin{equation}
\label{p1om12}
p_1^{(m)}+p_2^{(m)}\leq p_1^{(\nu)}+\frac{p_2^{(\mu)}}{2^{|R|+1}}
\qquad\mbox{for }\omega=m
\end{equation}
which also holds for $\nu=m-1$ because $p_1^{(m)}+p_2^{(m)}=p_1^{(\nu)}$ in this case.

\begin{figure}[h]
\vspace{2mm}
\centering
\includegraphics[height=5.6cm]{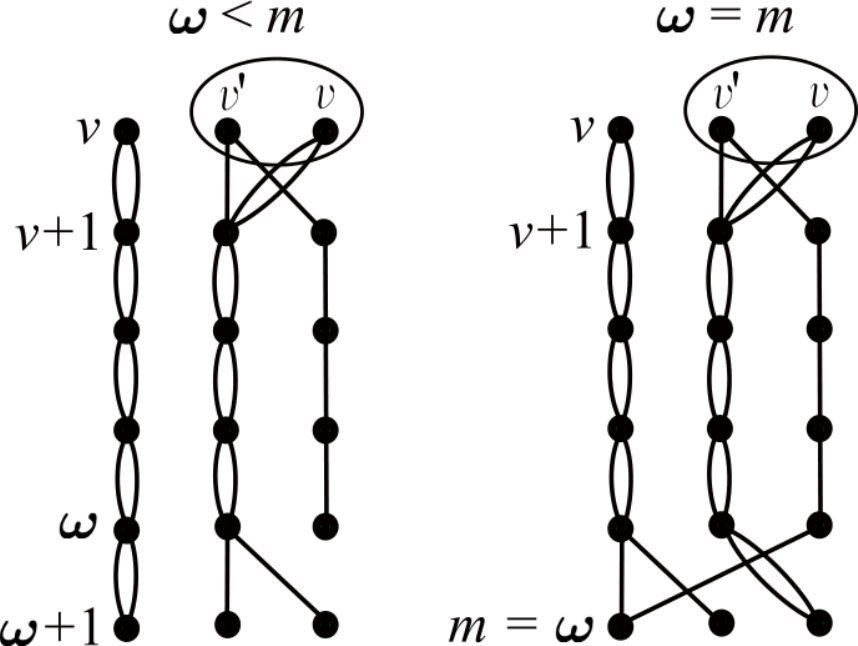}
\caption{The block structure from $\nu<\omega$ to $\omega+1$ (or to $m$ if $\omega=m$).}
\label{nutoo}
\end{figure}

In addition, we will prove the following lemma:
\begin{lemma}
\label{lemp1p24p2}
\begin{equation}
\label{p1p24p2}
p_1^{(\mu)}+p_2^{(\mu)}\leq 4p_2^{(\mu)}\,.
\end{equation}
\end{lemma}
\begin{proof}
First consider the case of $\mu>m'$. Clearly, $t_{21}^{(\mu)}>0$ follows
from the definition of $\mu'$ for $\mu=\mu'$, while for $\mu=\mu'-1$, the
case of $t_{21}^{(\mu)}=0$ translates to original $t_{11}^{(\mu)}=0$ (before
$v_1^{(\mu)}$ and $v_2^{(\mu)}$ were swapped) which contradicts
the normalization of $P$ by Lemma~\ref{mptomu}.
Hence, we have $p_1^{(\mu)}+p_2^{(\mu)}=p_1^{(\mu-1)}+p_2^{(\mu-1)}
\leq 2p_1^{(\mu-1)}\leq 4p_2^{(\mu)}$ according to Lemma~\ref{mptomu}.
For $\mu=m'$, on the other hand, we will distinguish three cases.
For $t_{32}^{(\mu)}=t_{33}^{(\mu)}=1$, we know $p_1^{(\mu)}=p_2^{(\mu)}$
by Lemma~\ref{mpc1}, which implies (\ref{p1p24p2}).
For $t_{12}^{(\mu)}+t_{22}^{(\mu)}=\frac{1}{2}$, we have $t_{33}^{(\mu)}=1$ by
Lemma~\ref{aboveedge}.ii, which gives
$p_1^{(\mu)}\leq$ $\frac{1}{2}p_1^{(\mu-1)}+\frac{1}{2}p_2^{(\mu-1)}
<\frac{1}{2}+\frac{1}{6}=\frac{2}{3}$ since $p_2^{(\mu-1)}<\frac{1}{3}$
by $m'$-conditions~\ref{mc2} and (\ref{p3m'}). In addition,
$p_3^{(\mu-1)}<p_3^{(\mu)}<\frac{1}{6}$ implying
$p_2^{(\mu-1)}+p_3^{(\mu-1)}<\frac{1}{2}$ which means
$p_2^{(\mu)}\geq\frac{1}{2}p_1^{(\mu-1)}>\frac{1}{4}$ by Lemma~\ref{mpc1}.
It follows that $p_1^{(\mu)}<\frac{2}{3}<\frac{3}{4}<3p_2^{(\mu)}$ which
gives (\ref{p1p24p2}). Similarly for
$t_{13}^{(\mu)}+t_{23}^{(\mu)}>0$, we have $t_{32}^{(\mu)}=1$ implying
$\frac{1}{6}>p_3^{(\mu)}\geq p_2^{(\mu-1)}\geq p_3^{(\mu-1)}$ due
to (\ref{p3m'}), and hence $p_1^{(\mu-1)}>\frac{2}{3}$ which ensures
$p_2^{(\mu)}>\frac{1}{3}$ by Lemma~\ref{mpc1}, while
$p_1^{(\mu)}\leq\frac{1}{2}p_1^{(\mu-1)}
+p_3^{(\mu-1)}<\frac{1}{2}+\frac{1}{6}=\frac{2}{3}<1<3p_2^{(\mu)}$
which completes the proof of the lemma.
\end{proof}

For $\nu<m$, equation (\ref{p1nu}) is plugged into (\ref{p1om1}) if $\omega<m$
or into (\ref{p1om12}) if $\omega=m$, while equation (\ref{p1nub}) is considered
for $\nu=m$ (implying $\omega=m$). Then Lemma~\ref{lemp1p24p2} and equation
(\ref{Mm'mu}) are employed, which results in
\begin{eqnarray}
\label{p1om1b}
p_1^{(\omega+1)}&\leq& p_1^{(\mu)}
+p_2^{(\mu)}\left(1-\frac{1}{2^{|R|}}\right)
+\frac{p_2^{(\mu)}}{2^{|R|+1}}
=p_1^{(\mu)}+p_2^{(\mu)}\left(1-\frac{1}{2^{|R|+1}}\right)
\qquad\nonumber\\
&\leq&\left(p_1^{(m')}+p_2^{(m')}\right)
\left(1-\frac{1}{2^{|R|+3}}\right)
\qquad\mbox{for }\omega<m\,,
\end{eqnarray}
\begin{equation}
 \label{p1om1b2}
p_1^{(m)}+p_2^{(m)}\leq\left(p_1^{(m')}+p_2^{(m')}\right)
\left(1-\frac{1}{2^{|R|+3}}\right)\qquad\mbox{for }\omega=m\,.
\end{equation}
Formula (\ref{p1om1b}) can further be plugged into (\ref{p1p2m}) giving
\begin{equation}
\label{p1p2mb}
p_1^{(m)}+p_2^{(m)}\leq 1-\left(1-\left(p_1^{(m')}+p_2^{(m')}\right)
\left(1-\frac{1}{2^{|R|+3}}\right)\right)
\prod_{j=1}^q\left(1-\frac{1}{2^{|Q_j|}}\right)
\end{equation}
which is even valid for $\omega=m$ (i.e.\ $q=0$) since equation
(\ref{p1p2mb}) coincides with (\ref{p1om1b2}) in this case.

\section{The recursion}
\label{recursionm}

In the previous Sections~\ref{suffc}--\ref{sbelowmu}, we have analyzed the
structure of the block of $P$ from level $m'$ through $m$ (see
Figure~\ref{block}). We will now employ
this block analysis recursively so that $m=m_r$ is replaced by $m'=m_{r+1}$.
For this purpose, we introduce additional index $b=1,\ldots,r$ to the
underlying objects in order to differentiate among respective blocks. For
example, the sets $R,Q_1,\ldots,Q_q$, defined in
Section~\ref{dfpart}, corresponding to the $b$th block are denoted as
$R_b,Q_{b1},\ldots,Q_{bq_b}$, respectively.

It follows from the definition of
partition class in Paragraph~\ref{dfpartR} that, for any $b>1$, the nodes
labeled with the variables whose indices are in $R_b$ are connected
with the nodes corresponding to $R_{b-1}$ through a computational path
which traverses nodes $v_1^{(\nu_b')},v_1^{(\nu_b'+1)},\ldots,v_1^{(m_b-1)}$
since $\nu_b'\leq m_b-1$ according to (\ref{defnup}). Hence, sets
$R_1,\ldots,R_r$ are pair-wise disjoint because $P$ is read-once, and
thus they create a partition.

\subsection{Inductive assumptions}

In particular, we will proceed by induction on $r$, starting with
$r=0$ and $m_0=d$. In the induction step for $r+1$, we assume that the four
$m_r$-conditions from Paragraph~\ref{prplan} are met for the last level $m=m_r$
of the $(r+1)$st block (see Paragraph~\ref{inicasemd} for $r=0$), and let the
assumption (\ref{p3mu}) be satisfied for the previous blocks, that is,
\begin{equation}
\label{p3mui}
p_3^{(\mu_b)}<\frac{1}{6}
\end{equation}
for every $b=1,\ldots,r$. In addition, assume
\begin{equation}
\label{1Pir}
1-\Pi_r<\delta=\min\left(\varepsilon-\varepsilon',\frac{6\varepsilon-5}{7}\right)<\frac{1}{7}
\end{equation}
where $\varepsilon>\frac{5}{6}$ and $\varepsilon'<\varepsilon$ are the parameters of
Theorem~\ref{suff}, and denote
\begin{equation}
\label{Pikpii}
\Pi_k=\prod_{b=1}^k\pi_b\,,\qquad
\pi_b=\prod_{j=1}^{q_b}\left(1-\frac{1}{2^{|Q_{bj}|}}\right)\,,
\end{equation}
\begin{equation}
\varrho_k=\prod_{b=1}^k\alpha_b\,,\qquad
\alpha_b=\left(1-\frac{1}{2^{|R_b|+3}}\right)
\end{equation}
for $k=1,\ldots,r$, $\varrho_0=\Pi_0=1$, and $\pi_b=1$ for
$q_b=0$, which will be used below to shorten the recursive inequality (\ref{p1p2mb}) and its solution. It follows from (\ref{Pikpii}) and (\ref{1Pir}) that
\begin{equation}
\label{piiPir1}
\pi_b\geq\Pi_r>1-\delta>\frac{2}{3}
\end{equation}
which verifies assumption (\ref{pik45}) for every $b=1,\ldots,r$.
Hence, we can employ recursive inequality (\ref{p1p2mb}) from
Section~\ref{sbelowmu} which is rewritten as
\begin{equation}
\label{pi1rec}
p_{b-1}\leq 1-(1-p_b\alpha_b)\pi_b=1-\pi_b+p_b\alpha_b\pi_b
\end{equation}
for $b=1,\ldots,r$ where notation $p_b=p_1^{(m_b)}+p_2^{(m_b)}$
is introduced. Starting with
\begin{equation}
p_0=p_1^{(d)}+p_2^{(d)}\geq\varepsilon
\end{equation}
which follows from (\ref{podmhitt}), recurrence (\ref{pi1rec}) can be
solved as
\begin{eqnarray}
\varepsilon&\leq&\sum_{k=1}^r(1-\pi_k)
\prod_{b=1}^{k-1}\alpha_b\pi_b+p_r\prod_{b=1}^r\alpha_b\pi_b
<\sum_{k=1}^r(1-\pi_k)\Pi_{k-1}+p_r\varrho_r\Pi_r\qquad\nonumber\\
\label{epsrec}
&=&1-\Pi_r+p_r\varrho_r\Pi_r\,.
\end{eqnarray}
In addition,
\begin{equation}
\label{rhor}
\varrho_r>p_r\varrho_r\Pi_r>\varepsilon-\delta\geq\varepsilon'
\end{equation}
follows from (\ref{epsrec}) and (\ref{1Pir}).

\subsection{Recursive step}
\label{recst}

Throughout this paragraph, we will consider the case when
\begin{equation}
\label{1Pir1}
1-\Pi_{r+1}<\delta
\end{equation}
(cf.\ assumption (\ref{1Pir})), while the case complementary to
(\ref{1Pir1}), which concludes the recursion, will be resolved below
in Section~\ref{indend}. Assuming condition (\ref{1Pir1}), we will
prove that inductive assumptions (\ref{p3mui}) and (\ref{1Pir}) are met for
$r$ replaced with $r+1$ together with the four $m_{r+1}$-conditions
for the first level $m_{r+1}$ of the $(r+1)$st block so that we can
further proceed in the recursion.

\medskip
By analogy to (\ref{piiPir1}), inequality (\ref{1Pir1}) implies
\begin{equation}
\label{pir1sqrt1213}
\pi_{r+1}>1-\delta>\frac{2}{3}\,.
\end{equation}
For $\omega_{r+1}<m_r$, we know
\begin{equation}
p_r\leq 1-\left(p_2^{(\omega_{r+1}+1)}+p_3^{(\omega_{r+1}+1)}\right)\pi_{r+1}
\end{equation}
according to (\ref{p1p2m}), and
\begin{equation}
p_2^{(\omega_{r+1}+1)}+p_3^{(\omega_{r+1}+1)}\geq p_3^{(\mu_{r+1})}
\end{equation}
by the definition of $\omega_{r+1}$ and Lemma~\ref{aboveedge}.iii--iv
(for $k=\omega_{r+1}$), which altogether gives
\begin{equation}
\varepsilon<1-\Pi_r+\left(1-p_3^{(\mu_{r+1})}\pi_{r+1}\right)\varrho_r\Pi_r
\end{equation}
according to (\ref{epsrec}). Hence,
\begin{equation}
\varepsilon-\delta
<\left(1-p_3^{(\mu_{r+1})}\pi_{r+1}\right)\varrho_r\Pi_r
<1-p_3^{(\mu_{r+1})}\pi_{r+1}
\end{equation}
follows from (\ref{1Pir}), which gives
\begin{equation}
\label{p3mur1112}
p_3^{(\mu_{r+1})}<\frac{1-\varepsilon+\delta}{1-\delta}\leq\frac{1}{6}
\qquad\mbox{for }\omega_{r+1}<m_r
\end{equation}
by (\ref{pir1sqrt1213}) and  (\ref{1Pir}). Inequality (\ref{p3mur1112}) is even
valid for $\omega_{r+1}=m_r$ since
\begin{equation}
\label{p3mur1112b}
p_3^{(\mu_{r+1})}\leq p_3^{(m_r)}<\frac{1}{6}\qquad\mbox{for }\omega_{r+1}=m_r
\end{equation}
according to $m_r$-condition~\ref{mc3}. Therefore, assumptions (\ref{p3mu})
and (\ref{pik45}) of the analysis in Section~\ref{sbelowmu} are also met for
the $(r+1)$st block according to (\ref{p3mur1112})--(\ref{p3mur1112b}) and
(\ref{pir1sqrt1213}), respectively, which justifies recurrence inequality
(\ref{pi1rec}) for $b=r+1$ leading to the solution
\begin{equation}
\label{epsrecr1}
\varepsilon<1-\Pi_{r+1}+p_{r+1}\varrho_{r+1}\Pi_{r+1}
\end{equation}
by analogy to (\ref{epsrec}) where $r$ is replaced with $r+1$. Similarly to
(\ref{rhor}), we obtain
\begin{equation}
\label{rhor1}
\varrho_{r+1}>
\varepsilon'
\end{equation}
by combining (\ref{epsrecr1}) with (\ref{1Pir1}).
Thus, inductive assumptions (\ref{p3mui}) and (\ref{1Pir}) are valid for
$r$ replaced by $r+1$ according to (\ref{p3mur1112})--(\ref{p3mur1112b})
and (\ref{1Pir1}), respectively.

\medskip
In order to proceed in the next induction step, we still
need to verify the four $m_{r+1}$-conditions from Paragraph~\ref{prplan}
for level $m_{r+1}$. In Paragraph~\ref{cond1-3},
$m_{r+1}$-conditions~\ref{mc1}--\ref{mc3} have been proven, and thus, it
suffices to validate $m_{r+1}$-condition~\ref{mc4}. For this purpose, we
exploit the fact that $A$ is $\varepsilon'^{11}$-rich after we show
corresponding condition (\ref{acond}) for partition
$\{R_1,\ldots,R_{r+1}\}$ of $I=\bigcup_{b=1}^{r+1}R_b$. In particular,
\begin{equation}
\varepsilon'^{11}<\varrho_{r+1}^{11}
<\prod_{b=1}^{r+1}\left(1-\frac{1}{2^{|R_b|}}\right)
\end{equation}
follows from (\ref{rhor1}) since for any $1\leq b\leq r+1$,
\begin{equation}
\label{to11}
\left(1-\frac{1}{2^{|R_b|+3}}\right)^{11}<1-\frac{1}{2^{|R_b|}}
\end{equation}
for $|R_b|\geq 1$ because $f(x)=\ln(1-\frac{1}{x})/\ln(1-\frac{1}{8x})$
is a decreasing function for $x=2^{|R_b|}\geq 2$, and $f(2)<11$.
This provides required $\mathbf{a}^{(m_{r+1})}\in A$ such that for every
$b=1,\ldots,r+1$ there exists $i\in R_b$ that meets $a_i^{(m_{r+1})}\not=c_i$
according to (\ref{cond}) for $Q=\emptyset$. Obviously, the computational path
for this $\mathbf{a}^{(m_{r+1})}$ ends up in sink $v_1^{(d)}$ or $v_2^{(d)}$
labeled with 1 when we put $\mathbf{a}^{(m_{r+1})}$ at node $v_1^{(m_{r+1})}$
or $v_2^{(m_{r+1})}$ by the definition of $R_b$, $c_i$ and by the structure of
branching program~$P$ (see Figure~\ref{mutonu}), which proves
$m_{r+1}$-condition~\ref{mc4}. Thus, the inductive assumptions
are met for $r+1$ and we can proceed recursively for $r$ replaced with $r+1$ etc.
until condition (\ref{1Pir1}) is broken.

\section{The end of recursion}
\label{indend}

In this section, we will consider the case of
\begin{equation}
\label{1Pir1n}
1-\Pi_{r+1}\geq\delta
\end{equation}
complementary to (\ref{1Pir1}), which concludes the recursion from
Section~\ref{recursionm} as follows.
Suppose $|Q_{bj}|>\log n$ for every $b=1,\ldots,r+1$ and
$j=1,\ldots,q_b$, then we would have
\begin{eqnarray}
\Pi_{r+1}&=&\prod_{b=1}^{r+1}\prod_{j=1}^{q_b}\left(1-\frac{1}{2^{|Q_{bj}|}}\right)
\geq\left(1-\frac{1}{2^{\log n}}\right)^{\frac{n}{\log n}}\noindent\\
&>&1-\frac{1}{n}\cdot\frac{n}{\log n}=1-\frac{1}{\log n}\,,
\end{eqnarray}
which breaks (\ref{1Pir1n}) for sufficiently large $n$.
Hence, there must be $1\leq b^*\leq r+1$ and $1\leq j^*\leq q_{b^*}$ such that
$|Q_{b^*j^*}|\leq\log n$, and we denote $Q=Q_{b^*j^*}$.
Clearly, $Q\cap R_b=\emptyset$ for $b=1,\ldots,b^*-2$
due to $P$ is read-once while it may happen that $Q\cap R_{b^*-1}\not=\emptyset$
for $j^*=1$, $\kappa_{b^*1}=\sigma_{b^*1}=m_{b^*-1}$, and $t_{23}^{(m_{b^*-1})}=0$.
Thus, let $r^*$ be the maximum of $b^*-2$ and $b^*-1$ such that
$Q\cap R_{r^*}=\emptyset$. We will again employ the fact that $A$ is
$\varepsilon'^{11}$-rich. First condition (\ref{acond}) for partition
$\{R_1,\ldots,R_{r^*}\}$ of $I=\bigcup_{b=1}^{r^*}R_b$ is verified as
\begin{equation}
\prod_{b=1}^{r^*}\left(1-\frac{1}{2^{|R_b|}}\right)
>\varrho_r^{11}>\varepsilon'^{11}
\end{equation}
according to (\ref{to11}) and (\ref{rhor}). This provides $\mathbf{a}^*\in A$ such
that $a_i^*=c_i^Q$ for every $i\in Q$ and at the same time,
for every $b=1,\ldots,r^*$ there exists $i\in R_b$ that meets $a_i^*\not=c_i^{R_b}$
according to (\ref{cond}).
\begin{lemma}
\label{gsp}
Denote $\lambda=\lambda_{b^*j^*}$.
There are two generalized `switching' paths (cf.\ Lemma~\ref{aboveedge}.ii)
starting from $v_2^{(k)}$ and $v_3^{(k)}$, respectively, at level $k$
satisfying $3<\max(\lambda-2,\mu_{b^*})\leq k<\lambda$, which end in any node from the set $\left\{v_1^{(\lambda-1)},v_1^{(\lambda)},v_3^{(\lambda)}\right\}$.
\end{lemma}
\begin{proof}
For the notation simplicity, we will omit the block index $b^*$ in this proof.
We know $\omega<m$ due to $q>0$, and $\lambda>\mu$ from Paragraph~\ref{dfpartQ}.
Consider first the case when $t_{12}^{(\lambda)}=t_{13}^{(\lambda)}=0$.
Obviously, $t_{22}^{(\lambda)}<1$ follows from the definition of $\lambda$
for $\lambda>\omega$ and from the definition of $\omega$ for $\lambda=\omega$,
which gives $t_{22}^{(\lambda)}=t_{32}^{(\lambda)}=\frac{1}{2}$ and
$t_{23}^{(\lambda)}>0$ by the normalization of $P$.
For $t_{33}^{(\lambda)}=\frac{1}{2}$, we obtain two switching paths
$v_2^{(\lambda-1)},v_3^{(\lambda)}$ and $v_3^{(\lambda-1)},v_3^{(\lambda)}$.
Thus assume $t_{33}^{(\lambda)}=0$
which ensures $t_{23}^{(\lambda)}=1$ and $\lambda>\mu+1$ since $\lambda=\mu+1$
would give $\omega>\lambda$. Consider first the case when
$t_{12}^{(\lambda-1)}=t_{13}^{(\lambda-1)}=0$,
which implies $t_{22}^{(\lambda-1)}>0$ and $t_{23}^{(\lambda-1)}>0$ by
$t_{11}^{(\lambda-1)}=1$ and the normalization of $P$, providing two
switching paths $v_2^{(\lambda-2)},v_2^{(\lambda-1)},v_3^{(\lambda)}$ and
$v_3^{(\lambda-2)},v_2^{(\lambda-1)},v_3^{(\lambda)}$.
Two switching paths $v_2^{(\lambda-2)},v_1^{(\lambda-1)}$ and
$v_3^{(\lambda-2)},v_1^{(\lambda-1)}$ are also guaranteed when
$t_{12}^{(\lambda-1)}>0$ and $t_{13}^{(\lambda-1)}>0$ appear simultaneously.
For $t_{12}^{(\lambda-1)}=0$ and $t_{13}^{(\lambda-1)}>0$,
we have $t_{22}^{(\lambda-1)}>0$ by the normalization of $P$, which together
with $t_{32}^{(\lambda)}=\frac{1}{2}$ produces two switching paths
$v_2^{(\lambda-2)},v_2^{(\lambda-1)},v_3^{(\lambda)}$ and
$v_3^{(\lambda-2)},v_1^{(\lambda-1)}$. For $t_{12}^{(\lambda-1)}>0$
and $t_{13}^{(\lambda-1)}=0$, the case of $t_{23}^{(\lambda-1)}>0$ ensures
two switching paths $v_2^{(\lambda-2)},v_1^{(\lambda-1)}$ and
$v_3^{(\lambda-2)},v_2^{(\lambda-1)},v_3^{(\lambda)}$, while for
$t_{23}^{(\lambda-1)}=0$ we obtain
$t_{12}^{(\lambda-1)}=t_{22}^{(\lambda-1)}=\frac{1}{2}$ and
$t_{33}^{(\lambda-1)}=1$, which implies $\lambda=\nu+1$ and $\omega>\lambda$
by Lemma~\ref{aboveedge}.iii contradicting the definition of
$\lambda\geq\omega\geq\nu-1$. This completes the argument for
$t_{12}^{(\lambda)}=t_{13}^{(\lambda)}=0$.

The case of $t_{13}^{(\lambda)}>0$ and $t_{12}^{(\lambda)}>0$ produces
two switching paths $v_2^{(\lambda-1)},v_1^{(\lambda)}$ and
$v_3^{(\lambda-1)},v_1^{(\lambda)}$. Further consider the case when
$t_{13}^{(\lambda)}>0$ and $t_{12}^{(\lambda)}=0$. Obviously,
$t_{22}^{(\lambda)}<1$ follows from the definition of $\lambda$ for
$\lambda>\omega$ and from the definition of $\omega$ for $\lambda=\omega$.
Hence, $t_{32}^{(\lambda)}>0$ which provides two switching paths
$v_2^{(\lambda-1)},v_3^{(\lambda)}$ and
$v_3^{(\lambda-1)},v_1^{(\lambda)}$. Finally, consider the case when
$t_{12}^{(\lambda)}>0$ and $t_{13}^{(\lambda)}=0$, for which
$t_{33}^{(\lambda)}>0$ generates two switching paths
$v_2^{(\lambda-1)},v_1^{(\lambda)}$ and $v_3^{(\lambda-1)},v_3^{(\lambda)}$,
while for $t_{33}^{(\lambda)}=0$ we obtain $t_{32}^{(\lambda)}=\frac{1}{2}$
and $t_{23}^{(\lambda)}=1$, which implies $\lambda=\nu$ and $\omega>\lambda$
by Lemma~\ref{aboveedge}.iii contradicting the definition of
$\lambda\geq\omega\geq\nu-1$.
\end{proof}

By a similar argument to Lemma~\ref{aboveedge}.ii, Lemma~\ref{gsp} gives
a 2-neighbor $\mathbf{a}'\in\Omega_2(\mathbf{a}^*)\subseteq H$ of
$\mathbf{a}^*\in A$ such that
$\mathbf{a}'\in M(v_1^{(\lambda)})\cup M(v_3^{(\lambda)})$. Thus,
either $\mathbf{a}'\in M(v_1^{(\lambda)})\subseteq M(v_1^{(m_{b^*-1})})
\cup M(v_2^{(m_{b^*-1})})$ or $\mathbf{a}'\in M(v_3^{(\lambda)})$
which implies $\mathbf{a}'\in M(v_1^{(\kappa_{b^*j^*})})
\subseteq M(v_1^{(m_{b^*-1})})\cup M(v_2^{(m_{b^*-1})})$ since
$a_i'=a_i^*=c_i^Q$ for every $i\in Q$ according to (\ref{cond})
(see Figure~\ref{block} and~\ref{nutom}). Note that
$M(v_1^{(\kappa_{b^*j^*})})=M(v_1^{(m_{b^*-1})})$ for $r^*=b^*-2$.
Hence, $P(\mathbf{a}')=1$ because for every $b=1,\ldots,r^*$ there
exists $i\in R_b$ that meets $a_i'=a_i^*\not=c_i^{R_b}$
due to (\ref{cond}) (see Figure~\ref{block} and~\ref{mutonu}).
This completes the proof of Theorem~\ref{suff}.
\end{proof}

\section{The richness of almost $k$-wise independent sets}
\label{kwISrich}

In order to achieve an explicit polynomial-time construction of a hitting set
for read-once branching programs of width~3 we will combine Theorem~\ref{suff}
with the result due to Alon et al.~\cite{Alon92} who provided simple efficient
constructions of almost $k$-wise independent sets. In particular,
for $\beta>0$ and $k=O(\log n)$ it is possible to construct a
\emph{$(k,\beta)$-wise independent set} ${\cal A}\subseteq\{0,1\}^*$ in
time polynomial in $\frac{n}{\beta}$ such that for sufficiently large $n$ and
any index set $S\subseteq\{1,\ldots,n\}$ of size $|S|\leq k$, the probability
distribution on $S$ is almost uniform, i.e.\ the probability that a given
$\mathbf{c}\in\{0,1\}^n$ coincides with the strings from
${\cal A}_n={\cal A}\cap\{0,1\}^n$ on the bit locations
from $S$ can be approximated as
\begin{equation}
\label{almkwind}
\left|\frac{\left|{\cal A}_n^S(\mathbf{c})\right|}{|{\cal A}_n|}
-\frac{1}{2^{|S|}}\right|\leq\beta\,,
\end{equation}
where
${\cal A}_n^S(\mathbf{c})=\{\mathbf{a}\in{\cal A}_n\,|\,(\forall i\in S)\,a_i=c_i\}$.
We will prove that, for suitable $k$, any almost $k$-wise independent
set is $\varepsilon$-rich. It follows that almost $O(\log n)$-wise
independent sets are hitting sets for the class of read-once conjunctions of
DNF and CNF (cf.~\cite{DeETT10}).
\begin{theorem}
\label{richis}
Let $\varepsilon>0$, $C$ be the least odd integer greater than
$(\frac{2}{\varepsilon}\ln\frac{1}{\varepsilon})^2$, and $0<\beta<\frac{1}{n^{C+3}}\,$.
Then any $(\lceil(C+2)\log n\rceil,\beta)$-wise independent set is $\varepsilon$-rich.
\end{theorem}
\begin{proof}
Let ${\cal A}\subseteq\{0,1\}^*$ be a $(\lceil(C+2)\log n\rceil,\beta)$-wise independent
set. We will show that ${\cal A}$ is $\varepsilon$-rich. Assume $\{R_1,\ldots,R_r\}$ is
a partition of index set $I\subseteq\{1,\ldots,n\}$ satisfying condition (\ref{acond}),
and $Q\subseteq\{1,\ldots,n\}\setminus I$ such that $|Q|\leq\log n$. In order to show
for a given $\mathbf{c}\in\{0,1\}^n$ that there is $\mathbf{a}\in {\cal A}_n$ that meets
(\ref{cond}) for $Q$ and partition $\{R_1,\ldots,R_r\}$, we will prove that the
probability
\begin{equation}
\label{probability}
p=p({\cal A}_n)=\frac{\left|{\cal A}_n^Q(\mathbf{c})\setminus
\bigcup_{j=1}^r{\cal A}_n^{R_j}(\mathbf{c})\right|}{|{\cal A}_n|}
\end{equation}
of the event that $\mathbf{a}\in {\cal A}_n$ chosen uniformly at random satisfies
$\mathbf{a}\in{\cal A}_n^Q(\mathbf{c})$ and
$\mathbf{a}\not\in{\cal A}_n^{R_j}(\mathbf{c})$ for every $j=1,\ldots,r$,
is \emph{strictly positive}.

\begin{figure}[htbp]
\vspace{-2mm}
\centering
\includegraphics[height=19cm]{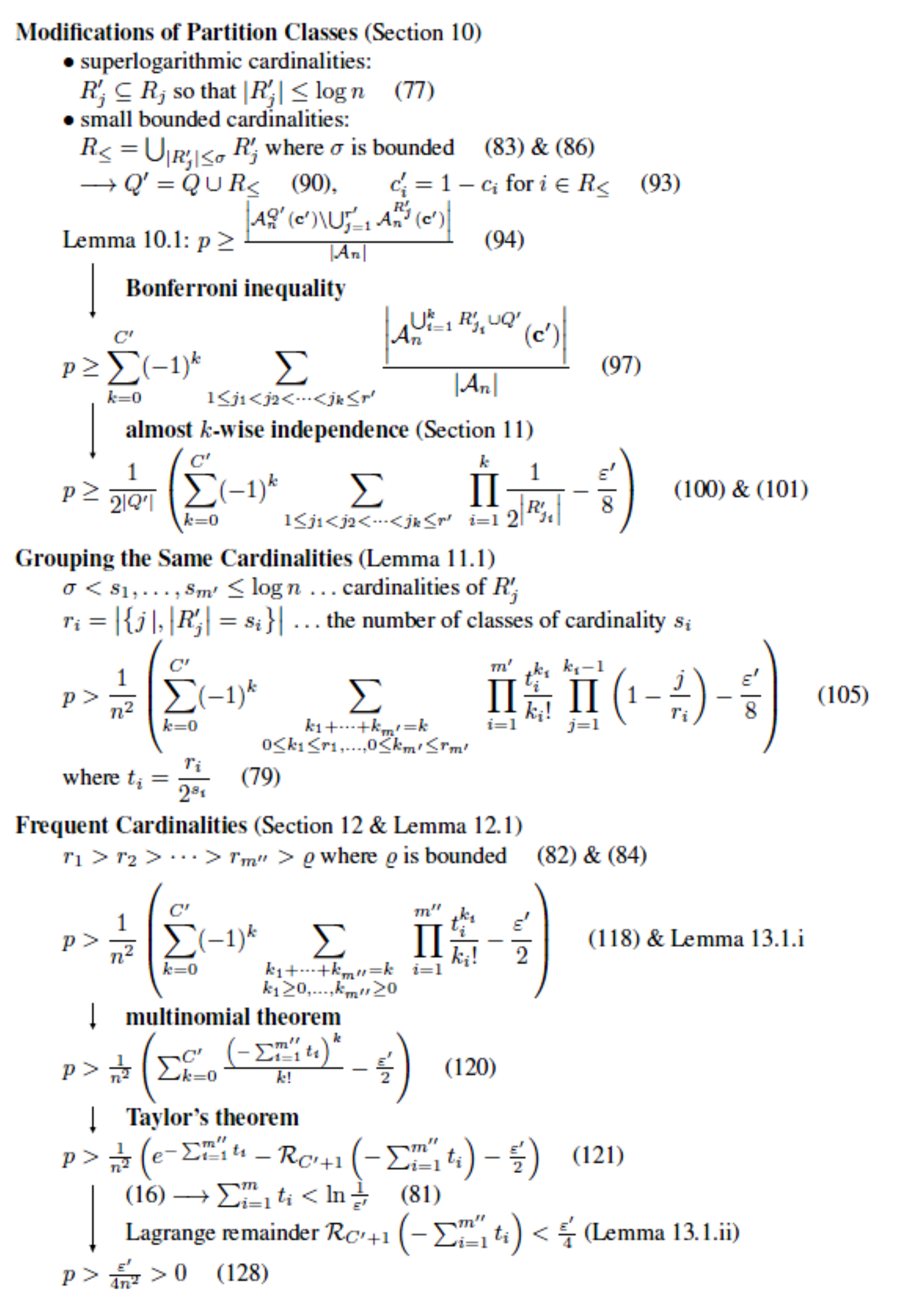}\vspace*{-2mm}
\caption{The main steps of the proof of Theorem~\ref{richis}.}
\label{prstr}
\end{figure}

\medskip
The main idea of the proof lies in lower-bounding the probability
(\ref{probability}). By using the assumption that ${\cal A}$ is almost
$O(\log n)$-wise independent this probability can be approximated by
the probability that any $\mathbf{a}\in\{0,1\}^n$ (not necessarily
in ${\cal A}_n$) satisfies (\ref{cond}) which can be expressed and
lower-bounded as
\begin{equation}
p(\{0,1\}^n)=\frac{1}{2^{|Q|}}\prod_{j=1}^r\left(1-\frac{1}{2^{|R_j|}}\right)
\geq\frac{\varepsilon}{n}>0
\end{equation}
according to (\ref{acond}) and $|Q|\leq\log n$. In particular, we briefly
comment on the main steps of the proof which are schematically
depicted in Figure~\ref{prstr} including references to corresponding sections,
lemmas, and equations. In Section~\ref{pclmod}, we will first modify
the partition classes $R_j$ so that their cardinalities are at most logarithmic
whereas the classes of small bounded cardinalities are merged with $Q$ and also
$\mathbf{c}$ is adjusted correspondingly. Lemma~\ref{lbmod} then ensures that
the probability $p$ from (\ref{probability}) is lower-bounded when using these
modified classes. Furthermore, Bonferroni inequality (the inclusion-exclusion
principle) and the assumption concerning the almost $k$-wise independence are
employed in Section~\ref{bonikwind} where also the classes of the same cardinality
are grouped. In Section~\ref{freqcard}, we will further reduce the underlying lower
bound on $p$ only to a sum over frequent cardinalities of partition classes to which
Taylor's theorem is applied in Section~\ref{Taylorth}, whereas a corresponding
Lagrange remainder is bounded using the assumption on constant~$C$.

\section{Modifications of partition classes}
\label{pclmod}

We properly modify the underlying partition classes in order
to further upper-bound their cardinalities by the logarithmic function so that
the assumption concerning almost $\lceil(C+2)\log n\rceil$-wise independence
of ${\cal A}$ can be applied in the following Section~\ref{bonikwind}. Thus, we
confine ourselves to at most logarithmic-size arbitrary subsets $R'_j$ of
partition classes $R_j$, that is
\begin{equation}
\label{dfRjp}
R'_j\left\{
\begin{array}{ll}
=R_j&\quad\mbox{if }|R_j|\leq\log n\\
\subset R_j\mbox{ so that }|R'_j|=\lfloor\log n\rfloor&\quad\mbox{otherwise}\,,
\end{array}
\right.
\end{equation}
which ensures $R'_j\subseteq R_j$ and $|R'_j|\leq\log n$ for every $j=1,\ldots,r$.
For these new classes, assumption (\ref{acond}) can be rewritten as
\begin{eqnarray}
\prod_{j=1}^{r}\left(1-\frac{1}{2^{|R'_j|}}\right)&>&
\left(1-\frac{1}{2^{\log n}}\right)^{\frac{n}{\log n}}
\prod_{|R_j|\leq\log n}\left(1-\frac{1}{2^{|R_j|}}\right)
\nonumber\\
\label{acondRjp}
&>&\left(1-\frac{1}{n}\cdot\frac{n}{\log n}\right)\varepsilon
=\left(1-\frac{1}{\log n}\right)\varepsilon=\varepsilon'\,,
\end{eqnarray}
where $\varepsilon'>0$ is arbitrarily close to $\varepsilon$ for sufficiently
large $n$.

\medskip
Denote by $\{s_1,s_2,\ldots,s_m\}=\{|R'_1|,\ldots,|R'_r|\}$ the set of all cardinalities $1\leq s_i\leq\log n$ of classes $R'_1,\ldots,R'_r$,
and for every $i=1,\ldots,m$, let $r_i=|\{j\,|\,|R'_j|=s_i\}|$ be the
number of classes $R'_j$ having cardinality $s_i$, that is,
$r=\sum_{i=1}^m r_i$. Furthermore, we define
\begin{equation}
\label{dfti}
t_i=\frac{r_i}{2^{s_i}}>0\quad\mbox{for }i=1,\ldots,m\,.
\end{equation}
It follows from (\ref{acondRjp}) and (\ref{dfti}) that
\begin{eqnarray}
0<\varepsilon'&<&\prod_{j=1}^{r}\left(1-\frac{1}{2^{|R'_j|}}\right)=
\prod_{i=1}^{m}\left(1-\frac{1}{2^{s_i}}\right)^{r_i}\nonumber\\
\label{prodeps}
&=&\prod_{i=1}^{m}\left(\left(1-\frac{1}{2^{s_i}}\right)^{2^{s_i}}\right)^{t_i}
<e^{-\sum_{i=1}^{m} t_i}
\end{eqnarray}
implying
\begin{equation}
\label{sumtieps}
\sum_{i=1}^{m} t_i<\ln\frac{1}{\varepsilon'}\,.
\end{equation}

\medskip
Moreover, we define parameters
\begin{eqnarray}
\label{dfrho}
\varrho&=&\frac{C}{1-\left(1-\frac{\varepsilon'^2}
{4(1+\varepsilon'^2)}\right)^\frac{1}{C}}>C\geq 1\,,\\
\label{dfsigma}
\sigma&=&\log\left(\frac{4\varrho\,(1+\varepsilon'^2)}{\varepsilon'^2}\right)
\end{eqnarray}
which are bounded since $\varepsilon'$ is arbitrarily close to $\varepsilon$ for sufficiently large $n$. The parameter $\varrho$ is used to distinguish between the frequent cardinalities of classes $R_1',\ldots,R_r'$ and the rare ones, while $\sigma$ represents an upper bound on these cardinalities. Thus, the cardinalities $s_1,\ldots,s_m$ are sorted so that
\begin{eqnarray}
\label{podmsiC}
r_i>\varrho\mbox{ and }s_i>\sigma&\mbox{ for }&i=1,\ldots,m''\\
\label{podmsiCb}
r_i\leq\varrho\mbox{ and }s_i>\sigma&\mbox{ for }&i=m''+1,\ldots,m'\\
\label{podmsiCc}
s_i\leq\sigma&\mbox{ for }&i=m'+1,\ldots,m\,.
\end{eqnarray}
We will further confine ourselves to the first $m'\geq 0$
cardinalities satisfying $s_i>\sigma$ for $i=1,\ldots,m'$. Without loss
of generality, we can also sort the corresponding partition classes so that
\begin{eqnarray}
\label{podmsiCr}
|R'_j|>\sigma&\quad\mbox{for }&j=1,\ldots,r'\\
|R'_j|\leq\sigma&\quad\mbox{for }&j=r'+1,\ldots,r\,,
\end{eqnarray}
which implies
\begin{equation}
\label{lbr'}
r'=\sum_{i=1}^{m'} r_i=\sum_{i=1}^{m'}t_i2^{s_i}>
\frac{4\varrho\,(1+\varepsilon'^2)}{\varepsilon'^2}\,\sum_{i=1}^{m'}t_i
\end{equation}
according to (\ref{dfti}), (\ref{podmsiC})--(\ref{podmsiCb}), and (\ref{dfsigma}).
We include the remaining bounded-size classes $R'_j$ for $j=r'+1,\ldots,r$ into
$Q$, that is,
\begin{equation}
\label{dfQp}
Q'=Q\cup\bigcup_{j=r'+1}^rR'_j
\end{equation}
whose size can be upper-bounded as
\begin{equation}
\label{cardQp}
|Q'|\leq\log n+\sum_{i=m'+1}^m r_i\log\left(\frac{4\varrho\,(1+\varepsilon'^2)}
{\varepsilon'^2}\right)<2 \log n
\end{equation}
for sufficiently large $n$, since
\begin{equation}
\sum_{i=m'+1}^m r_i=\sum_{i=m'+1}^m t_i2^{s_i}<
\frac{4\varrho\,(1+\varepsilon'^2)}{\varepsilon'^2}\,
\ln\frac{1}{\varepsilon'}
\end{equation}
according to (\ref{dfti}), (\ref{sumtieps}), (\ref{podmsiCc}), and (\ref{dfsigma}).
\eject
\noindent This completes the definition of new classes $Q',R'_1,$ $\ldots,R'_{r'}$.
In addition, we define $\mathbf{c}'\in\{0,1\}^n$ that differs from~$\mathbf{c}$ exactly on a bounded number of bit locations from $R'_{r'+1},\ldots,R'_r$, e.g.
\begin{equation}
\label{dfcip}
c_i'=\left\{
\begin{array}{ll}
1-c_i&\quad\mbox{if }i\in\bigcup_{j=r'+1}^r R'_j\\
c_i&\quad\mbox{otherwise.}
\end{array}
\right.
\end{equation}
The modified $Q',R'_1,\ldots,R'_{r'}$ and $\mathbf{c}'$ are used in the following
lemma for lower-bounding the probability (\ref{probability}).
\begin{lemma}
\label{lbmod}
\begin{equation}
\label{prP2}
p\geq\frac{\left|{\cal A}_n^{Q'}(\mathbf{c'})\setminus\bigcup_{j=1}^{r'}
{\cal A}_n^{R'_j}(\mathbf{c'})\right|}{|{\cal A}_n|}=
\frac{\left|{\cal A}_n^{Q'}(\mathbf{c'})\right|}{|{\cal A}_n|}
-\frac{\left|\bigcup_{j=1}^{r'}{\cal A}_n^{R'_j\cup Q'}(\mathbf{c'})\right|}{|{\cal A}_n|}\,.
\end{equation}
\end{lemma}
\begin{proof}
For verifying the lower bound in (\ref{prP2}) it suffices to show that
\begin{equation}
{\cal A}_n^{Q'}(\mathbf{c'})\setminus\bigcup_{j=1}^{r'}
{\cal A}_n^{R'_j}(\mathbf{c'})\subseteq{\cal A}_n^Q(\mathbf{c})\setminus
\bigcup_{j=1}^r{\cal A}_n^{R_j}(\mathbf{c})
\end{equation}
according to (\ref{probability}).
Assume $\mathbf{a}\in{\cal A}_n^{Q'}(\mathbf{c'})\setminus\bigcup_{j=1}^{r'}
{\cal A}_n^{R'_j}(\mathbf{c'})$, which means
$\mathbf{a}\in{\cal A}_n^{Q'}(\mathbf{c'})\subseteq{\cal A}_n^{Q}(\mathbf{c'})
={\cal A}_n^{Q}(\mathbf{c})$ and
$\mathbf{a}\not\in{\cal A}_n^{R'_j}(\mathbf{c'})=
{\cal A}_n^{R'_j}(\mathbf{c})\supseteq{\cal A}_n^{R_j}(\mathbf{c})$ for every
$j=1,\ldots,r'$ by definitions (\ref{dfRjp}), (\ref{dfQp}), (\ref{dfcip}), and
the fact that $S_1\subseteq S_2$ implies
${\cal A}_n^{S_2}(\mathbf{c})\subseteq{\cal A}_n^{S_1}(\mathbf{c})$. In addition,
$\mathbf{a}\in{\cal A}_n^{Q'}(\mathbf{c'})$ implies
$\mathbf{a}\not\in{\cal A}_n^{R_j}(\mathbf{c})$ for every
$j=r'+1,\ldots,r$ according to (\ref{dfcip}), and hence, $\mathbf{a}\in
{\cal A}_n^Q(\mathbf{c})\setminus\bigcup_{j=1}^r{\cal A}_n^{R_j}(\mathbf{c})$.
This completes the proof of the lower bound, while the equality in (\ref{prP2})
follows from ${\cal A}_n^{R'_j\cup Q'}(\mathbf{c'})
\subseteq{\cal A}_n^{Q'}(\mathbf{c'})$ for every $j=1,\ldots,r'$.
\end{proof}

\section{Almost $k$-wise independence}
\label{bonikwind}

Furthermore, we will upper-bound the probability of the finite
union of events appearing in formula (\ref{prP2}) by using Bonferroni inequality for constant number $C'=\min(C,r')$ of terms, which gives
\begin{eqnarray}
\label{prP3a}
p&\geq&\frac{\left|{\cal A}_n^{Q'}(\mathbf{c'})\right|}{|{\cal A}_n|}
-\sum_{k=1}^{C'}(-1)^{k+1}\sum_{1\leq j_1<j_2<\cdots<j_k\leq r'}
\frac{\left|\bigcap_{i=1}^k{\cal A}_n^{R'_{j_i}\cup Q'}(\mathbf{c'})\right|}{|{\cal A}_n|}\\
\label{prP3}
&=&\sum_{k=0}^{C'}(-1)^{k}\sum_{1\leq j_1<j_2<\cdots<j_k\leq r'}
\frac{\left|{\cal A}_n^{\bigcup_{i=1}^k R'_{j_i}\cup Q'}(\mathbf{c'})\right|}{|{\cal A}_n|}
\end{eqnarray}
according to Lemma~\ref{lbmod}. For notational simplicity, the inner sum in (\ref{prP3})
over $1\leq j_1<j_2<\cdots<j_k\leq r'$ for $k=0$
reads formally as it includes one summand $|{\cal A}_n^{Q'}(\mathbf{c'})|/|{\cal A}_n|$.
Note that $C'$ is odd for $C<r'$, while equality holds in~(\ref{prP3a}) for $C'=r'$,
which is the probabilistic inclusion-exclusion principle. For any $0\leq k\leq C'\leq C$,
we know $\left|\bigcup_{i=1}^k R'_{j_i}\cup Q'\right|\leq\lceil(C+2)\log n\rceil$
according to (\ref{dfRjp}) and (\ref{cardQp}), and hence,
\begin{equation}
\label{kwiseind}
\frac{\left|{\cal A}_n^{\bigcup_{i=1}^k R'_{j_i}\cup Q'}(\mathbf{c'})\right|}{|{\cal A}_n|}
\geq\frac{1}{2^{|Q'|+\sum_{i=1}^k\left|R'_{j_i}\right|}}-\beta
=\frac{1}{2^{|Q'|}}\prod_{i=1}^k\frac{1}{2^{\left|R'_{j_i}\right|}}-\beta
\end{equation}
(where the product in (\ref{kwiseind}) equals formally 1 for $k=0$) and similarly,
\begin{equation}
-\frac{\left|{\cal A}_n^{\bigcup_{i=1}^k R'_{j_i}\cup Q'}(\mathbf{c'})\right|}
{|{\cal A}_n|}
\geq-\frac{1}{2^{|Q'|}}\prod_{i=1}^k\frac{1}{2^{\left|R'_{j_i}\right|}}-\beta
\end{equation}
according to (\ref{almkwind}) since ${\cal A}$ is $(\lceil(C+2)\log n\rceil,\beta)$-wise
independent. We plug these inequalities into (\ref{prP3}), which leads to
\begin{eqnarray}
p&\geq&\sum_{k=0}^{C'}(-1)^{k}\sum_{1\leq j_1<j_2<\cdots<j_k\leq r'}\enspace
\frac{1}{2^{|Q'|}}\prod_{i=1}^k\frac{1}{2^{\left|R'_{j_i}\right|}}
-\beta\sum_{k=0}^{C'}\binom{r'}{k}\nonumber\\
\label{prP4}
&\geq&\frac{1}{2^{|Q'|}}\left(\sum_{k=0}^{C'}(-1)^{k}
\sum_{1\leq j_1<j_2<\cdots<j_k\leq r'}\enspace
\prod_{i=1}^k\frac{1}{2^{\left|R'_{j_i}\right|}}
-\beta\,2^{|Q'|}\,(r'+1)^{C'}\right)\,,\qquad
\end{eqnarray}
where
\begin{equation}
\label{kwer}
\beta\,2^{|Q'|}\,(r'+1)^{C'}<\frac{1}{n^{C+3}}\,n^2\,n^C=\frac{1}{n}
<\frac{\varepsilon'}{8}
\end{equation}
for sufficiently large $n>8/\varepsilon'$ by using the assumption on $\beta$,
inequality (\ref{cardQp}), $r'<n$ (e.g., $r'=n$ would break  (\ref{podmsiCr})
and (\ref{dfsigma})), and $C'\leq C$. The following lemma rewrites the inner
sum in formula (\ref{prP4}).
\begin{lemma}
\label{l10}
For $0\leq k\leq C'$,
\begin{equation}
\label{isprP4}
\sum_{1\leq j_1<j_2<\cdots<j_k\leq {r'}}\enspace
\prod_{i=1}^k\frac{1}{2^{|R'_{j_i}|}}=
\sum_{\substack{k_1+\cdots+k_{m'}=k\\0\leq k_1\leq r_1,\ldots,0\leq k_{m'}\leq r_{m'}}}\enspace
\prod_{i=1}^{m'}\frac{t_i^{k_i}}{k_i!}\,\prod_{j=1}^{k_i-1}\left(1-\frac{j}{r_i}\right)\,.
\end{equation}
\end{lemma}
\begin{proof}
By grouping the classes of the same cardinality together, the left-hand side of inequality (\ref{isprP4}) can be rewritten as
\begin{equation}
\sum_{1\leq j_1<j_2<\cdots<j_k\leq {r'}}\enspace
\prod_{i=1}^k\frac{1}{2^{|R'_{j_i}|}}=
\sum_{\substack{k_1+k_2+\cdots+k_{m'}=k\\0\leq k_1\leq r_1,\ldots,0\leq k_{m'}\leq r_{m'}}}\enspace
\prod_{i=1}^{m'}\binom{r_i}{k_i}\left(\frac{1}{2^{s_i}}\right)^{k_i}\,,
\end{equation}
where $k_1,\ldots,k_{m'}$ denote the numbers of classes of corresponding cardinalities
$s_1,\ldots,s_{m'}$ considered in a current summand, and
\begin{equation}
\binom{r_i}{k_i}\left(\frac{1}{2^{s_i}}\right)^{k_i}=
\frac{r_i\left(r_i-1\right)\cdots\left(r_i-k_i+1\right)}{k_i!}
\left(\frac{t_i}{r_i}\right)^{k_i}
=\frac{t_i^{k_i}}{k_i!}\,\prod_{j=1}^{k_i-1}\left(1-\frac{j}{r_i}\right)
\end{equation}
according to (\ref{dfti}).
\end{proof}
Thus, we plug equations (\ref{kwer}) and (\ref{isprP4}) into (\ref{prP4}) and obtain
\begin{equation}
\label{prP5}
p>\frac{1}{n^2}\left(\sum_{k=0}^{C'}(-1)^{k}
\sum_{\substack{k_1+\cdots+k_{m'}=k\\0\leq k_1\leq r_1,\ldots,0\leq k_{m'}\leq r_{m'}}}\enspace
\prod_{i=1}^{m'}\frac{t_i^{k_i}}{k_i!}\,\prod_{j=1}^{k_i-1}\left(1-\frac{j}{r_i}\right)
-\frac{\varepsilon'}{8}\right)\,.
\end{equation}
Note that for $m'=0$ (implying $r'=C'=0$), the inner sum in (\ref{prP5}) equals 1.

\section{Frequent cardinalities}
\label{freqcard}

We sort out the terms with frequent cardinalities (\ref{podmsiC}) from
the sum in formula~(\ref{prP5}), that is,
\begin{equation}
\label{prP5aa}
p>\frac{1}{n^2}\left(\sum_{k=0}^{C'}(-1)^{k}
\sum_{\substack{k_1+\cdots+k_{m''}=k\\0\leq k_1\leq r_1,\ldots,0\leq k_{m''}\leq r_{m''}}}\enspace
\prod_{i=1}^{m''}\frac{t_i^{k_i}}{k_i!}\,\prod_{j=1}^{k_i-1}\left(1-\frac{j}{r_i}\right)
-T_1-\frac{\varepsilon'}{8}\right),
\end{equation}
where the inner sum in (\ref{prP5aa}) equals zero for $k>r''=\sum_{i=1}^{m''}r_i\,$, and
\begin{equation}
\label{dft1}
T_1=\sum_{k=0}^{C'}(-1)^{k+1}
\sum_{\substack{k_1+\cdots+k_{m'}=k\\
0\leq k_1\leq r_1,\ldots,0\leq k_{m'}\leq r_{m'}\\
(\exists\, m''+1\leq\ell\leq m')\,k_\ell>0}}\enspace
\prod_{i=1}^{m'}\frac{t_i^{k_i}}{k_i!}\,
\prod_{j=1}^{k_i-1}\left(1-\frac{j}{r_i}\right)
\end{equation}
sums up the terms including rare cardinalities (\ref{podmsiCb}). In addition, we know
\begin{eqnarray}
\label{produb}
1\geq\prod_{i=1}^{m''}\,\prod_{j=1}^{k_i-1}\left(1-\frac{j}{r_i}\right)&\geq&\\
\label{prodlb}
\prod_{i=1}^{m''}\,\left(1-\frac{C-1}{\varrho}\right)^{k_i-1}
&>&\left(1-\frac{C}{\varrho}\right)^{C}
=1-\frac{\varepsilon'^2}{4(1+\varepsilon'^2)}
\end{eqnarray}
according to (\ref{podmsiC}), (\ref{dfrho}), and
$k_i\leq k=\sum_{i=1}^{m''}k_i\leq C'\leq C<\varrho\,$.
The upper bound (\ref{produb}) and lower bound (\ref{prodlb}) on the underlying
product are used to lower-bound the negative terms of (\ref{prP5aa})
for odd $k$ and the positive ones for even $k$, respectively, that is,
\begin{equation}
\label{prP5ab}
p>\frac{1}{n^2}\left(\sum_{k=0}^{C'}(-1)^{k}
\sum_{\substack{k_1+\cdots+k_{m''}=k\\0\leq k_1\leq r_1,\ldots,0\leq k_{m''}\leq r_{m''}}}\enspace
\prod_{i=1}^{m''}\frac{t_i^{k_i}}{k_i!}-\frac{\varepsilon'^2}{4(1+\varepsilon'^2)}\,T_2
-T_1-\frac{\varepsilon'}{8}\right)
\end{equation}
where
\begin{equation}
\label{prP5b}
T_2=\sum_{k=0,2,4,\ldots}^{C'}\enspace
\sum_{\substack{k_1+\cdots+k_{m''}=k\\0\leq k_1\leq r_1,\ldots,0\leq k_{m''}\leq r_{m''}}}\enspace
\prod_{i=1}^{m''}\frac{t_i^{k_i}}{k_i!}\,.
\end{equation}
The following lemma upper-bounds the above-introduced terms $T_1$ and $T_2$.
\begin{lemma}~
\begin{enumerate}
\renewcommand\labelenumi{\theenumi}
\renewcommand{\theenumi}{(\roman{enumi})}
\label{rarec}
\item
$T_1<\frac{\varepsilon'}{8}\,$.
\item
$T_2<\frac{1+\varepsilon'^2}{2\,\varepsilon'}\,$.
\end{enumerate}
\end{lemma}
\begin{proof}

\vspace*{-6mm}
\begin{enumerate}
\renewcommand\labelenumi{\theenumi}
\renewcommand{\theenumi}{(\roman{enumi})}
\item
We can only take the terms of (\ref{dft1}) for odd $k=1,3,5,\ldots$ into account
since those for even $k$ are nonpositive (e.g.\ the term for $k=0$ equals zero
because there is no $m''+1\leq\ell\leq m'$ such that $k_\ell>0$ in this case).
Thus,
\begin{eqnarray}
T_1&\leq&\sum_{k=1,3,5,\ldots}^{C'}\enspace
\sum_{\substack{k_1+\cdots+k_{m'}=k\\
0\leq k_1\leq r_1,\ldots,0\leq k_{m'}\leq r_{m'}\\
(\exists\, m''+1\leq\ell\leq m')\,k_\ell>0}}
\frac{r_\ell}{2^{s_\ell}}\,\frac{1}{k_\ell}\,
\frac{t_\ell^{k_\ell-1}}{(k_\ell-1)!}\,
\prod_{\substack{i=1\\i\not=\ell}}^{m'}\frac{t_i^{k_i}}{k_i!}\nonumber\\
\label{srlb}
&\leq&\frac{\varrho}{2^\sigma}\sum_{k=1,3,5,\ldots}^{C'}\enspace
\sum_{\substack{k_1+\cdots+k_{m'}=k\\
0\leq k_1\leq r_1,\ldots,0\leq k_{m'}\leq r_{m'}\\
(\exists\, m''+1\leq\ell\leq m')\,k_\ell>0}}
\frac{t_\ell^{k_\ell-1}}{(k_\ell-1)!}\,
\prod_{\substack{i=1\\i\not=\ell}}^{m'}\frac{t_i^{k_i}}{k_i!}
\end{eqnarray}
according to (\ref{dfti}) and (\ref{podmsiCb}). Formula (\ref{srlb}) is
rewritten by replacing indices $k_\ell-1$ and $k-1$ with $k_\ell$ and $k$,
respectively, which is further upper-bounded by removing the upper bounds
that are set on indices $k_1,\ldots,k_{m'}$ and by omitting the condition
concerning the existence of special index $\ell$, as follows:
\begin{equation}
\label{srlb2}
T_1\leq\frac{\varrho}{2^\sigma}\sum_{k=0,2,4,\ldots}^{C'-1}\enspace
\sum_{\substack{k_1+\cdots+k_{m'}=k\\
k_1\geq 0,\ldots,k_{m'}\geq 0}}\enspace
\prod_{i=1}^{m'}\frac{t_i^{k_i}}{k_i!}
=\frac{\varrho}{2^\sigma}\sum_{k=0,2,4,\ldots}^{C'-1}
\frac{\left(\sum_{i=1}^{m'}t_i\right)^k}{k!}\,,
\end{equation}
where the multinomial theorem is employed. Notice that the sum on the right-hand
side of equation (\ref{srlb2}) represents the first few terms of $\,$Taylor series
of the hyperbolic$\,$ cosine$\,$ at point
  $\sum_{i=1}^{m'}t_i\geq 0$, which implies
\begin{equation}
\label{Taylcods}
T_1<\frac{\varrho}{2^\sigma}\,\cosh\left(\sum_{i=1}^{m'}t_i\right)
<\frac{\varepsilon'^2}{4(1+\varepsilon'^2)}\,\cdot\,
\frac{\frac{1}{\varepsilon'}+\varepsilon'}{2}=\frac{\varepsilon'}{8}
\end{equation}
according to  (\ref{sumtieps}) and (\ref{dfsigma}) since the hyperbolic cosine
is an increasing function for nonnegative arguments.

\item
Similarly as in the proof of (i), we apply the multinomial theorem
(cf.\ (\ref{srlb2})) and the Taylor series of the hyperbolic cosine (cf.\
(\ref{Taylcods})) to (\ref{prP5b}), which gives
\begin{eqnarray}
T_2&\leq&\sum_{k=0,2,4,\ldots}^{C'}\enspace
\sum_{\substack{k_1+\cdots+k_{m''}=k\\k_1\geq 0,\ldots,k_{m''}\geq 0}}\enspace
\prod_{i=1}^{m''}\frac{t_i^{k_i}}{k_i!}
=\sum_{k=0,2,4,\ldots}^{C'}\frac{\left(\sum_{i=1}^{m''}t_i\right)^k}{k!}\\
&\leq&\cosh\left(\sum_{i=1}^{m''}t_i\right)<
\frac{1+\varepsilon'^2}{2\,\varepsilon'}\,.
\end{eqnarray}
\end{enumerate}
\end{proof}

\noindent We plug the bounds from Lemma~\ref{rarec} into (\ref{prP5ab}) and obtain
\begin{equation}
\label{prP5a}
p>\frac{1}{n^2}\left(\sum_{k=0}^{C'}(-1)^{k}
\sum_{\substack{k_1+\cdots+k_{m''}=k\\0\leq k_1\leq r_1,\ldots,0\leq k_{m''}\leq r_{m''}}}\enspace
\prod_{i=1}^{m''}\frac{t_i^{k_i}}{k_i!}
-\frac{3\,\varepsilon'}{8}\right)\,.
\end{equation}

\section{Taylor's theorem}
\label{Taylorth}

In order to apply the multinomial theorem again, we remove the upper bounds
that are set on indices in the inner sum of formula (\ref{prP5a}), that is,
\begin{equation}
\label{prP6}
p>\frac{1}{n^2}\left(\sum_{k=0}^{C'}(-1)^{k}
\sum_{\substack{k_1+\cdots+k_{m''}=k\\k_1\geq 0,\ldots,k_{m''}\geq 0}}\enspace
\prod_{i=1}^{m''}\frac{t_i^{k_i}}{k_i!}-T-\frac{3\,\varepsilon'}{8}\right)\,,
\end{equation}
which is corrected by introducing additional term
\begin{equation}
\label{dfT}
T=\sum_{k=0}^{C'}(-1)^{k}
\sum_{\substack{k_1+\cdots+k_{m''}=k\\k_1\geq 0,\ldots,k_{m''}\geq 0\\
(\exists 1\leq\ell\leq m'')\,k_\ell>r_\ell}}\enspace
\prod_{i=1}^{m''}\frac{t_i^{k_i}}{k_i!}\,.
\end{equation}
Thus, inequality (\ref{prP6}) can be further rewritten as
\begin{eqnarray}
\label{prP7a}
p&>&\frac{1}{n^2}\left(\sum_{k=0}^{C'}
\frac{\left(-\sum_{i=1}^{m''}t_i\right)^k}{k!}
-T-\frac{3\,\varepsilon'}{8}\right)\\
\label{prP7}
&=&\frac{1}{n^2}\left(e^{-\sum_{i=1}^{m''} t_i}
-{\cal R}_{C'+1}\left(-\sum_{i=1}^{m''} t_i\right)
-T-\frac{3\,\varepsilon'}{8}\right)\,,
\end{eqnarray}
where Taylor's theorem is employed for the exponential function at point $-\sum_{i=1}^{m''} t_i$ producing the Lagrange remainder
\begin{equation}
\label{Lr}
{\cal R}_{C'+1}\left(-\sum_{i=1}^{m''} t_i\right)=
\frac{\left(-\sum_{i=1}^{m''}t_i\right)^{C'+1}}{(C'+1)!}\,
e^{-\vartheta\sum_{i=1}^{m''} t_i}
<\left(\frac{\sum_{i=1}^{m''}t_i}{\sqrt{C'}}\right)^{C'+1}
\end{equation}
with parameter $0<\vartheta<1$. Note that the upper bound in (\ref{Lr})
assumes $C'>0$, whereas for $C'=r'=0$ implying
$m''=m'=0$, we know ${\cal R}_1(0)=0$.
This remainder and term $T$ are upper-bounded in the following lemma.
\begin{lemma}~
\label{Lrb}
\begin{enumerate}
\renewcommand\labelenumi{\theenumi}
\renewcommand{\theenumi}{(\roman{enumi})}
\item
$T<\frac{\varepsilon'}{8}\,$.
\item
${\cal R}_{C'+1}\left(-\sum_{i=1}^{m''} t_i\right)<\frac{\varepsilon'}{4}\,$.
\end{enumerate}
\end{lemma}
\begin{proof}

\vspace*{-6mm}
\begin{enumerate}
\renewcommand\labelenumi{\theenumi}
\renewcommand{\theenumi}{(\roman{enumi})}
\item
We take only the summands of (\ref{dfT}) for even $k\geq 2$ into account since
the summands for odd $k$ are not positive, while for $k=0$ there is no
$1\leq\ell\leq m''$ such that $0=k\geq k_\ell>r_\ell\geq 1$, which gives
\begin{equation}
T\leq\sum_{k=2,4,6,\ldots}^{C'}\enspace
\sum_{\substack{k_1+\cdots+k_{m''}=k\\k_1\geq 0,\ldots,k_{m''}\geq 0\\
(\exists 1\leq\ell\leq m'')\,k_\ell>r_\ell}}
\frac{1}{2^{s_\ell}}\,\frac{r_\ell}{k_\ell}\,\frac{t_\ell^{k_\ell-1}}{(k_\ell-1)!}\enspace
\prod_{\substack{i=1\\i\not=\ell}}^{m''}\frac{t_i^{k_i}}{k_i!}\nonumber
\end{equation}
\begin{equation}
\label{ubT}
\leq\frac{1}{2^\sigma}
\sum_{k=2,4,6,\ldots}^{C'}\enspace
\sum_{\substack{k_1+\cdots+k_{m''}=k\\k_1\geq 0,\ldots,k_{m''}\geq 0\\
(\exists 1\leq\ell\leq m'')\,k_\ell>r_\ell}}\frac{t_\ell^{k_\ell-1}}{(k_\ell-1)!}\enspace
\prod_{\substack{i=1\\i\not=\ell}}^{m''}\frac{t_i^{k_i}}{k_i!}
\end{equation}
using (\ref{dfti}) and (\ref{podmsiC}). Formula (\ref{ubT}) is rewritten by replacing
indices $k_\ell-1$ and $k-1$ with $k_\ell$ and $k$, respectively, which is further upper
bounded by omitting the condition concerning the existence of special index $\ell$,
as follows:
\begin{equation}
\label{ubT2}
T\leq\frac{1}{2^\sigma}
\sum_{k=1,3,5,\ldots}^{C'-1}\enspace
\sum_{\substack{k_1+\cdots+k_{m''}=k\\k_1\geq 0,\ldots,k_{m''}\geq 0}}\enspace
\prod_{i=1}^{m''}\frac{t_i^{k_i}}{k_i!}
=\frac{1}{2^\sigma}\sum_{k=1,3,5,\ldots}^{C'-1}
\frac{\left(\sum_{i=1}^{m''}t_i\right)^k}{k!}\,,
\end{equation}
where the multinomial theorem is employed. Notice that the sum on the right-hand
side of equation (\ref{ubT2}) represents the first few terms of Taylor series
of the hyperbolic sine at point $\sum_{i=1}^{m''}t_i$, which implies
\begin{equation}
T\leq\frac{1}{2^\sigma}\,\sinh\left(\sum_{i=1}^{m''}t_i\right)<
\frac{\varepsilon'^2}{4\varrho\,(1+\varepsilon'^2)}\,\cdot\,
\frac{\frac{1}{\varepsilon'}-\varepsilon'}{2}<\frac{\varepsilon'}{8}
\end{equation}
according to (\ref{sumtieps}) and (\ref{dfsigma}) since the hyperbolic sine
is an increasing function.

\item
For $C'=C\geq 1$, Lagrange remainder (\ref{Lr}) can further be upper-bounded as
\begin{equation}
{\cal R}_{C'+1}\left(-\sum_{i=1}^{m''} t_i\right)
<\left(\frac{\ln\frac{1}{\varepsilon'}}{\sqrt{C}}\right)^{C+1}
<\left(\frac{\varepsilon'}{2}\right)^{C+1}<\frac{\varepsilon'}{4}
\end{equation}
for sufficiently large $n$ by using (\ref{sumtieps}) and
the definition of $C$, while for $C'=r'<C$, the underlying upper bound
\begin{equation}
{\cal R}_{C'+1}\left(-\sum_{i=1}^{m''} t_i\right)
\leq\left(\frac{\sum_{i=1}^{m'}t_i}
{\frac{4\varrho\,(1+\varepsilon'^2)}{\varepsilon'^2}}\right)^{\frac{r'+1}{2}}
<\frac{\ln\frac{1}{\varepsilon'}}{\frac{4\varrho\,(1+\varepsilon'^2)}{\varepsilon'^2}}
<\frac{\varepsilon'}{4}
\end{equation}
can be obtained from (\ref{lbr'}) and (\ref{sumtieps}).
\end{enumerate}

\vspace*{-7mm}
\end{proof}

Finally, inequality (\ref{prodeps}) together with the upper bounds
from Lemma~\ref{Lrb} are plugged into (\ref{prP7}), which leads to
\begin{equation}
\label{prPf}
p>\frac{\varepsilon'}{4n^2}=\frac{\varepsilon}{4n^2}\left(1-\frac{1}{\log n}\right)>0
\end{equation}
according to (\ref{acondRjp}).
Thus, we have proven that for any
$\mathbf{c}\in\{0,1\}^n$ the probability that there is $\mathbf{a}\in {\cal A}_n$
satisfying the conjunction (\ref{cond}) for $Q$ and partition $\{R_1,\ldots,R_r\}$
is strictly positive, which means such $\mathbf{a}$ does exist. This completes
the proof that ${\cal A}$ is $\varepsilon$-rich.
\end{proof}

\section{Conclusion}
\label{concl}

In the present paper, we have made an important step in the effort of
constructing hitting set generators for the model of read-once branching
programs of bounded width. Such constructions have so far been known only
in the case of width 2 and in very restricted cases of bounded width (e.g.\
regular oblivious read-once branching programs). We have now
provided an explicit polynomial-time construction of a hitting set for
read-once branching programs of width 3 with acceptance probability
$\varepsilon>\frac{5}{6}$. Although this model seems to be
relatively weak, the presented proof is far from being trivial. In particular,
we have formulated a so-called richness condition which is independent of
the notion of branching programs. This condition characterizes the hitting
sets for read-once branching programs of width~3. We have shown that such
a hitting set hits read-once conjunctions of DNF and CNF, which corresponds
to the weak richness condition. On the other hand, the richness condition
proves to be sufficient for a set extended with all strings within Hamming
distance of 3 to be a hitting set for width-3 1-branching programs.
In addition, we have proven for a suitable constant $C$ that any almost
$(C\log n)$-wise independent set which can be constructed in polynomial
time due to Alon et al.~\cite{Alon92}, satisfies this richness condition,
which implies our result. It also follows that almost $O(\log n)$-wise
independent sets are hitting sets for read-once conjunctions of DNF and CNF.

From the point of view of derandomization of unrestricted models, our result
still appears to be unsatisfactory but it is the best we know so far.
The issue of whether our technique based on the richness condition can be
extended to the case of width 4 or to bounded width represents an open problem
for further research. Another challenge for improving our result is to optimize
parameter $\varepsilon$, e.g.\ to achieve the result for
$\varepsilon\leq\frac{1}{n}$, which would be important for practical
derandomizations.

\subsection*{Acknowledgement}

Ji\v{r}\'{\i} \v{S}\'{\i}ma  was partially supported by the institutional support RVO: 67985807 and by the grant of the Czech Science Foundation GBP202/12/G061.

\smallskip\noindent
Stanislav \v{Z}\'{a}k  was partially supported by the institutional support RVO: 67985807 and by the grant of the Czech Science Foundation GAP202/10/1333.

\smallskip
The authors would like to thank Pavel Pudl\'ak for pointing out the problem
of hitting sets for width-3 read-once branching programs. The presentation of this paper benefited from valuable
suggestions of anonymous reviewers.

\end{document}